\documentclass[12pt]{article}   
\newcommand{\anon}{1}

\usepackage{graphicx} 
\usepackage{color}

\usepackage[hidelinks]{hyperref}
\usepackage{authblk}
\usepackage{caption}
\usepackage{subcaption}
\usepackage{natbib}
\usepackage{amsmath}
\usepackage{amssymb,amsthm}
\usepackage{microtype}
\usepackage{multirow}
\usepackage{setspace}
\usepackage{float}
\usepackage{enumerate}
\usepackage[group-separator={,}]{siunitx}

\newcommand{\ind}{\overset{\text{ind.}}{\sim}}

\newtheorem{proposition}{Proposition}

\newtheorem{algorithm}{Algorithm}

\newtheorem{corollary}{Corollary}

\newtheorem{remark}{Remark}
\newtheorem{lemma}{Lemma}

\usepackage[prependcaption,colorinlistoftodos,textsize=small]{todonotes}

\usepackage{setspace}

\newcommand{\Atr}{A^{(\textnormal{tr})}}
\newcommand{\atr}{a^{(\textnormal{tr})}}

\newcommand{\Antr}{A_n^{(\textnormal{tr})}}
\newcommand{\antr}{a_n^{(\textnormal{tr})}}
\newcommand{\Antrij}{A_{n,ij}^{(\textnormal{tr})}}

\newcommand{\Ante}{A_n^{(\textnormal{te})}}
\newcommand{\Anteij}{A_{n,ij}^{(\textnormal{te})}}

\newcommand{\Ztr}{\hat{Z}^{(\textnormal{tr})}}
\newcommand{\ZtrT}{\hat{Z}^{(\textnormal{tr})\top}}

\newcommand{\Zntr}{\hat{Z}_n^{(\textnormal{tr})}}

\newcommand{\ZntrT}{\hat{Z}_n^{(\textnormal{tr})\top}}
\newcommand{\Zniktr}{\hat{Z}_{n, i k}^{(\textnormal{tr})}}
\newcommand{\Znjltr}{\hat{Z}_{n, j \ell}^{(\textnormal{tr})}}
\newcommand{\Ziktr}{\hat{Z}_{i k}^{(\textnormal{tr})}}
\newcommand{\Zjltr}{\hat{Z}_{j \ell}^{(\textnormal{tr})}}

\newcommand{\Ate}{A^{(\textnormal{te})}}
\newcommand{\Aijtr}{A_{ij}^{(\textnormal{tr})}}

\newcommand{\Aijte}{A_{ij}^{(\textnormal{te})}}

\newcommand{\Ic}{\mathcal{I}}

\newcommand{\Calpha}{\mathcal{C}^\alpha}

\newcommand{\nskl}{n^{(s)}_{k \ell}}
\newcommand{\nzkl}{n^{(0)}_{k \ell}}
\newcommand{\nokl}{n^{(1)}_{k \ell}}
\newcommand{\nkl}{n_{k \ell}}
\newcommand{\Bskl}{B^{(s)}_{k \ell}}
\newcommand{\Bzkl}{B^{(0)}_{k \ell}}
\newcommand{\Bokl}{B^{(1)}_{k \ell}}
\newcommand{\Ktr}{K^{\text{true}}}

\newcommand{\Ickl}{\mathcal{I}_{k \ell}}
\newcommand{\Iczkl}{\mathcal{I}^{(0)}_{k \ell}}
\newcommand{\Icokl}{\mathcal{I}^{(1)}_{k \ell}}
\newcommand{\Icskl}{\mathcal{I}^{(s)}_{k \ell}}

\newcommand{\Ickltil}{\tilde{\mathcal{I}}_{k \ell}}
\newcommand{\Iczkltil}{\tilde{\mathcal{I}}^{(0)}_{k \ell}}
\newcommand{\Icokltil}{\tilde{\mathcal{I}}^{(1)}_{k \ell}}
\newcommand{\Icskltil}{\tilde{\mathcal{I}}^{(s)}_{k \ell}}

\newcommand{\Icnskl}{\mathcal{I}^{(s)}_{n,k \ell}}
\newcommand{\Icnkl}{\mathcal{I}_{n, k \ell}}
\newcommand{\Icnzkl}{\mathcal{I}^{(0)}_{n, k \ell}}
\newcommand{\Icnokl}{\mathcal{I}^{(1)}_{n, k \ell}}

\newcommand{\Hkl}[1]{H_{kl}^{(#1)}}

\DeclareMathOperator{\expit}{expit}
\DeclareMathOperator{\vct}{vec}
\DeclareMathOperator{\diag}{diag}
\DeclareMathOperator{\Pois}{Poisson}
\DeclareMathOperator{\Unif}{Uniform}
\DeclareMathOperator{\Bin}{Binomial}
\DeclareMathOperator{\Bern}{Bernoulli}
\DeclareMathOperator*{\logit}{logit}
\DeclareMathOperator{\E}{E}
\DeclareMathOperator{\Var}{Var}

\DeclareMathOperator*{\argmax}{arg\,max}
\usepackage{authblk}

\title{Post-selection inference with a single realization of a network}
\author[1]{Ethan Ancell}
\author[1,2]{Daniela Witten}
\author[3,4]{Daniel Kessler}
\affil[1]{Department of Statistics, University of Washington}
\affil[2]{Department of Biostatistics, University of Washington}
\affil[3]{Department of Statistics and Operations Research, University of North Carolina at Chapel Hill}
\affil[4]{School of Data Science and Society, University of North Carolina at Chapel Hill}

\begin{document}

\maketitle

\begin{abstract}
  Given a dataset consisting of a single network, we consider inference on a
parameter selected from the data. We focus on the setting where the selected
parameter is a linear combination of the mean connectivities within and between
estimated communities. Inference in this setting poses a challenge, as the
communities are themselves estimated from the data. Furthermore, since only a
single realization of the network is available, sample splitting is not
possible. We show how to split a single realization of a network
with $n$ nodes into two (or more) networks involving the same $n$ nodes; the
first network can be used to select a data-driven parameter, and the second to
conduct inference on that parameter. In the case of weighted networks with
Poisson or Gaussian edges, we obtain two independent realizations of the
network; by contrast, in the case of Bernoulli edges, the two realizations are
dependent, and so extra care is required. We establish the theoretical
properties of our estimators, in the sense of confidence intervals that attain
the nominal (selective) coverage, and demonstrate their utility in numerical
simulations and in application to a dataset representing the relationships among
dolphins in Doubtful Sound, New Zealand.

\end{abstract}

\section{Introduction}
\label{sec:introduction}

A \textit{network} captures the pairwise relationships (called \textit{edges}) among
a set of \textit{nodes}. Networks arise in a plethora of application areas, including the
social \citep{omalley2008AnalysisSocialNetworks,
  snijders2011StatisticalModelsSocial} and biological
\citep{desilva2005ComplexNetworksSimple, liu2020ComputationalNetworkBiology}
sciences. In many settings, the edges (e.g., their presence, sign, associated weight, etc.) are
treated as random. A number of models for random networks have been well-studied
in the literature; examples include the exponential random graph
\citep{chatterjee2013EstimatingUnderstandingExponential}, the random dot product
graph \citep[RDPG,][]{young2007RandomDotProduct,
  athreya2018StatisticalInferenceRandom}, and the stochastic block model
\citep[SBM,][]{holland1983StochasticBlockmodelsFirst} along with its variants
\citep{airoldi2008MixedMembershipStochastic,karrer2011StochasticBlockmodelsCommunity,kao2019HybridMixedMembershipBlockmodel}.

This paper focuses on the setting where we have access to a single
realization of a network whose edges are random, and we wish to (i) use that
single realization to select a parameter of interest, and (ii) conduct
 inference on that selected parameter. For instance,
if we suspect the presence of latent community structure in a network, then
we might (i) estimate community membership among the nodes (where the parameter of interest is defined in terms of estimated communities as in Figure~\ref{fig:target_selection_cartoon}(a)), and (ii) conduct inference on the
expected connectivity within or between the estimated communities (Figure~\ref{fig:target_selection_cartoon}(b)). Critically, 
for step (ii) to yield valid inference, it must account for the fact that the communities were
estimated using the observed network.  In general, failure to account for
the data-dependent selection of the parameter in (i) leads to
statistical issues in (ii), including lack of type 1 error control and
confidence intervals that do not attain the nominal coverage; such issues are related to what has been described in the scientific literature as \emph{double dipping} \citep{kriegeskorte2009CircularAnalysisSystems, button2019DoubledippingRevisited}.

Our contributions enable statistical answers to data-driven questions of a network. For example, an analyst might estimate communities in a network and then seek a confidence interval for the difference in mean connectivities between two estimated communities (see Figure~\ref{fig:target_selection_cartoon}). Although the literature on community estimation \citep{amini2013PseudolikelihoodMethodsCommunity, harenberg2014community, hwang2024estimation}, statistical inference \citep{funke2019stochastic, fan2022SimpleStatisticalInference, tang2022AsymptoticallyEfficientEstimators, duranthon2023optimal, agterberg2023OverviewAsymptoticNormality}, and model selection \citep{wang2017likelihoodsbm, li2020NetworkCrossvalidationEdge} for networks is rich,
we are unaware of any general strategy that offers valid inference when a single network is used to both select a parameter \textit{and} conduct statistical inference on this data-driven parameter.

\begin{figure}[!htbp]
    \centering
    \includegraphics[width=0.90\linewidth]{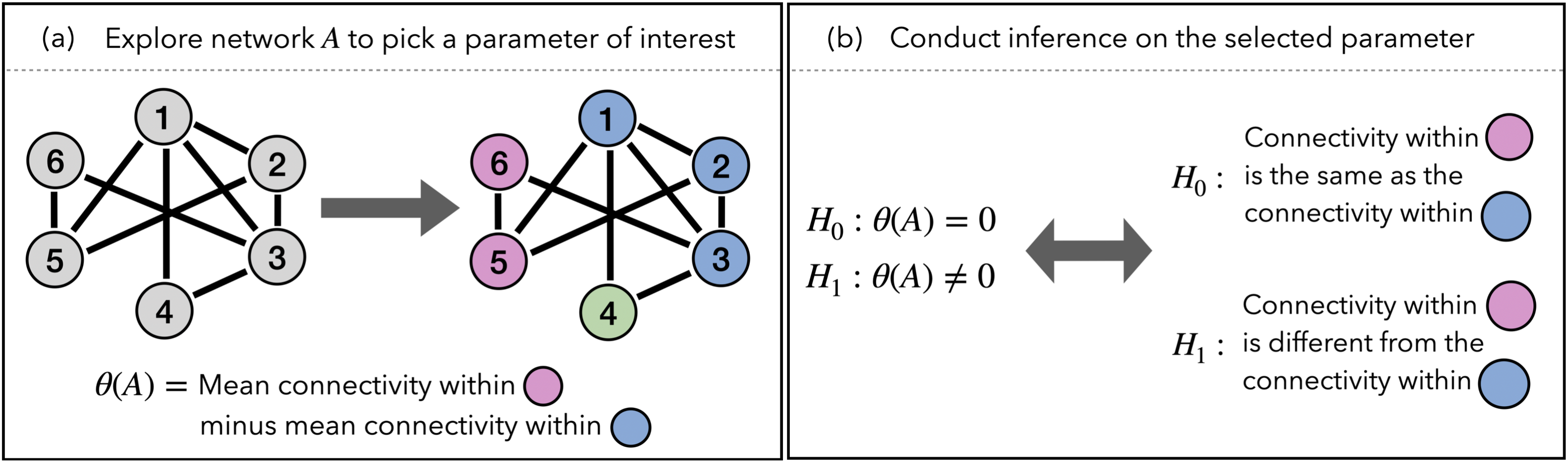}
\caption{We consider the setting where an analyst (a) uses a single realization of a network $A$ to select a parameter, and (b) proceeds to conduct inference on that parameter. In step (b), it is crucial to account for the fact that the parameter was selected using the data.}
    \label{fig:target_selection_cartoon}
\end{figure}
 
Often, the simplest strategy for inference on a data-driven parameter is sample splitting
\citep{cox1975NoteDataSplittingEvaluation}, in which a sample of $n$ independent
and identically distributed observations is partitioned into a train set and a test set. If the train and test sets are independent of one another, the train set can be used to select a parameter, and the test set can be used for inference on the selected parameter. However, when only a \textit{single} realization of a network is available, sample splitting cannot be
readily applied. \citet{chen2018NetworkCrossValidationDetermining} suggest
partitioning the nodes into two sets\textemdash$\mathcal{N}_1$ and $\mathcal{N}_2$\textemdash to
achieve two disjoint sets of edges: one set composed of all edges incident to
nodes in $\mathcal{N}_1$, and another set composed of the remaining edges, as shown in Figure~\ref{fig:cartoon}(a). However, this approach is
not applicable when the parameter of interest depends on the entire network (e.g., a function of the estimated community membership such as the average expected degree of all nodes in the first estimated community). \citet{chakrabarty2025NetworkCrossValidationModel} propose a computational improvement to this approach, which still inherits this restriction. In contrast,
\citet{li2020NetworkCrossvalidationEdge} hold out individual edges of the
network to use as the test set (see Figure~\ref{fig:cartoon}(b)), and show that under a low-rank
assumption, matrix completion techniques can be used to obtain a train network that asymptotically resembles the original network. However,
this approach is predicated on the assumption of a low-rank mean structure, is
applicable only to a relatively narrow class of parameters that can be estimated
using a small number of edges, lacks finite-sample guarantees, and requires that
the majority of the edges be used for training in order for the
matrix-completion to be well-behaved.

In this paper, our goal is to ``split'' a single realization of a network into train and test networks, where each contains the same set of nodes as the original network; see Figure~\ref{fig:cartoon}(c). We will then (i) select a parameter based on the train network, and (ii) conduct inference on that selected parameter using the test network. The strategy used to split the network into train and
test networks, and the details of inference with the test set, will depend on the
distribution of the edges. If the edges are independent and follow Poisson or Gaussian distributions, then
we apply \textit{data thinning} to obtain independent train and test networks that follow the
same distribution, up to a known scaling of the mean parameter \citep{rasines2023SplittingStrategiesPostselection,dharamshi2023GeneralizedDataThinning}. If the edges are
independent and follow a Bernoulli distribution, then we apply \textit{data fission} to obtain dependent train and test
networks, and we conduct inference using the test network conditional on the
train network \citep{leiner2025DataFissionSplitting}. In the specific case that each edge in the network follows
a Poisson distribution, our proposal is closely related to recent work by \citet{chen2021EstimatingGraphDimension};
however, we exploit recent developments in the field of selective inference to
expand the reach of that proposal to a far larger set of distributions, and
furthermore we focus on the task of inference. Our work bears a passing resemblance
to recent papers on the network jackknife and bootstrap, which involve
generating multiple ``copies'' of the network \citep{thompson2016UsingBootstrapStatistical, green2022BootstrappingExchangeableRandom, levin2025bootstrapping, linTheoreticalPropertiesNetwork}. However, in contrast to those proposals, our approach yields train and test networks whose dependence is
well-understood. This is critical to downstream inference on parameters selected with the train network.

This paper makes only two assumptions about the network: (i) each edge is
independent; and (ii) the edges are drawn from one of three distributions: Gaussian with known variance, Poisson, or Bernoulli. Critically, we do not make any further assumptions about the
parameters of the edge distributions nor their structure. For instance, we do not
assume that there are true communities in the network, nor that the network is
drawn from a specific model such as an SBM \citep{holland1983StochasticBlockmodelsFirst} or an RDPG \citep{young2007RandomDotProduct}. While the SBM acts as a working model to motivate the selected parameter that we consider in Section~\ref{sec:defining_target_of_inference} and beyond, our theoretical results require no such assumption and allow each edge to have a different mean parameter.

\begin{figure}[!htbp]
    \centering
\includegraphics[width=0.90\linewidth]{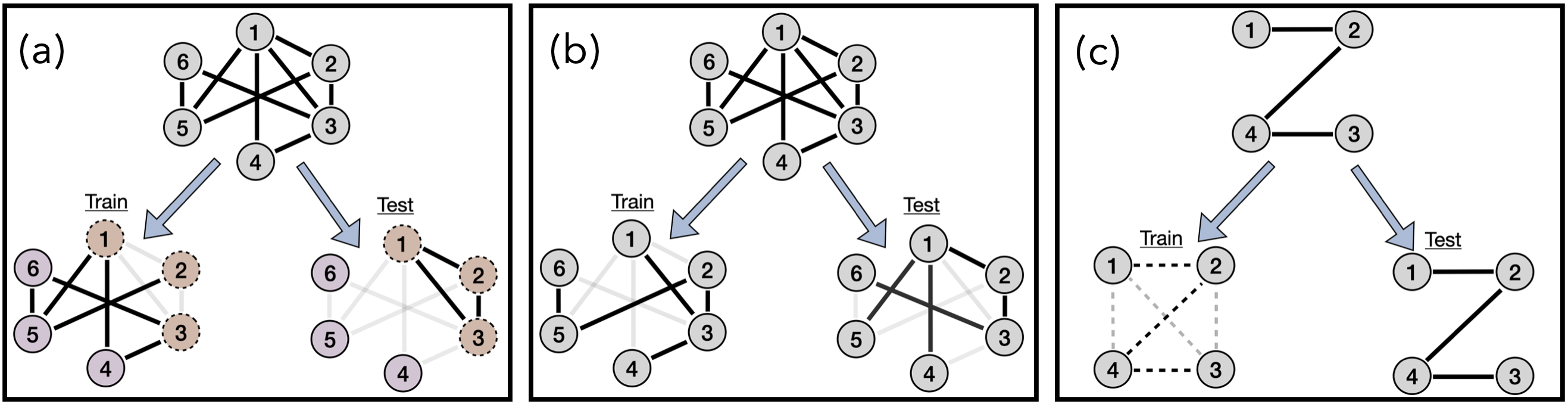}
\caption{\emph{(a):} \cite{chen2018NetworkCrossValidationDetermining} propose partitioning the
  nodes into two disjoint sets, depicted with solid and dashed circles. Edges incident to solid nodes are used for training, and testing is performed using the remaining edges. \emph{(b):} \cite{li2020NetworkCrossvalidationEdge} propose partitioning the edges into two disjoint
  sets: training uses the first set with the aid of matrix completion, and
  testing uses the second set. \emph{(c):} For networks with Bernoulli edges, our proposal produces a train network by ``toggling'' each edge
  (or non-edge) with probability $\gamma \in (0, 0.5)$ (see Proposition~\ref{prop:univariate_bernoulli_fission}). The conditional distribution of the original network given the train network is used for inference.}
    \label{fig:cartoon}
\end{figure}

The rest of this paper is organized as follows. We present an overview of the
general strategy in Section~\ref{sec:strategy}. Then, in Sections \ref{subsec:thinning_fission}--\ref{sec:inference_on_data_driven_parameter},
we instantiate each step of the general strategy in a setting where the edges of the network
are assumed to independently follow Gaussian (with common known variance), Poisson, or Bernoulli distributions. In Section~\ref{sec:simulation} we present a simulation study, and in Section~\ref{sec:dolphins_application} we consider an application to data consisting of the relationships among a group of dolphins in Doubtful Sound, New Zealand \citep{lusseau2003bottlenose}. The discussion is in
Section~\ref{sec:discussion}. Additional simulation details and proofs of all theoretical results are provided in the Supplement.

\section{The general strategy}
\label{sec:strategy} 
The edges in a network with $n$ nodes are represented via the adjacency
matrix $A \in \mathcal{S}^{n \times n}$, where the value of $A_{ij}$ encodes the status of an edge linking node \(i\) to node \(j\). In a Bernoulli
network, $\mathcal{S} = \left\{ 0, 1 \right\}$, where a zero indicates the
absence of an edge and a one indicates its presence. In a weighted
network, \(\mathcal{S}\) may be more general, e.g., all of \(\mathbb{R}\).
Networks may be undirected so that $A$ is an upper-triangular matrix, or disallow self-loops with the
convention that $A_{ii} = 0$ for all $i = 1, 2, \dots, n$. To streamline
discussion, in the main text we assume that $A$ is a directed network that
allows self-loops, but Supplement \ref{app:undirected} extends our results to 
undirected networks and networks that disallow self-loops.

We propose the following approach for inference on data-driven network parameters. 

\begin{algorithm}[Inference on data-driven network parameters]
  \label{alg:general_algorithm} \leavevmode
  \begin{enumerate}[I.]
  \item Split the adjacency matrix $A$ into two $n \times n$ adjacency
    matrices $\Atr$ and $\Ate$ such that the conditional distribution of $\Ate$
    given $\Atr$ is known, $A=T(\Atr, \Ate)$ for some deterministic
    function $T(\cdot,\cdot)$, and both $\Atr$ and $\Ate$ contain information
    about all unknown parameters in the distribution of $A$.
    \label{bulletpoint:step1_split_inference}
  \item Define a parameter $\theta(\Atr )$, which is a function of $\Atr$.
    \label{bulletpoint:step2_split_inference}
  \item Perform inference on $\theta(\Atr )$ using the conditional distribution
    of $\Ate \mid \Atr$.
    \label{bulletpoint:step3_split_inference}
  \end{enumerate}
\end{algorithm}

Our goal is to conduct valid inference on $\theta(\Atr )$, in the sense of
confidence sets that attain the nominal selective coverage \citep{fithian2017OptimalInferenceModel}. That is, for
any $\alpha \in (0,1)$, we want to construct $\Calpha(\Ate; \Atr)$
satisfying
\begin{equation}
    P\left(\theta\left(\Atr\right) \in \Calpha(\Ate; \Atr) \mid \Atr\right) \ge 1 - \alpha,
    \label{eq:prop_alg_control_have}
\end{equation}
where the probability is taken over the randomness in $\Ate \mid \Atr$.

\section{Step I:  splitting a single network}
\label{subsec:thinning_fission}

We now consider Step~\ref{bulletpoint:step1_split_inference} in Algorithm~\ref{alg:general_algorithm}, under the
assumption that the entries of the adjacency matrix \(A\) are mutually
independent, and $M_{ij} := \E \left[ A_{ij} \right]$ is unknown. In the case of Gaussian or Poisson edges, we make use of recent results
that allow us to ``thin'' each edge into two independent edges
\citep{neufeld2024DataThinningConvolutionclosed,
  dharamshi2023GeneralizedDataThinning,
  rasines2023SplittingStrategiesPostselection,
  tian2018SelectiveInferenceRandomized, leiner2025DataFissionSplitting},
ultimately arriving at two independent adjacency matrices. Critically, this is quite different from partitioning the edges or the nodes into two sets (see Figure~\ref{fig:cartoon}). 

\begin{proposition}[Thinning for Gaussian edges]
\label{prop:univariate_gaussian_thinning}
Suppose that $\epsilon \in (0,1)$, and $A_{ij} \ind \mathcal{N}(M_{ij}, \tau^2)$ for $i=1,\ldots,n$
and $j=1,\ldots,n$. For \\ $\Atr_{ij} \mid A_{ij} \ind \mathcal{N}(\epsilon A_{ij}, \epsilon (1- \epsilon) \tau^2)$ and $\Ate_{ij} := A - \Atr_{ij}$, it follows that
\begin{enumerate}[(i)]
    \item $\Atr_{ij} \sim \mathcal{N}(\epsilon M_{ij}, ~ \epsilon \tau^2)$, 
    \item $\Ate_{ij} \sim \mathcal{N}((1-\epsilon) M_{ij}, ~ (1-\epsilon) \tau^2)$, and 
    \item $\Atr$ is independent of $\Ate$. 
\end{enumerate}
\end{proposition}

\begin{proposition}[Thinning for Poisson edges] 
\label{prop:univariate_poisson_thinning}
Suppose that $\epsilon \in (0,1)$, and $A_{ij} \ind \Pois(M_{ij})$ for $i=1,\ldots,n$
and $j=1,\ldots,n$. For \\ $\Atr_{ij} \mid A_{ij} \ind \Bin(A_{ij}, \epsilon)$ and $\Ate_{ij} := A_{ij} - \Atr_{ij}$, it follows that
\begin{enumerate}[(i)]
\item $\Atr_{ij} \sim \Pois(\epsilon M_{ij})$,
\item $\Ate_{ij} \sim \Pois((1-\epsilon) M_{ij})$, and  
\item $\Atr$ is independent of $\Ate$.
\end{enumerate}
\end{proposition}

The $n \times n$ matrices $\Atr$ and $\Ate$ arising from
Propositions~\ref{prop:univariate_gaussian_thinning} and
\ref{prop:univariate_poisson_thinning} are independent, and the Fisher
information about the unknown parameter $M_{ij}$ is neatly allocated between
$\Atr_{ij}$ and $\Ate_{ij}$ in proportion to $\epsilon \in
(0,1)$ \citep{neufeld2024DataThinningConvolutionclosed,
  dharamshi2023GeneralizedDataThinning}. However, as demonstrated in
\citet{dharamshi2023GeneralizedDataThinning}, it is not possible to decompose
$A_{ij} \sim \Bern(M_{ij})$ into (non-trivially) independent
$\Atr_{ij}$ and $\Ate_{ij}$ that satisfy \(A_{ij} = T \left( \Aijtr, \Aijte
\right)\) for a deterministic function \(T\). Instead, we make use of results
from \citet{leiner2025DataFissionSplitting} to obtain a \emph{dependent} pair
\((\Aijtr, \Aijte)\).

\begin{proposition}[Fission for Bernoulli edges] \label{prop:univariate_bernoulli_fission}
Suppose that $\gamma \in (0, 0.5)$, and $A_{ij} \ind \Bern(M_{ij})$ for $i=1,\ldots,n$
and $j=1,\ldots,n$. For independent noise $W_{ij} \ind \Bern(\gamma)$, $\Atr_{ij} := A_{ij}(1 - W_{ij}) + (1 - A_{ij}) W_{ij}$, and \(\Ate_{ij} := A_{ij}\), it follows that
\begin{enumerate}[(i)]
\item
  \(\Atr_{ij} \sim \Bern(M_{ij} + \gamma - 2 M_{ij} \gamma)\), and
\item \label{prop:fission_define_T} $\Aijte \mid \Aijtr \sim \Bern(T_{ij})$, where $T_{ij} := \frac{M_{ij}}{M_{ij} + (1 - M_{ij}) \left(\frac{\gamma }{1 - \gamma}\right)^{2 \Atr_{ij} - 1}}$.
\end{enumerate}
\end{proposition}
When applying Proposition~\ref{prop:univariate_bernoulli_fission}, $\Aijtr$ is obtained by toggling the entry $A_{ij} \in \{0,1\}$ with probability
$\gamma$ (see Figure~\ref{fig:cartoon}(c)). For small values of $\gamma$, more information about $M_{ij}$ is
allocated to $\Aijtr$, and for values of $\gamma$ close to $0.5$, more
information about $M_{ij}$ is allocated to $\Aijte \mid \Aijtr$. Following
\citet{leiner2025DataFissionSplitting}, we refer to the application of
Proposition~\ref{prop:univariate_bernoulli_fission} as Bernoulli ``fission.''

\section{Step II: defining the selected parameter}
\label{sec:defining_target_of_inference}

Step~\ref{bulletpoint:step2_split_inference} of
Algorithm~\ref{alg:general_algorithm} involves selecting a parameter that is a
function of $\Atr$. While Algorithm~\ref{alg:general_algorithm} is generally applicable to any network parameter selected using $\Atr$, to fix ideas we consider estimating latent node attributes, and then we define a data-driven parameter that is a function of those latent node attributes.

The literature contains a number of network models in which the edge
distribution depends on latent node attributes. Examples include the SBM
\citep{holland1983StochasticBlockmodelsFirst} and the RDPG \citep{young2007RandomDotProduct,rubin-delanchy2022StatisticalInterpretationSpectral}.
However, in what follows we do not assume that
any such network model holds: we assume only that the selected parameter is a
function of some \emph{estimated} latent node attributes.

The estimated latent node attributes can be either discrete, as in the context
of an SBM, or continuous-valued, as in the context of an RDPG or a mixed membership SBM
\citep{airoldi2008MixedMembershipStochastic}. For simplicity, we consider
estimating discrete latent node attributes from the train network $\Atr$; we will
interpret these as estimated ``communities.'' To encode estimated community
membership, we use $\Ztr \in \{0, 1\}^{n \times K}$, where $\Ztr_{ik} = 1$ when the $i$th node belongs to the $k$th estimated community. We emphasize that $\Ztr$ is a function of $\Atr$, and
perhaps also of auxiliary randomness \citep[e.g., in the context of spectral
clustering, as in][]{amini2013PseudolikelihoodMethodsCommunity}; however, in
what follows, for simplicity of notation we suppress any dependence on
auxiliary randomness.

Next, we consider the $K \times K$ matrix
\begin{equation}
    B(\Atr) := \left( \ZtrT \Ztr \right)^{-1} \ZtrT \E \left[ A  \right] \Ztr \left( \ZtrT \Ztr \right)^{-1},
    \label{eq:B_matrix_theta_def}
\end{equation}
where $B(\Atr)$ depends on $\Atr$ through $\Ztr$. In what follows, we will often
suppress the argument $\Atr$ and simply write $B$. The $(k, \ell)$th entry of $B$ takes
the form
\begin{equation}
    B_{k \ell} = \dfrac{1}{|\Ic_{k \ell}|} \sum_{(i,j) \in \Ic_{k \ell}} \E \left[ A_{ij}  \right],
    \label{eq:B_entry_definition}
\end{equation}
where we define $\Ic_{k \ell} := \left\{ (i,j) ~:~ \Ztr_{ik} = 1, \Ztr_{j \ell} =
  1 \right\}$ (i.e., $\Ic_{k \ell}$ is the set of edges originating in the $k$th estimated community and
terminating in the $\ell$th estimated community). Hence, $B$
contains the mean pairwise connectivities between the $K$ estimated communities.
We define the selected parameter to be a linear combination of the elements of
$B$, i.e.,
\begin{align}
        \theta \left( \Atr \right) &:= u^\top \vct(B) = u^\top \vct \left( \left( \ZtrT \Ztr \right)^{-1} \ZtrT \E \left[ A \right] \Ztr \left( \ZtrT \Ztr \right)^{-1}\right),
    \label{eq:target_of_inference}
\end{align}
where $u \in \mathbb{R}^{K^2}$ satisfies $\Vert u \Vert_2 = 1$ and is allowed to
depend on $\Atr$ if desired. For example, if $u = (1, 0, \dots, 0)^\top$, then the selected parameter 
is the mean connectivity within the first estimated community, and if $u = \left( \frac{1}{\sqrt{2}}, -\frac{1}{\sqrt{2}}, 0, \dots,
  0\right)^{\top}$ then the selected parameter is the mean connectivity within the first estimated community minus the mean connectivity from the second to the 
first estimated community.

The selected parameter $\theta \left( \Atr \right)$ is random in the sense that
it depends on $\Atr$. Thus, to conduct valid inference on this parameter, in the
spirit of \citet{fithian2017OptimalInferenceModel} we will construct confidence
intervals that cover $\theta(\Atr)$ at a rate of $1 - \alpha$, \emph{conditional} on $\Atr$.

\begin{remark}
  Suppose that $A$ follows an SBM with $n$ nodes and $K$ communities, where $Z
  \in \{0, 1\}^{n \times K}$ encodes ``true'' community
  membership, and $C \in \mathbb{R}^{K \times K}$ is the connectivity
  matrix. Let \(Z_{i}\) denote the \(i\)th row of \(Z\). Then, for Gaussian, Poisson, or
  Bernoulli edges, it follows that $\E[A_{ij}] = Z_i C Z_j^\top$ and $\E[A]
  = Z C Z^\top$. Hence, $C = (Z^\top Z)^{-1} Z^\top \E[A] Z (Z^\top
  Z)^{-1}$.
    
  Consequently, when $\Ztr = Z$ (i.e., the true communities are exactly
  recovered), $B$ defined in (\ref{eq:B_matrix_theta_def}) equals $C$. Thus, in a sense, the SBM motivates the selected parameter in
  \eqref{eq:target_of_inference}. However, this paper does not assume that $A$
  follows an SBM.
\end{remark}

One may wonder about the usefulness of a data-driven parameter such as $\theta(\Atr)$ that depends on estimated communities, as we have not assumed the presence of true communities in the data. Even if true communities exist, our estimated communities are not likely to recover them exactly. Does this make $\theta(\Atr)$ meaningless? We argue that inference for $\theta(\Atr)$ is nonetheless useful: it allows practitioners to answer statistical questions about \textit{their} estimated communities, and to determine whether they have captured signal, as opposed to simply noise in the data.

For example, suppose that an analyst applies community detection to identify two communities in a network where there are no true communities (e.g., a network arising from a homogeneous Erd\H{o}s-R\'{e}nyi model). To assess whether the two detected communities are meaningful, the analyst can test the null hypothesis that the mean connectivity within the two communities does not differ from the mean connectivity between them. To accomplish this, they can set $u = \left( \frac{1}{\sqrt{4}}, -\frac{1}{\sqrt{4}}, -\frac{1}{\sqrt{4}}, \frac{1}{\sqrt{4}} \right)$ and construct a confidence interval for $\theta(\Atr)$. If the confidence interval contains zero, then the analyst does not have evidence to conclude that their estimated communities are meaningful, and so inference for $\theta(\Atr)$ proved informative even in the absence of true communities.

In practice, even when true communities exist, exact community recovery is rarely guaranteed, or even verifiable when only a single observed network is available. In contrast to results that require exact community recovery (e.g., \citet{tang2022AsymptoticallyEfficientEstimators}), our selected parameter $\theta(\Atr)$ remains informative even when community recovery is not exact. For example, if two true communities exist but the estimated communities differ from the true communities by a single node, then inference procedures relying on exact recovery may fail, whereas valid inference for $\theta(\Atr)$ is still available under our framework.



\section{Step III: inference for a selected parameter}
\label{sec:inference_on_data_driven_parameter}

The selected parameter $\theta(\Atr)$ defined in \eqref{eq:target_of_inference} is a function
of $\Atr$, so our interest lies in \emph{selective} coverage \citep{fithian2017OptimalInferenceModel} in the sense of \eqref{eq:prop_alg_control_have}.
In Section~\ref{sec:gaussian-poisson}, we first show how this can be accomplished via data thinning for Gaussian or Poisson edges, where the former is stated as a finite sample result and the latter as an asymptotic result. Finally, in Section~\ref{sec:bernoulli_edges}, we address the case of
Bernoulli edges, which requires special considerations due to the inter-dependence of $\Atr$ and $\Ate$.

\subsection{Networks with Gaussian and Poisson edges}
\label{sec:gaussian-poisson}

In the case of Gaussian edges, we obtain an exact finite sample result. Here, $\phi_{1 - \alpha/2}$ is the $(1-\alpha/2)$-quantile of the $\mathcal{N}(0,1)$ distribution.

\begin{proposition}
    \label{prop:normal_estimation}
    Suppose that the random adjacency matrix $A$ has entries $ A_{ij} \ind
    \mathcal{N}(M_{ij}, \tau^2)$ with common known variance $\tau^2$ and unknown
    mean $M_{ij}$. Suppose that we fix $\epsilon \in (0,1)$ and construct $\Ate$ and $\Atr$ from $A$ by Proposition \ref{prop:univariate_gaussian_thinning}, and we then apply community detection to $\Atr$
    to yield the estimated community membership matrix $\Ztr \in \{0,1\}^{n
      \times K}$. Define
    \begin{equation}
        \hat{\theta}\left( \Ate, \Atr \right) := (1-\epsilon)^{-1} u^\top \vct \left( \left( \ZtrT \Ztr \right)^{-1} \ZtrT \Ate \Ztr \left( \ZtrT \Ztr \right)^{-1} \right), \label{eq:thetahat-gaussian}
    \end{equation}
    where $u \in \mathbb{R}^{K^2}$ satisfies $\Vert u \Vert_2 = 1$, and is allowed to depend on $\Atr$ if desired. Then,
    \begin{align*}
        P \Bigg( \theta(\Atr) \in \Big[ \hat{\theta}(\Ate, \Atr) \pm \phi_{1-\alpha/2} \cdot \sigma \Big] ~\Bigg|~ \Atr \Bigg) = 1 - \alpha,
    \end{align*}
    where $\theta \left( \Atr \right)$ was defined in \eqref{eq:target_of_inference}, $\sigma^2 := (1-\epsilon)^{-1} \tau^2 u^\top \{(\ZtrT \Ztr)^{-1} \otimes (\ZtrT \Ztr)^{-1}\} u$,  and $\otimes$ is the Kronecker product.
\end{proposition}
Proposition~\ref{prop:normal_estimation} requires that the variance $\tau^2$ is known, a requirement of data thinning for Gaussian random variables (Proposition~\ref{prop:univariate_gaussian_thinning}). Recent work by \citet{dharamshi2024DecomposingGaussiansUnknown} suggests avenues to relax this requirement.

In the case of Poisson edges we arrive at the following asymptotic result.
\begin{proposition}
  \label{prop:poisson_estimation} Consider a sequence of $n \times n$ random adjacency matrices
  $(A_n)_{n=1}^\infty$ with entries $A_{n,ij} \ind \Pois(M_{n,ij})$,
  where $0 < N_0 \le M_{n,ij} \le N_1 < \infty$ holds for constants $N_0$ and
  $N_1$ not depending on $n$. Suppose we fix $\epsilon \in (0,1)$ and construct $\Antr$ and $\Ante$ from $A_n$ by applying Proposition~\ref{prop:univariate_poisson_thinning}, and then apply community detection to $\Antr$ to yield the estimated community
  membership matrices $\Zntr \in \{0,1\}^{n \times K}$. Define
  \begin{align} \hat{\theta}_n & \left( \Ante, \Antr \right) \nonumber :=
    (1-\epsilon)^{-1} u_n^\top \vct\left( \left( \ZntrT \Zntr
      \right)^{-1} \ZntrT \Ante \Zntr \left(\ZntrT \Zntr\right)^{-1}
    \right), \label{eq:thetahat_definition_poission}
  \end{align} where $u_n \in \mathbb{R}^{K^2}$ satisfies $\Vert u_n \Vert_2 =
  1$, and is allowed to depend on $\Antr$ if desired. Further define $\Ic_{n, k \ell} := \left\{ (i,j) :
    \Zniktr = 1, \Znjltr = 1 \right\}$, $\hat{B}_{n, k \ell}:= \frac{1}{|\Ic_{n,k \ell}|}
  \sum_{(i,j) \in \Ic_{n, k \ell}} \Anteij$, $\hat{\Delta}_n \in \mathbb{R}^{K \times K}$ with entries
  $ \hat{\Delta}_{n, k \ell} := \dfrac{ \hat{B}_{n, k \ell}}{|\Icnkl|}$, and $\hat{\sigma}^2_n:=
  (1-\epsilon)^{-2} u_n^\top \diag(\vct(\hat{\Delta}_n)) u_n$. Then, for
  $\theta_n(\Antr)$ defined in \eqref{eq:target_of_inference}, we have
    \begin{align*}
        \lim_{n \to \infty} P \Bigg( \theta_n(\Antr) \in \Big[ \hat{\theta}_n(\Ante, \Antr) \pm \phi_{1-\alpha/2} \cdot \hat{\sigma}_n \Big] ~\Bigg|~ \Antr \Bigg) = 1-\alpha,
    \end{align*}
  provided that the sequence of realizations $\left\{\Zntr =
    \hat{z}_n \right\}_{n=1}^{\infty}$ is such that $|\Icnkl|^{-1} = O(n^{-2})$ for all $k, \ell \in \{1, 2, \dots, K\}$.
\end{proposition}

\begin{remark}
  In Proposition \ref{prop:poisson_estimation}, we assume that the entries of the sequence of
  matrices \(\left( M_{n} \right)_{n = 1}^{\infty}\) are uniformly bounded away
  from zero. This implies that the sequence of networks is dense in the sense that
  the expected degree grows unboundedly
  \citep{bickel2009NonparametricViewNetwork}. While networks
  encountered in reality are typically sparse
  \citep{barabasi2016NetworkScience}, we assume this lower bound to justify the use of a normal approximation to the Poisson (see the
  proof in Supplement \ref{app:poisson_proposition_proof}).
  We will make an
  analogous assumption for networks with Bernoulli entries in Proposition
  \ref{prop:squiggle_estimation} for the same reason.
\end{remark}

\subsection{Networks with Bernoulli edges}
\label{sec:bernoulli_edges}

We now turn to Bernoulli edges, whose treatment is considerably more complicated
than the Gaussian and Poisson edges considered in Section~\ref{sec:gaussian-poisson}. In that section, recall that we conduct inference on $\theta(\Atr)$, defined in \eqref{eq:target_of_inference} as $u^\top
\vct\left( B(\Atr) \right)$, where the $(k, \ell)$th element of
$B(\Atr)$ is defined in \eqref{eq:B_entry_definition} to be $B_{k \ell}(\Atr) := \frac{1}{|\Ickl|} \sum_{(i,j) \in \Ickl }
M_{ij}$, where $M_{ij} := \E[A_{ij}]$.

Trouble arises because Bernoulli fission (Proposition \ref{prop:univariate_bernoulli_fission}) results in
\emph{dependent} networks \(\Atr\) and \(\Ate\) (recall that \(\Ate := A\)). Thus, $T_{ij} = \E[\Aijte \mid \Aijtr]$ defined in Proposition~\ref{prop:univariate_bernoulli_fission}(\ref{prop:fission_define_T}) is not equal to $M_{ij} = \E[A_{ij}]$. So, the arguments used to derive Propositions~\ref{prop:normal_estimation}~and~\ref{prop:poisson_estimation} would \textit{not} lead to inference on functions of $B_{k \ell}(\Atr)$, but rather on functions of
\begin{equation}
    V_{k\ell}(\Atr) := \frac{1}{|\Ickl|} \sum_{(i,j) \in \Ickl } E[\Aijte \mid \Aijtr] = \frac{1}{|\Ickl|} \sum_{(i,j) \in \Ickl } T_{ij}.
    \label{eq:Vkl}
\end{equation}
While $B_{k \ell}(\Atr)$ depends on $\Atr$ \textit{only} through the estimated communities $\Ztr = \hat{Z}(\Atr)$, $V_{k \ell}(\Atr)$ depends on $\Atr$ also through $T_{ij}$, whose definition is somewhat cryptic. Most importantly, the quantities $V_{k \ell}(\Atr)$ and $B_{k \ell}(\Atr)$ may be quite different, especially when $\gamma$ is small, as is shown in the blue curves in Figure~\ref{fig:daniela}.

\begin{figure}[!htbp]
    \centering
\includegraphics[width=0.9\linewidth]{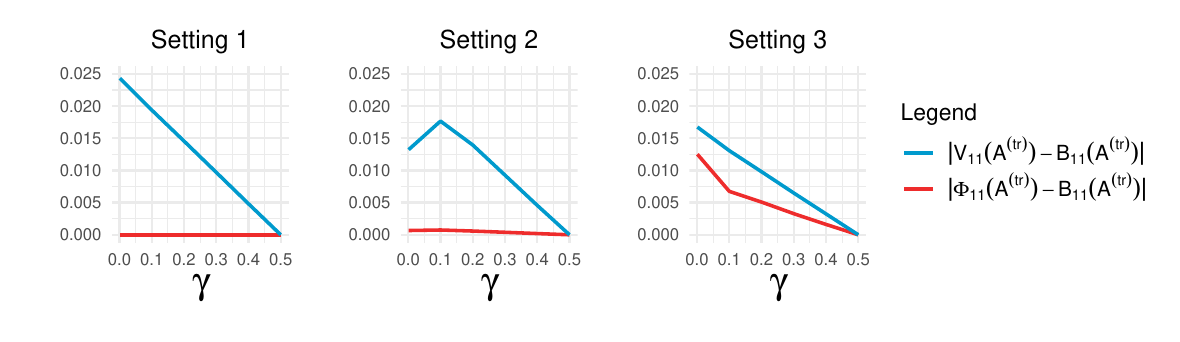}
    \caption{Simulations comparing $|V_{11}(\Atr) - B_{11}(\Atr)|$  (blue curves) and $|\Phi_{11}(\Atr) - B_{11}(\Atr)|$ (red curves) where $B_{k \ell}(\Atr)$, $V_{k \ell}(\Atr)$, and $\Phi_{k \ell}(\Atr)$ are defined in \eqref{eq:B_entry_definition}, \eqref{eq:Vkl},~and~\eqref{eq:Phi_kl_def} respectively, plotted over a range of $\gamma$. The networks have $n=100$ nodes, and results are averaged across 5,000 repetitions. \textit{Setting 1:} $M_{ij}=0.5$ for all $i$ and $j$. \textit{Setting 2:} The entries of $M$ belong to two equally-sized communities, where the intra-community entries of $M$ equal $0.6$ and the inter-community entries equal $0.4$. \textit{Setting 3:} Each entry of $M$ is drawn from a $\Unif(0,1)$ distribution.} 
    \label{fig:daniela}
\end{figure}

Unfortunately, inference on $V_{k \ell}(\Atr)$ does not directly enable inference on $B_{k \ell}(\Atr)$. To see why, following \citet{neufeld2025DiscussionDataFission}, for \(M_{ij} \in (0, 1)\),
 \begin{align}
   T_{ij}   
     &= \begin{cases}
        f(M_{ij},  \frac{1-\gamma}{\gamma} ), & \text{if } \Aijtr = 1, \\
       f(M_{ij},  \frac{\gamma}{1-\gamma} ), & \text{if } \Aijtr = 0,
    \end{cases}
    \label{eq:annas_insight_1}
\end{align}
where the function $f: (0,1) \times \mathbb{R}_+ \rightarrow (0,1)$ is defined as
\begin{equation}
    f(a,v) :=  \expit \left( \logit(a) + \log \left( v \right) \right),
    \label{eq:f_def_bernoulli}
\end{equation}
and has the property that $f(f(a,v), 1/v)=a$. Although \eqref{eq:annas_insight_1} reveals an invertible mapping between $T_{ij}$ and $M_{ij}$ when $\Aijtr$ is known, there is \textit{not} an invertible mapping between $V_{k \ell}(\Atr)$ from \eqref{eq:Vkl} and $B_{k \ell}(\Atr)$ from \eqref{eq:B_entry_definition}. Thus, inference on $V_{k \ell}(\Atr)$ does not enable inference on $B_{k \ell}(\Atr)$.

Therefore, we propose an alternative selected parameter for which inference \textit{is} possible, and we show that it is very close to $B_{k \ell}(\Atr)$. Toward this aim, define
\begin{equation}
    V_{k\ell}^{(0)}(\Atr)  := \frac{1}{|\Iczkl|} \sum_{(i,j) \in \Iczkl } T_{ij},
    ~~\text{and}~~V_{k\ell}^{(1)}(\Atr) := \frac{1}{|\Icokl|} \sum_{(i,j) \in \Icokl }
T_{ij}, \label{eq:V}
\end{equation}
where $\Ic_{k \ell}^{(0)} := \left\{(i,j) \in \Ickl : \Aijtr =
0 \right\}$ and $\Ic_{k \ell}^{(1)} := \left\{(i,j) \in \Ickl : \Aijtr = 1 \right\}$. Our next
result establishes that
   \begin{equation}\Phi_{k \ell}(\Atr) := \frac{|\Iczkl|}{|\Ickl|}
f\left(V_{k\ell}^{(0)}(\Atr), \frac{1-\gamma}{\gamma}\right) +
\frac{|\Icokl|}{|\Ickl|} f\left(V_{k\ell}^{(1)}(\Atr),
\frac{\gamma}{1-\gamma}\right)\label{eq:Phi_kl_def}\end{equation} is
close to $B_{k \ell}(\Atr)$; see also the red curves in
Figure~\ref{fig:daniela}. 

\begin{proposition}\label{prop:taylor}
    Consider $\Atr$ fixed and recall the definitions of $B_{k \ell}(\Atr)$ and $\Phi_{k \ell}(\Atr)$ in \eqref{eq:B_entry_definition} and \eqref{eq:Phi_kl_def} respectively, and define $B^{(s)}_{k \ell} := \frac{1}{|\Icskl|} \sum_{(i,j) \in \Icskl} M_{ij}$ for $s \in \{0,1\}$. 
    \begin{enumerate}[(a)]

        \item \label{prop:taylor_a} For some $t_0, t_1 \in (0,1)$, it holds that
        \begin{align*}
            \Phi_{k \ell}(\Atr) = B_{k \ell}(\Atr) + \sum_{s \in \{0,1\}} \frac{|\Icskl|}{|\Ickl|} \sum_{\substack{(i,j) \in \Icskl \\ (i',j') \in \Icskl}} (M_{ij} - B^{(s)}_{k \ell})(M_{i'j'}-B^{(s)}_{k \ell}) h^{(s)(k \ell)}_{iji'j'} (t_s),
        \end{align*}
        where $h^{(s)(k \ell)}_{iji'j'}$ is defined in \eqref{eq:h_function_prop6a_def} in Supplement~\ref{app:taylor_a_proof}.

        \item \label{prop:taylor_b} We have $\Phi_{k \ell}(\Atr) = B_{k \ell}(\Atr) + \left( 1 - \frac{\gamma}{1-\gamma} \right) \left( \frac{1}{|\Ickl|} (H^{(0)}_{k \ell} - H^{(1)}_{k \ell} )\right) + R_{k \ell}$, where for $s \in \{0,1\}$ we define $H^{(s)}_{k \ell} := \sum_{(i,j) \in \Icskl} (M_{ij} - \Bskl)^2$, and where $R_{k \ell} := \left( 1- \frac{\gamma}{1-\gamma} \right)^2 q_{k \ell}(\lambda_0, \lambda_1)$
        is a remainder term involving $\lambda_0, \lambda_1 \in \left[ \frac{\gamma}{1-\gamma}, 1 \right]$, where $q_{k \ell}$ is a continuous function defined in \eqref{eq:q_s_def_prop6b_proof} in Supplement~\ref{app:taylor_b_proof}.
    \end{enumerate}
\end{proposition}

\begin{remark}
    \label{remark:interpreting_taylor_a}
    Proposition \ref{prop:taylor}(\ref{prop:taylor_a}) implies that $\Phi_{k \ell}(\Atr)$ is close to $B_{k \ell}(\Atr)$ when the values in $\left\{M_{ij} ~:~ (i,j) \in \Icskl \right\}$ are close to their mean $\Bskl = \frac{1}{|\Icskl|} \sum_{(i,j) \in \Icskl} M_{ij}$ for $s \in \{0,1\}$. Indeed, $\Phi_{k \ell}(\Atr) = B_{k \ell}(\Atr)$ if $M_{ij} = \Bskl$ for all $(i,j) \in \Icskl$ for both $s = 0$ and $s=1$.
\end{remark}

\begin{remark}
    When $\gamma = 0.5$, we have $T_{ij} = f(M_{ij}, 1) = M_{ij}$, and so $\Phi_{k \ell}(\Atr) = B_{k \ell}(\Atr)$. Because $R_{k \ell}$ from Proposition \ref{prop:taylor}(\ref{prop:taylor_b}) is continuous in $\gamma$, this implies that $\Phi_{k \ell}(\Atr) \rightarrow B_{k \ell}(\Atr)$ as $\gamma \to 0.5$. However, setting $\gamma=0.5$ in practice is not recommended, as then $\Aijtr \ind \text{Bernoulli}(0.5)$, and so community estimation is conducted on pure noise.
 
     Proposition \ref{prop:taylor}(\ref{prop:taylor_b}) also suggests that $\lvert \Phi_{k \ell}(\Atr) - B_{k \ell}(\Atr)\rvert$ is smaller whenever \(\Hkl{0}\) and \(\Hkl{1}\) are of similar magnitude, where \(\Hkl{s}\) from Proposition~\ref{prop:taylor}(\ref{prop:taylor_b}) is interpreted as a measure of heterogeneity within $\left\{ M_{ij} ~:~ (i,j) \in \Icskl \right\}$. 
\end{remark}


As $\gamma \to 0.5$, less information is allocated to $\Atr$ and more information is allocated to $\Ate \mid \Atr$. Conversely, as $\gamma \to 0$, more information is allocated to $\Atr$. Proposition~\ref{prop:gamma_to_zero} establishes the behavior of the target of inference $\Phi_{k \ell}(\Atr)$ as $\gamma \to 0$: it converges to a weighted average of functions of the arithmetic and harmonic means of subsets of the $\frac{M_{ij}}{1-M_{ij}}$'s in the $(k, \ell)$th estimated community pairs. This quantity is related to $B_{k \ell}(A)$, which is the arithmetic mean of the elements of $M$ in the $(k, \ell)$th estimated community pair, when communities are estimated on the original network $A$.

\begin{proposition}\label{prop:gamma_to_zero} 
    Suppose that $M_{ij} \in (0,1)$ for all $i$ and $j$, and that for fixed $n$, the community detection algorithm $\hat{Z}(\cdot)$ is such that $\lim_{\gamma \to 0} P \left( \Ickl \ne \Ickltil \right) = 0$, where
    $\Ickltil := \left\{ (i,j): \hat{Z}(A)_{ik} =1, \hat{Z}(A)_{j \ell} = 1 \right\}$. Define \\
    $\Icskltil := \left\{ (i,j) \in \Ickltil: A_{ij} = s \right\}$ for $s \in \{0,1\}$,
    \begin{equation*}
        \tilde{\Phi}_{k \ell}(A) := \frac{|\Iczkltil|}{|\Ickltil|} \expit \left( \log \left( \Lambda^{(0)}_{k \ell} \right) \right) + \frac{|\Icokltil|}{|\Ickltil|} \expit \left( \log \left( \Lambda^{(1)}_{k \ell} \right) \right),
    \end{equation*}
    $\Lambda^{(0)}_{k \ell} := \frac{1}{|\Iczkltil|} \sum_{(i,j) \in \Iczkltil} \frac{M_{ij}}{1-M{ij}}$ (the arithmetic mean of the odds $\left\{ \frac{M_{ij}}{1-M_{ij}} ~:~ (i,j) \in \Iczkltil \right\}$), and $\Lambda^{(1)}_{k \ell} := \left( \frac{1}{|\Icokltil|} \sum_{(i,j) \in \Icokltil} \frac{1-M_{ij}}{M_{ij}} \right)^{-1}$ (the harmonic mean of the odds $\left\{ \frac{M_{ij}}{1-M_{ij}} ~:~ (i,j) \in \Icokltil \right\}$). Then, $\lim_{\gamma \to 0} P \left( \left| \Phi_{k \ell}(\Atr) - \tilde{\Phi}_{k \ell}(A) \right| > \epsilon \right) = 0$ for all $\epsilon > 0$.
\end{proposition}

As $\Phi_{k \ell}(\Atr)$ is approximately $B_{k \ell}(\Atr)$, we define our selected parameter to be
\begin{equation}
    \xi(\Atr) := u^\top \vct(\Phi(\Atr)),
    \label{eq:xi}
\end{equation}
where the $(k, \ell)$th entry of $\Phi(\Atr) \in \mathbb{R}^{K \times K}$ is defined in \eqref{eq:Phi_kl_def}, and where $u \in \mathbb{R}^{K^2}$ satisfies $\Vert u \Vert_2 = 1$ and may depend on $\Atr$. Our next result establishes that we can construct asymptotically valid confidence intervals for $\xi(\Atr)$.  

\begin{proposition}
    \label{prop:squiggle_estimation}
    Consider a sequence of $n \times n$ random adjacency matrices $(A_n)_{n=1}^\infty$, consisting of entries $A_{n,ij} \ind \Bern(M_{n,ij})$, where $0 < N_0 \le
    M_{n,ij} \le N_1 < 1$ holds for constants $N_0$ and $N_1$ not depending on $n$. Suppose
    that we fix $\gamma \in (0, 0.5)$ and construct $\Antr$ and $\Ante$ from $A$ as in Proposition \ref{prop:univariate_bernoulli_fission}, and then apply community detection to $\Antr$ to
    yield the estimated community membership matrix $\Zntr \in \{0,1\}^{n \times K}$.
    
    Define $\hat{\xi}_n(\Ante, \Antr) := u_n^\top \vct \left( \hat{\Phi}_n \left(\Ante,  \Antr \right) \right)$, where $u_n \in \mathbb{R}^{K^2}$ satisfies $\Vert u_n \Vert_2 = 1$, and may depend on $\Antr$ if desired, and where $\hat{\Phi}_n(\Ante, \Antr) \in \mathbb{R}^{K \times K}$ is defined entry-wise as
    \begin{equation*}
        \hat{\Phi}_{n, k \ell}(\Ante, \Antr) := \dfrac{|\Icnzkl|}{|\Icnkl|} \hat{V}^{(0)}_{n, k \ell}(\Ante, \Antr) + \dfrac{|\Icnokl|}{|\Icnkl|} \hat{V}^{(1)}_{n, k \ell}(\Ante, \Antr),
    \end{equation*}
    where $\Icnkl := \left\{ (i,j) : \Zniktr =1, \Znjltr = 1 \right\}$, $\Icnskl := \left\{ (i,j) \in \Icnkl : \Antrij = s \right\}$ for $s \in \{0, 1\}$, $\hat{V}^{(s)}_{n, k \ell}(\Ante, \Antr) := \frac{\hat{B}^{(s)}_{n, k \ell}}{\hat{B}^{(s)}_{n, k \ell} + (1 - \hat{B}^{(s)}_{n, k \ell}) e^{c^{(s)}}}$ for $c^{(0)} := \log(\gamma/(1-\gamma))$ and $c^{(1)} := \log((1-\gamma)/\gamma)$, and $\hat{B}^{(s)}_{n, k \ell} := \frac{1}{|\Icnskl|} \sum_{(i,j) \in \Icnskl} \Anteij$. Also define $\hat{\sigma}_n^2 := u_n^\top \diag(\vct(\hat{\Delta}_n)) u_n$ and
    \begin{equation*}
        \hat{\Delta}_{n, k \ell} := \sum_{s \in \{0,1\}} \frac{|\Icnskl|^2}{|\Icnkl|^2} \hat{\Delta}^{(s)}_{n, k \ell},~\text{where}~~\hat{\Delta}^{(s)}_{n, k \ell} := \frac{\hat{B}^{(s)}_{n,k \ell} (1-\hat{B}^{(s)}_{n,k \ell}) e^{2 c^{(s)}}}{|\Icnskl|((1-\hat{B}^{(s)}_{n,k \ell}) e^{c^{(s)}} + \hat{B}^{(s)}_{n,k \ell})^4}. 
    \end{equation*}
    Then, for $\xi_n(\Antr)$ defined in \eqref{eq:xi}, we have that
    \begin{align*}
        \liminf_{n \to \infty} P \Bigg( \xi(\Antr) \in \Big[ \hat{\xi}(\Ante, \Antr) \pm \phi_{1-\alpha/2} \cdot \hat{\sigma}_n \Big] ~\Bigg|~ \Antr \Bigg) \ge 1-\alpha,
    \end{align*}
    provided that the sequence of realizations $\left\{\Antr = \antr \right\}_{n=1}^\infty$ and $\left\{ \Zntr = \hat{z}_n \right\}_{n=1}^\infty$ are such that $|\Icnzkl|^{-1} = O(n^{-2})$ and $|\Icnokl|^{-1} = O(n^{-2})$ for all $k, \ell \in \{1, 2, \dots, K\}$.
\end{proposition}

\begin{corollary}
    \label{cor:exact_coverage_bernoulli_target}
    Under the conditions of Proposition~\ref{prop:squiggle_estimation}, if the sequence of realizations $\left\{\Antr = \antr\right\}_{n=1}^{\infty}$ and $\left\{\Zntr = \hat{Z}_n\right\}_{n=1}^{\infty}$ are such that there exists an $N$ such that for all $n \ge N$, the set $\left\{ M_{n,ij} ~:~ (i,j) \in \Ickl \right\}$ is constant for all $(k, \ell)$ where the corresponding entry of $u_n \in \mathbb{R}^{K^2}$ is nonzero, then $\lim_{n \to \infty} P \Bigg( \theta(\Antr) \in \Big[ \hat{\xi}(\Ante, \Antr) \pm \phi_{1-\alpha/2} \cdot \hat{\sigma}_n \Big] ~\Bigg|~ \Antr \Bigg) = 1-\alpha$, where $\theta_n(\Antr)$ is the original selected parameter defined in \eqref{eq:target_of_inference}.
\end{corollary}

See Supplement~\ref{appendix:numerical_issues_squiggle_estimation} for a discussion of numerical issues and suggested fixes when applying Proposition~\ref{prop:squiggle_estimation}.

\section{Simulation study} \label{sec:simulation}

\subsection{Data generation and simulation design}
\label{subsec:simulation_methods}

For networks with Gaussian, Poisson, and Bernoulli edges, we simulate $A$ from an SBM of size $n \times n$ with  $\Ktr$ equally-sized unknown communities, and with a  $\Ktr \times \Ktr$ mean matrix
$C$ that takes value \(\rho_{1}\) along its diagonal, and \(\rho_{2}\)
everywhere else. So, $\rho_1$ is the intra-community connectivity, $\rho_2$ is the inter-community connectivity, and $|\rho_1 - \rho_2|$ is a measure of the separation between communities \citep{abbe2018CommunityDetectionStochastic}.

For networks with Bernoulli edges (Section~\ref{subsec:simulation_results_bernoulli}), we consider an additional setting where each $M_{ij} := \E[A_{ij}]$ is drawn independently from a $\Unif(0,1)$ distribution. In light of Proposition~\ref{prop:taylor}(\ref{prop:taylor_a}) and Remark~\ref{remark:interpreting_taylor_a}, this is an unfavorable scenario for our proposed selected parameter $\xi(\Atr)$ in \eqref{eq:xi}, in the sense that it will be far from $\theta(\Atr)$ in \eqref{eq:target_of_inference}.

Considering the proposal of this paper summarized by Algorithm~\ref{alg:general_algorithm}, we (i) split $A$ into $\Atr$ and $\Ate$ following Proposition~\ref{prop:univariate_gaussian_thinning} or \ref{prop:univariate_poisson_thinning} with parameter $\epsilon \in (0,1)$ (Gaussian or Poisson edges), or
Proposition \ref{prop:univariate_bernoulli_fission} (Bernoulli edges) with parameter $\gamma \in (0,0.5)$; (ii) use $\Atr$ to estimate communities $\Ztr \in \{0, 1\}^{n \times
  K}$ (where $K$ is varied may not equal $\Ktr$), and
define the selected parameter as $\theta(\Atr)$ in \eqref{eq:target_of_inference} for Gaussian or
Poisson edges and $\xi(\Atr)$ in \eqref{eq:xi} for Bernoulli edges, where $u = (1, 0, 0, \dots, 0)^{\top}$. This choice of $u$ simplifies to conducting inference for $B_{11}(\Atr)$ for Gaussian or Poisson edges, and $\Phi_{11}(\Atr)$ for Bernoulli edges, so we refer to the selected parameters as $B_{11}(\Atr)$ and $\Phi_{11}(\Atr)$. Finally, we (iii) apply one of Propositions \ref{prop:normal_estimation},
\ref{prop:poisson_estimation}, or \ref{prop:squiggle_estimation} to construct confidence intervals for the selected parameter.

We compare the proposed methods to a ``naive'' one, in which we (i) estimate communities  $\hat{Z} \in \mathbb{R}^{n\times K}$ with $A$; and then (ii) also use $A$ to build confidence intervals for  
\begin{equation}
    \theta(A) := u^\top \vct \left( \left( \hat{Z}^\top \hat{Z} \right)^{-1} \hat{Z}^\top \E[A] \hat{Z} \left( \hat{Z}^\top \hat{Z} \right)^{-1} \right),
    \label{eq:naive_target}
\end{equation}
without accounting for the use of the network in both community estimation and the construction of confidence intervals. Once again, we set $u = (1, 0, 0, \dots, 0)^{\top}$ and simplify the notation to $B_{11}(A)$. Details of the naive confidence intervals are in Supplement \ref{appendix:naive-cis}.

In all simulations, communities are estimated with spectral clustering using the proposal of  \citet{amini2013PseudolikelihoodMethodsCommunity} as implemented in the \verb|nett| R package \citep{aminiNettNetworkAnalysis}. Additional simulation details can be found in Supplement~\ref{appendix:sim-details}.

\subsection{Results for networks with Gaussian and Poisson edges}
\label{subsec:simulation_results_gaussian_poisson}

We simulate 5,000 networks with $n=200$, $\rho_1=30$, $\rho_2=27$, $\Ktr=5$, and we vary the value of $K$. For networks with Gaussian edges, we set the known variance to be $\tau^2 = 25$.
For each simulated dataset, we construct confidence intervals for $B_{11}(\Atr)$ or $B_{11}(A)$ using the
proposed and naive methods, respectively, as described in
Section~\ref{subsec:simulation_methods}. The left-hand panels of Figures~\ref{fig:conf_width_rand_gaussian}~and~\ref{fig:conf_width_rand} display the empirical versus nominal coverages of the confidence intervals for Gaussian and Poisson edges respectively. Even when $K$ is not equal to $\Ktr$, the proposed approach achieves the nominal coverage, whereas the naive method does not. In Supplement~\ref{appendix:extra-simulations}, we show similar results as $n$ varies.

Next, we simulate 5,000 networks with $n=200$, $K=5$, $\Ktr=5$, and $\rho_1=30$. The center and right-hand panels of Figures~\ref{fig:conf_width_rand_gaussian}~and~\ref{fig:conf_width_rand} display the average confidence interval width of the proposed method for $B_{11}(\Atr)$ and the average adjusted Rand index \citep{hubert1985ComparingPartitions} between the estimated and true community memberships, as $\rho_2$ and $\epsilon$ are varied. As $\epsilon$ increases, more information is allocated to estimating communities, and less is allocated to inference, leading to an improved adjusted Rand index but wider confidence intervals.
Furthermore, the adjusted Rand index grows with $\left\lvert \rho_{1} - \rho_{2} \right\rvert$, as community detection is easier when $\left\lvert \rho_1 - \rho_2 \right\rvert$ is large. The naive approach is not displayed in the center and right-hand panels of Figures~\ref{fig:conf_width_rand_gaussian}~and~\ref{fig:conf_width_rand}, as the left-hand panels of the figures indicate that it does not achieve the nominal coverage.

Supplement~\ref{appendix:extra-simulations} contains additional simulations where prior to simulating each network, we draw $M_{ij} \ind \Unif(0, 20)$, showing that the proposed approach achieves valid coverage for the selected parameter even when there are no true communities.

\begin{figure}[!htbp]
    \centering
    \fbox{
    \includegraphics[height=3.2cm]{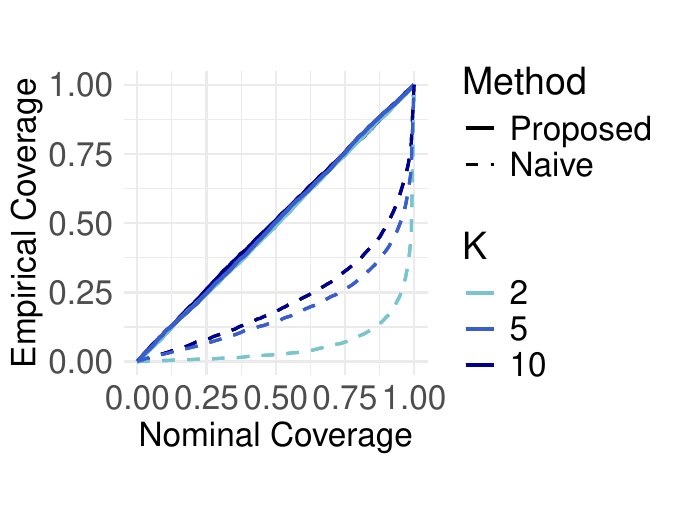}
    }
    \fbox{
    \includegraphics[height=3.2cm]{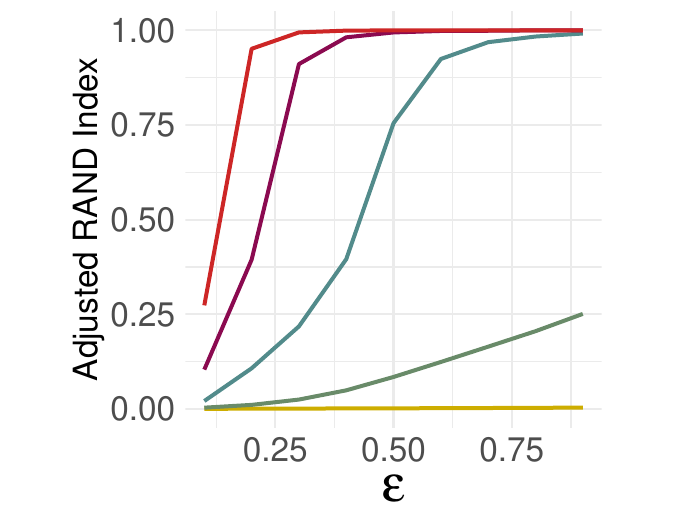}
    \includegraphics[height=3.2cm]{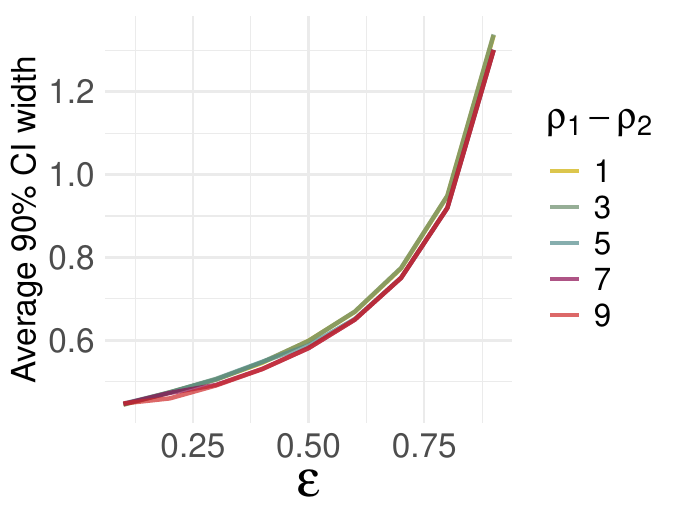}
    }
    \caption{Results for Gaussian edges, averaged over 5,000 simulated networks. \emph{Left:} Empirical versus nominal coverage of the confidence intervals for $B_{11}(\Atr)$ (proposed approach as described in Proposition~\ref{prop:poisson_estimation}) or $B_{11}(A)$ (naive approach as described in Supplement~\ref{appendix:naive-cis}), with $n=200$, $\Ktr=5$, $\rho_1=30$, $\rho_2=27$, $\tau^2 = 25$, and $\epsilon=0.5$ for the proposed approach. \emph{Center and Right:} Average adjusted Rand index between true and estimated communities, and average 90\% confidence interval width, as a function of $\epsilon$, for the proposed approach on networks with $n=200$, $K = 5$, $\Ktr = 5$, $\rho_1 = 30$, and $\tau^2 = 25$. \label{fig:conf_width_rand_gaussian}}
\end{figure}
\begin{figure}[!htbp]
    \centering
    \fbox{
    \includegraphics[height=3.2cm]{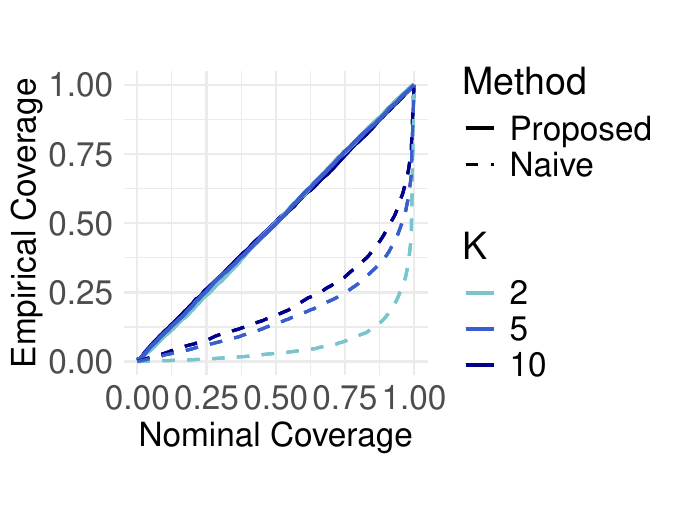}
    }
    \fbox{
    \includegraphics[height=3.2cm]{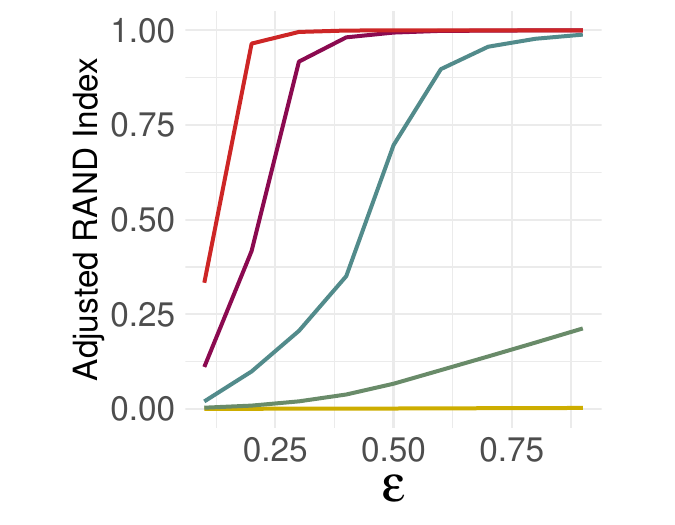}
    \includegraphics[height=3.2cm]{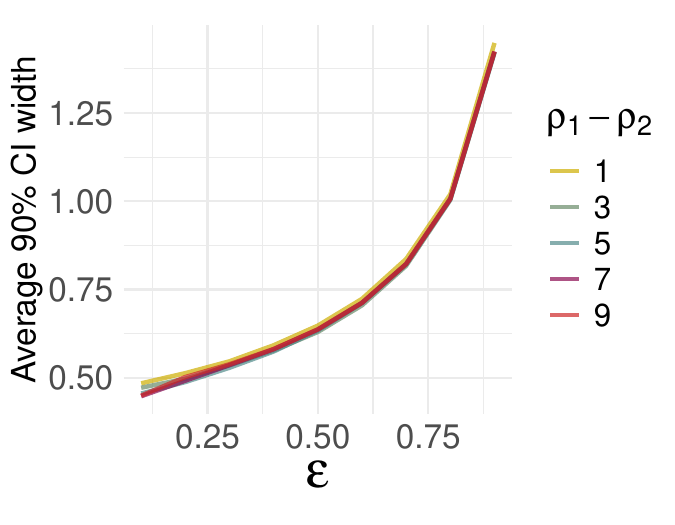}
    }
    \caption{Results for Poisson edges. All other details are the same as Figure~\ref{fig:conf_width_rand_gaussian}.}
    \label{fig:conf_width_rand}
\end{figure}

\subsection{Results for networks with Bernoulli edges}
\label{subsec:simulation_results_bernoulli}

We simulate 5,000 networks with Bernoulli edges, with $\rho_1 = 0.75$, $\rho_2 = 0.50$, $\Ktr = 5$, and we vary the value of $K$. For each simulated network, we consider both the proposed method and the naive method as described in Section~\ref{subsec:simulation_methods}. When employing the proposed method, we use Proposition~\ref{prop:squiggle_estimation} to construct confidence intervals targeting $\Phi_{11}(\Atr)$, but report the coverage for $B_{11}(\Atr)$, where we remind the reader that the latter is the ultimate parameter of interest (see Section~\ref{sec:bernoulli_edges}). The left-hand panel of Figure~\ref{fig:conf_width_rand_bernoulli} displays the empirical versus nominal coverages of the proposed and naive methods for networks with Bernoulli edges with $n = 200$, $\Ktr = 5$, and $(\rho_1, \rho_2) = (0.75, 0.5)$, and $\gamma = 0.25$. The proposed approach achieves the nominal coverage for $B_{11}(\Atr)$, even though the coverage guarantee is given for the related quantity $\Phi_{11}(\Atr)$. The naive method does not achieve the nominal coverage.

The center and right-hand panels of Figure~\ref{fig:conf_width_rand_bernoulli} show the results for a simulation setting where $n=200$, $K=5$, $\Ktr=5$, and $\rho_1=0.75$, as we vary the values of $\rho_2$ and $\gamma$. The center and right-hand panels display the average adjusted Rand index between the estimated and true community memberships, and the average confidence interval width, respectively. As $\gamma$ increases, less information is allocated to $\Atr$, and more information is available for inference, leading to a decrease in adjusted Rand index but narrower confidence intervals. Furthermore, the adjusted Rand index is larger when $\left\lvert\rho_1 - \rho_2\right\rvert$ is high, due to the increased separation between communities. The naive approach is not displayed in the center and right-hand panels of Figure~\ref{fig:conf_width_rand_bernoulli}, as the left-hand panel of Figure~\ref{fig:conf_width_rand_bernoulli} indicates that it does not achieve the nominal coverage.
\begin{figure}[!htbp]
    \centering
    \fbox{
    \includegraphics[height=3.2cm]{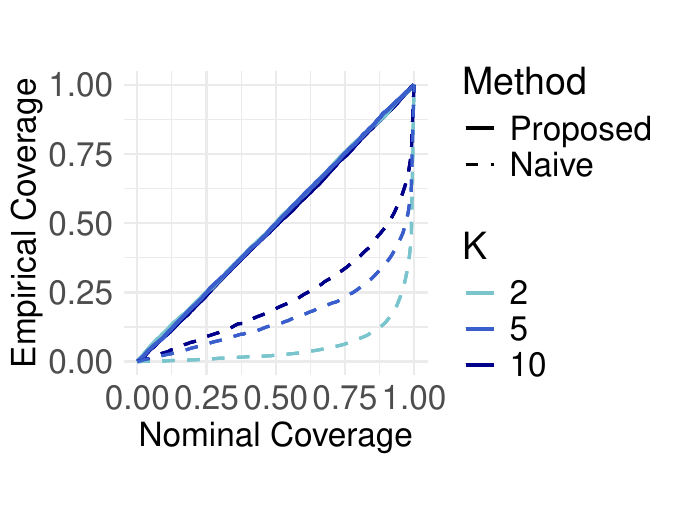}
    }
    \fbox{
    \includegraphics[height=3.2cm]{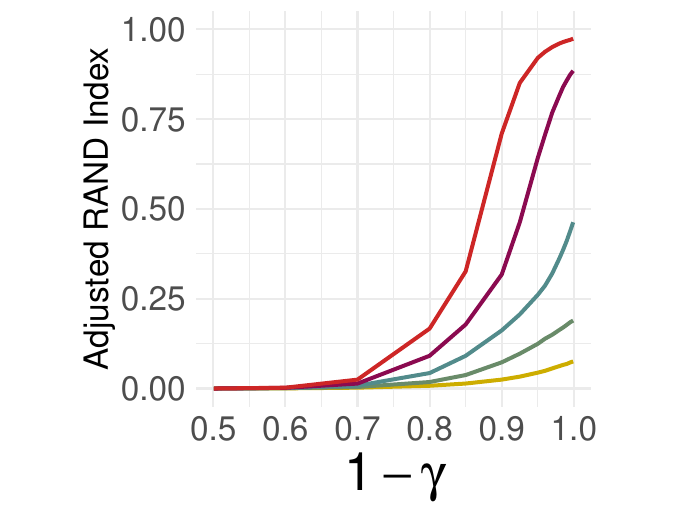}
    \includegraphics[height=3.2cm]{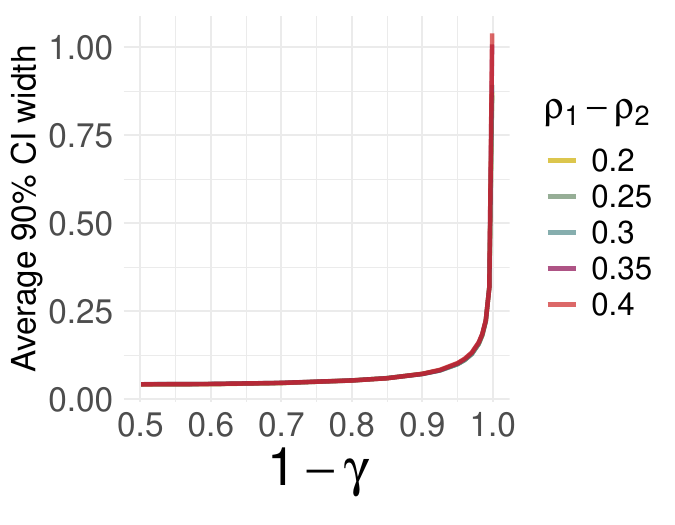}
    }
    \caption{Results for Bernoulli edges, averaged over 5,000 simulated networks. \textit{Left:} Empirical versus nominal coverage of the confidence intervals for $B_{11}(\Atr)$ (proposed approach targeting $\Phi_{11}(\Atr)$ as described in Proposition~\ref{prop:squiggle_estimation}), or $B_{11}(A)$ (naive approach as described in Supplement~\ref{appendix:naive-cis}), with $n = 200$, $\Ktr = 5$, $\rho_1 = 0.75$, $\rho_2 = 0.5$, and $\gamma = 0.25$ for the proposed approach. \textit{Center and Right:} Average adjusted Rand index between true and estimated communities, and average 90\% confidence interval width, as a function of $1-\gamma$, for the proposed approach on networks with $n = 200$, $K = 5$, $\Ktr = 5$, and $\rho_1 = 0.75$.}
    \label{fig:conf_width_rand_bernoulli}
\end{figure}

We also simulate data as in Setting~3 from Figure~\ref{fig:daniela} with $n=200$. Before simulating each network, we first draw $M_{ij} \ind \text{Uniform}(0,1)$. In light of Proposition~\ref{prop:taylor}(\ref{prop:taylor_a}), in this setting $\Phi_{11}(\Atr)$ will not approximate $B_{11}(\Atr)$ well. Figure~\ref{fig:conf_width_rand_bernoulli_maximum_heterogeneity} displays results averaged over 5,000 simulated networks. In the left panel of Figure~\ref{fig:conf_width_rand_bernoulli_maximum_heterogeneity}, we see that even in this unfavorable scenario, the empirical coverage of $B_{11}(\Atr)$ is still quite close to the nominal coverage.

\begin{figure}[!htbp]
    \centering
    \includegraphics[height=4.1cm]{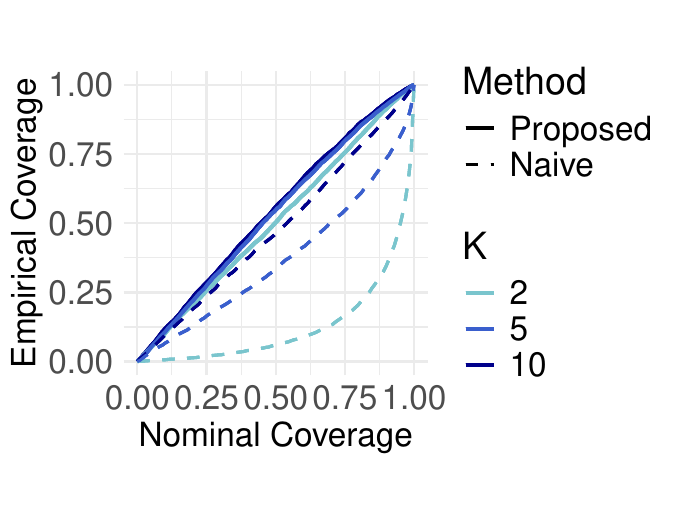}
    \includegraphics[height=4.1cm]{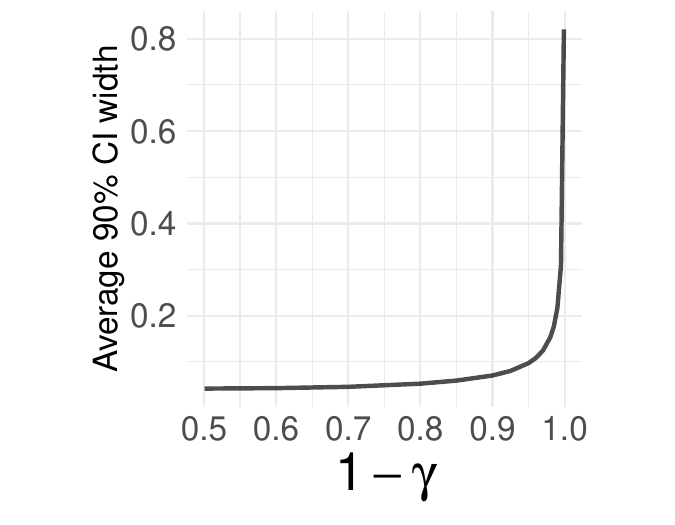}
    \caption{Results for Bernoulli edges, averaged over 5,000 simulated networks, where for each simulated network we draw $M_{ij} \ind \Unif(0,1)$. \textit{Left:} Empirical versus nominal coverage of the confidence intervals for $B_{11}(\Atr)$ (proposed approach targeting $\Phi_{11}(\Atr)$ as described in Proposition~\ref{prop:squiggle_estimation}), or $B_{11}(A)$ (naive approach as described in Supplement~\ref{appendix:naive-cis}), with $n = 200$, $\gamma = 0.25$, and where $K$ is varied. \textit{Right:} Average 90\% confidence interval width, as a function of $\gamma$, for the proposed approach on networks with $n=200$ and $K=5$.}
    \label{fig:conf_width_rand_bernoulli_maximum_heterogeneity}
\end{figure}

\section{Application to dolphin relationship network}
\label{sec:dolphins_application}

We apply the methods of this paper to a network of relationships among 62 bottlenose dolphins observed in Doubtful Sound in New Zealand \citep{lusseau2003bottlenose}. Relationships were observed among a closed population of dolphins living in a geographically isolated fjord at the southern extreme of the species' range. The data consists of an undirected Bernoulli network without self-loops, encoded as an adjacency matrix $A \in \{0,1\}^{62 \times 62}$ where $A_{ij} = A_{ji} = 1$ whenever the $i$th and $j$th dolphins were observed to consistently associate with each other over the study period. The adjacency matrix and a schematic of the corresponding network are displayed in Figures~\ref{fig:dolphins_cartoon_adjacency_matrix} and \ref{fig:dolphins_network_cartoon} respectively, in Supplement~\ref{appendix:dolphin_visuals}. The data is accessible in \texttt{R} through the \verb|manynet| package \citep{manynet2025}.

We apply the methods of this paper to this network to investigate whether the connectivity within the estimated communities exceeds the connectivity between them.

Recall from Proposition \ref{prop:univariate_bernoulli_fission} that in Bernoulli fission, the nonnegative parameter $\gamma \in (0, 0.5)$ trades off the information available for community estimation versus inference. For the observed network $A$, we define $\Atr_\gamma$ to be the train network arising from Bernoulli fission for a given $\gamma$. So, when $\gamma$ is small, more of the information in $A$ is allocated to $\Atr_\gamma$, and when $\gamma$ is closer to $0.5$, more information is allocated to $\Ate \mid \Atr_\gamma$, where $\Ate := A$.

For a range of values of $\gamma$, we apply spectral clustering \citep{amini2013PseudolikelihoodMethodsCommunity} to $\Atr_\gamma$ to estimate two communities, $\hat{Z}(\Atr_\gamma) \in \{0,1\}^{62 \times 2}$. The left-hand panel of Figure~\ref{fig:dolphins_fig} shows the adjusted Rand index averaged over 500 iterations of Bernoulli fission for each value of \(\gamma\), for the agreement between $\hat{Z}(\Atr_\gamma)$ and $\hat{Z}(A)$, the set of communities estimated from the original observed network. The agreement between $\hat{Z}(\Atr_\gamma)$ and $\hat{Z}(A)$ decreases when $\gamma$ increases, as less information is allocated to the train set $\Atr_\gamma$.

We define the selected parameter $\theta(\Atr_\gamma)$ to be the difference between the mean connectivity within the estimated communities and the mean connectivity between the estimated communities, which we can interpret as a measure of the absolute separation between the estimated communities. That is, recalling the definition of $B_{k \ell}(\Atr_\gamma)$ in \eqref{eq:B_entry_definition}, we set $\theta(\Atr_{\gamma}) := \left( B_{11}(\Atr_{\gamma}) + B_{22}(\Atr_{\gamma}) - 2 B_{12}(\Atr_{\gamma}) \right)/\sqrt{6}$.

So, if $\theta(\Atr_{\gamma}) = 0$, then the mean connectivities within and between the estimated communities are equal. As discussed in Section~\ref{sec:bernoulli_edges}, we cannot conduct inference for $\theta(\Atr_\gamma)$ directly, so recalling the definition of $\Phi_{k \ell}(\Atr_\gamma)$ in \eqref{eq:Phi_kl_def}, we instead target \\ $\xi(\Atr_{\gamma}) := \left( \Phi_{11}(\Atr_{\gamma}) + \Phi_{22}(\Atr_{\gamma}) - 2\Phi_{12}(\Atr_{\gamma}) \right)/\sqrt{6}$ using Proposition \ref{prop:undirected_squiggle_estimation} from Supplement~\ref{app:undirected} (a variant of Proposition \ref{prop:squiggle_estimation} for undirected Bernoulli networks without self-loops).

The right-hand panel of Figure \ref{fig:dolphins_fig} displays the midpoint, as well as the lower and upper bounds, of the $90\%$ confidence intervals for $\xi(\Atr_\gamma)$, averaged over 500 iterations of Bernoulli fission for each value of $\gamma$. When $\gamma$ is far from $0.5$ and slightly away from $0$, the confidence intervals do not contain zero and are positive. That is, there is evidence that the mean connectivity within exceeds the mean connectivity between the estimated communities. When $\gamma$ is near $0.5$, little information is allocated to $\Atr_\gamma$, and so community estimation is noisy, which leads to decreased separation between the estimated communities. On the other hand, when $\gamma$ is close to $0$, little information is allocated to $\Ate \mid \Atr_{\gamma}$, and so although the midpoints of the confidence intervals are far from $0$, the confidence intervals widen dramatically and contain $0$.

As in sample splitting, there is a trade-off between the information allocated to selection (here: community estimation) and inference. Obtaining confidence intervals that do not contain $0$ requires both good community recovery \textit{and} having sufficient remaining information for inference.

\begin{figure}[!htbp]
    \centering
    \fbox{
    \includegraphics[height=3.2cm]{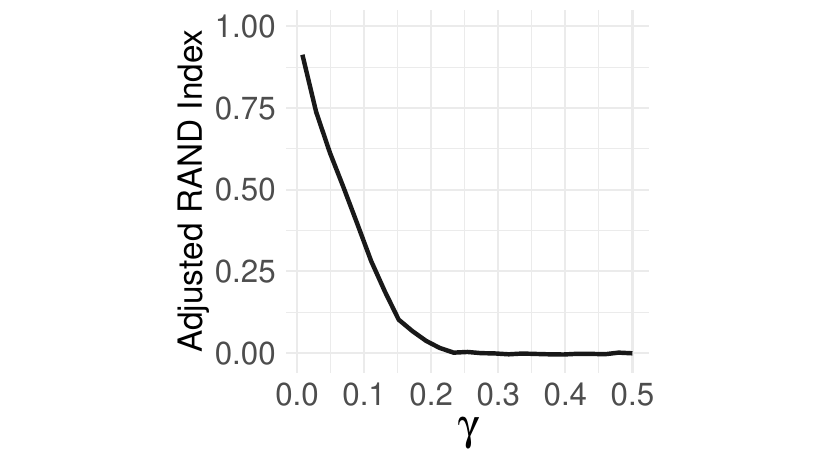}}
    \fbox{
    \includegraphics[height=3.2cm]{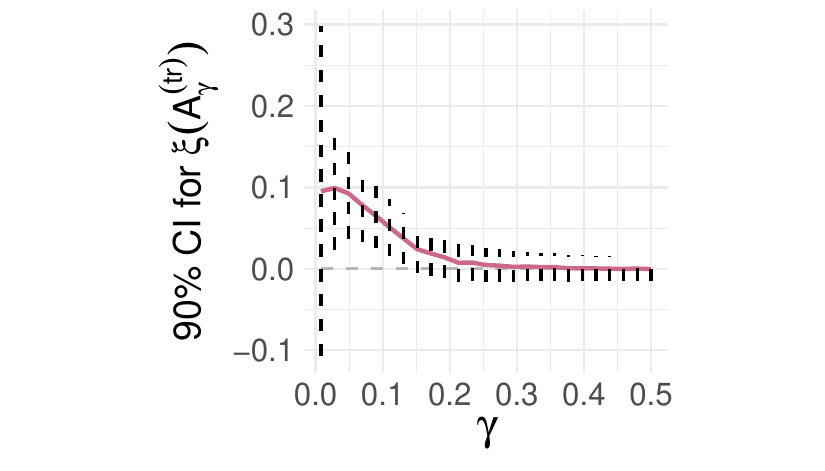}
    }
    \caption{Analysis of the dolphin relationship network as a function of $\gamma$, averaged over $500$ iterations of Bernoulli fission from Proposition~\ref{prop:univariate_bernoulli_fission} for each value of $\gamma$.  \emph{Left:} The adjusted Rand index of $\hat{Z}(\Atr_\gamma)$ compared to $\hat{Z}(A)$. As $\gamma$ increases, less information is allocated for $\Atr_\gamma$, so community estimation suffers. \emph{Right:} The midpoint (red line) and bounds (dashed lines) of a 90\% confidence interval for $\xi(\Atr_{\gamma})$ as a function of $\gamma$. \label{fig:dolphins_fig}}
\end{figure}

\section{Discussion} \label{sec:discussion} 

A primary challenge in contemporary data analysis pipelines is that of validating or conducting formal inference on data-driven parameters. In this work, we address that problem in the context of network models using ideas recently proposed in the selective inference literature: namely, data thinning and data fission. This yields valid inference on data-driven parameters in the presence of a single realization of a network. 

In contrast with much of the existing network literature, our proposed approach does not require the assumption of ``true'' communities. However, there is a catch: although  our proposed confidence intervals are guaranteed to attain the nominal selective coverage, those intervals are only as meaningful as the selected parameter, which is in turn determined by the estimated communities. Thus, the proposed approach is appealing only if the estimated communities are of interest.

Though the selected parameter that we consider is motivated by the SBM, we do not assume that the observed edges follow an SBM. Moreover, while we used the SBM as a working model, our approach is broadly applicable to conducting inference for data-driven network parameters. However,  like the data thinning and fission proposals on which it is based, we do require  that the edges are (i) independent, and (ii) members of a known distributional family. We leave a relaxation of these requirements to future work. 

An \texttt{R} package implementing the proposal in this paper, a tutorial illustrating its use, and scripts to reproduce all numerical results are available at \\
\href{https://ethanancell.github.io/networkinference/}{
  \if1\anon{https://ethanancell.github.io/networkinference/} \fi
\if0\anon{https://XXXXX.github.io/networkinference/}\fi}.





\clearpage
\appendix

\renewcommand{\theequation}{S\arabic{equation}}
\renewcommand{\theHequation}{\theequation}
\setcounter{equation}{0}

\renewcommand{\thefigure}{S\arabic{figure}}
\renewcommand{\theHfigure}{\thefigure}
\setcounter{figure}{0}

\renewcommand{\thetable}{S\arabic{table}}
\renewcommand{\theHtable}{\thetable}
\setcounter{table}{0}

\renewcommand{\thesection}{S\arabic{section}}
\renewcommand{\theHsection}{\thesection}
\setcounter{section}{0}

\renewcommand{\thesubsection}{S\arabic{section}.\arabic{subsection}}
\renewcommand{\theHsubsection}{\thesubsection}
\setcounter{subsection}{0}

\renewcommand{\thetheorem}{S\arabic{theorem}}
\renewcommand{\theHtheorem}{\thetheorem}
\setcounter{theorem}{0}

\renewcommand{\thelemma}{S\arabic{lemma}}

\renewcommand{\theproposition}{S\arabic{proposition}}

\section{Extensions to undirected networks and networks without self-loops}
\label{app:undirected}

We extend our results to networks that are undirected (with the convention that $A$ is an upper-triangular matrix), and networks without self loops (with the convention that $A$ contains zeroes along its diagonal), primarily by making careful changes to notation.

To modify the results of Section~\ref{subsec:thinning_fission}, we apply Propositions \ref{prop:univariate_gaussian_thinning}, \ref{prop:univariate_poisson_thinning}, and \ref{prop:univariate_bernoulli_fission}, but only for $(i,j) \in \mathcal{J}$ rather than all $(i,j) \in [n]^2$, where we define
\begin{equation}
    \mathcal{J} := \begin{cases}
        \left\{ (i,j) : i < j \right\}, & \text{for undirected networks with no self-loops}, \\
        \left\{ (i,j) : i \le j \right\}, & \text{for undirected networks with self-loops}, \\
        \left\{ (i,j) : i \ne j \right\}, & \text{for directed networks with no self-loops}.
    \end{cases}
    \label{eq:J_def}
\end{equation}

To modify the results of Section~\ref{sec:defining_target_of_inference}, we redefine $B = B(\Atr) \in \mathbb{R}^{K \times K}$ in \eqref{eq:B_matrix_theta_def}. For directed networks without self-loops, we define all entries $B_{k \ell}$ for $(k, \ell) \in [K]^2$, where for undirected networks (with or without self-loops) we define $B_{k \ell}$ only for $k \le \ell$ so that $B$ is a upper-triangular matrix. Then, we define $B_{k \ell} := \frac{1}{|\Ic'_{k \ell}|} \sum_{(i,j) \in \Ic'_{k \ell}} E[A_{ij}]$, where
\begin{equation}
    \Ickl' := \begin{cases}
        \{(i,j) \in \Ickl : i < j\}, & \text{for undirected networks with no self-loops}, \\
        \{(i,j) \in \Ickl : i \le j\}, & \text{for undirected networks with self-loops}, \\
        \{(i,j) \in \Ickl : i \ne j\}, & \text{for directed networks with no self-loops}. \\
    \end{cases}
    \label{eq:relevant_edge_def}
\end{equation}
Here, we recall from Section~\ref{sec:defining_target_of_inference} that $\Ic_{k \ell} := \left\{(i,j) ~:~ \Ziktr = 1, \Zjltr = 1 \right\}$ is the set of all edges originating in the $k$th estimated community, and terminating in the $\ell$th estimated community. Finally, our selected parameter is a linear combination of the entries of $B$:
\begin{equation}
    \theta(\Atr) := \begin{cases} \sum_{k = 1}^K \sum_{ \ell = 1}^K U_{k \ell} B_{k \ell}, & \text{for directed networks}, \\
    \sum_{k=1}^K \sum_{\ell \ge k}^K U_{k \ell} B_{k \ell}, & \text{for undirected networks},
    \end{cases}
    \label{eq:theta_undirected}
\end{equation}
where $U \in \mathbb{R}^{K \times K}$ satisfies $\sum_{k=1}^K \sum_{\ell = 1}^K U_{k \ell}^2 = 1$ and may depend on $\Atr$ if desired. When the network is undirected, $U$ (like $B$) must be an upper-triangular matrix.

Finally, we modify the results of Section~\ref{sec:inference_on_data_driven_parameter} by restating Propositions~\ref{prop:normal_estimation}, \ref{prop:poisson_estimation}, and \ref{prop:squiggle_estimation} as Propositions \ref{prop:undirected_normal_estimation}, \ref{prop:undirected_poisson_estimation}, and \ref{prop:undirected_squiggle_estimation}, respectively. The asymptotic arguments used to prove Propositions~\ref{prop:undirected_normal_estimation}, \ref{prop:undirected_poisson_estimation}, and \ref{prop:undirected_squiggle_estimation} are nearly identical to their counterparts in the main text. In what follows, let $\phi_{1-\alpha/2}$ denote the $(1-\alpha/2)$-quantile of the $\mathcal{N}(0,1)$ distribution.

\begin{proposition}
    \label{prop:undirected_normal_estimation} 
    Suppose that the random adjacency matrix $A$ has entries $A_{ij} \ind
    \mathcal{N}(M_{ij}, \tau^2)$ for $(i,j) \in \mathcal{J}$ with $\mathcal{J}$ defined in \eqref{eq:J_def}, where $\tau^2$ is a common known variance and the mean $M_{ij}$ is unknown. Suppose that we fix $\epsilon \in (0,1)$ and construct $\Ate$ and $\Atr$ from $A$ by applying Proposition~\ref{prop:univariate_gaussian_thinning} for $(i,j) \in \mathcal{J}$, and we then apply community detection to $\Atr$
    to yield the estimated community membership matrix $\Ztr \in \{0,1\}^{n
      \times K}$.
      
    Define $\hat{\theta}\left( \Ate, \Atr \right) := (1-\epsilon)^{-1} \sum_{k = 1}^K \sum_{\ell = 1}^K U_{k \ell} \hat{B}_{k \ell}$, where $U \in \mathbb{R}^{K \times K}$ satisfies \\$\sum_{k = 1}^K \sum_{\ell = 1}^K U_{k \ell}^2 = 1$, may depend on $\Atr$ if desired, and is upper-triangular if the network is undirected. Furthermore, define $\hat{B} \in \mathbb{R}^{K \times K}$ entry-wise as $\hat{B}_{k \ell} := \frac{1}{|\Ic'_{k \ell}|} \sum_{(i,j) \in \Ic'_{k \ell}} \Aijte$, and is upper-triangular if the network is undirected, and where $\Ickl'$ is defined in \eqref{eq:relevant_edge_def}.
    
    Then, $P \Big( \theta(\Atr) \in \Big[ \hat{\theta}(\Ate, \Atr) \pm \phi_{1-\alpha/2} \cdot \sigma \Big] ~\Big|~ \Atr \Big) = 1 - \alpha$, where $\theta \left( \Atr \right)$ was defined in \eqref{eq:theta_undirected}, $\sigma^2 := (1-\epsilon)^{-1} \tau^2 \sum_{k=1}^K \sum_{\ell=1}^K U_{k \ell}^2 \Delta_{k \ell}$, and $\Delta_{k \ell} := \dfrac{1}{|\Ic'_{k \ell}|}$.
\end{proposition}

For Poisson edges we arrive at the following asymptotic result.

 \begin{proposition} \label{prop:undirected_poisson_estimation} Consider a sequence of $n \times n$ random adjacency matrices $(A_n)_{n=1}^\infty$ with entries $A_{n,ij} \ind \textnormal{Poisson}(M_{n,ij})$ for $(i,j) \in \mathcal{J}_n$ with $\mathcal{J}_n$ defined in \eqref{eq:J_def}, where $0 < N_0 \le M_{n,ij} \le N_1 < \infty$ for constants $N_0$ and $N_1$ not depending on $n$. Fix $\epsilon \in (0,1)$, and construct $\Antr$ and $\Ante$ from $A_n$ by applying Proposition~\ref{prop:univariate_poisson_thinning} for $(i,j) \in \mathcal{J}$, and then apply community detection to $\Antr$ to yield the estimated community membership matrices $\Zntr \in \{0,1\}^{n \times K}$.
 
 Define $\hat{\theta}_n(\Ante, \Antr) := (1-\epsilon)^{-1} \sum_{k = 1}^K \sum_{\ell = 1}^K U_{n, k \ell} \hat{B}_{n, k \ell}$, where $U_n \in \mathbb{R}^{K \times K}$ satisfies $\sum_{k = 1}^K \sum_{\ell=1}^K U_{n, k \ell}^2 = 1$, may depend on $\Antr$ if desired, and is upper-triangular if the network is undirected. Furthermore, $\hat{B} \in \mathbb{R}^{K \times K}$ is defined entry-wise as $\hat{B}_{n, k \ell} := \frac{1}{|\Ic'_{n, k \ell}|} \sum_{(i,j) \in \Ic'_{n, k \ell}} \Anteij$, is upper-triangular if the network is undirected, and where $\Ickl'$ is defined in \eqref{eq:relevant_edge_def}. Additionally, define $\hat{\Delta}_n \in \mathbb{R}^{K \times K}$ with entries $\hat{\Delta}_{n, k \ell} := \dfrac{\hat{B}_{n, k \ell}}{|\Icnkl'|}$, and $\hat{\sigma}^2_n:= (1-\epsilon)^{-2} \sum_{k=1}^K \sum_{\ell=1}^K U_{n, k \ell}^2 \hat{\Delta}_{n, k \ell}$.
 
 Then, we have $\lim_{n \to \infty} P \Big( \theta_n(\Antr) \in \Big[ \hat{\theta}_n(\Ante, \Antr) \pm \phi_{1-\alpha/2} \cdot \hat{\sigma}_n \Big] ~\Big|~ \Antr \Big) = 1-\alpha$ for $\theta_n(\Antr)$ defined in \eqref{eq:theta_undirected}, provided that the sequence of realizations $\left\{\Zntr = \hat{z}_n \right\}_{n=1}^{\infty}$ is such that $(|\Icnkl'|)^{-1} = O(n^{-2})$ for all $k, \ell \in \{1, 2, \dots, K\}$.
\end{proposition}

Finally, we turn to the case of Bernoulli edges. Abbreviating the discussion in Section~\ref{sec:bernoulli_edges}, we redefine the selected parameter as
\begin{equation}
    \xi(\Atr) := \sum_{k=1}^K \sum_{\ell=1}^K U_{k \ell} \Phi_{k \ell}(\Atr),
    \label{eq:squiggle_undirected}
\end{equation}
where $U$ is as before, and $\Phi(\Atr) \in \mathbb{R}^{K \times K}$ is defined entry-wise as
\begin{equation*}
    \Phi_{k \ell}(\Atr) := \dfrac{|\Ickl'^{(0)}|}{|\Ickl'|} f \left( V_{k \ell}^{(0)}(\Atr), \dfrac{1-\gamma}{\gamma} \right) + \dfrac{|\Ickl'^{(1)}|}{|\Ickl'|} f \left( V_{k \ell}^{(1)}(\Atr), \dfrac{\gamma}{1-\gamma} \right),
\end{equation*}
where $f(a,v) := \textnormal{expit}(\textnormal{logit}(a) + \log(v))$, $\Ic'^{(s)}_{k \ell} := \left\{(i,j) \in \Ic'_{k \ell} : \Aijtr = s \right\}$ for $\Ickl'$ defined in \eqref{eq:relevant_edge_def}, $V^{(s)}_{k \ell} := \frac{1}{|\Ic'^{(s)}_{k \ell}|} \sum_{(i,j)} T_{ij}$ for $s = 0,1$, and where $T_{ij} := \E[A_{ij} \mid \Aijtr]$ is defined in \eqref{eq:annas_insight_1}. When the network is undirected, both $U$ and $\Phi$ must be upper-triangular matrices. Proposition~\ref{prop:undirected_squiggle_estimation} ensures that we can (asymptotically) estimate $\xi(\Atr)$ when the network is undirected or disallows self-loops.

\begin{proposition}
\label{prop:undirected_squiggle_estimation} Consider a sequence of $n \times n$ random adjacency matrices $(A_n)_{n=1}^\infty$, consisting of entries $A_{n,ij} \ind \operatorname{Bernoulli}(M_{n,ij})$ for $(i,j) \in \mathcal{J}_n$ with $\mathcal{J}_n$ defined in \eqref{eq:J_def}, where $0 < N_0 \le
    M_{n,ij} \le N_1 < 1$ holds for constants $N_0$ and $N_1$ not depending on $n$. Suppose
    that we fix $\gamma \in (0, 0.5)$ and construct $\Ante$ and $\Antr$ from $A$ by applying Proposition~\ref{prop:univariate_bernoulli_fission} for $(i,j) \in \mathcal{J}$, and then apply community detection to $\Antr$ to
    yield the estimated community membership matrix $\Zntr \in \{0,1\}^{n
      \times K}$.
      
    Define the estimator $\hat{\xi}_n(\Ante, \Antr) := \sum_{k=1}^K \sum_{\ell=1}^K U_{n, k \ell} \hat{\Phi}_{n, k \ell}(\Ante, \Antr)$, where $U_n \in \mathbb{R}^{K \times K}$ satisfies $\sum_{k =1}^K \sum_{\ell=1}^K U_{n, k \ell}^2 = 1$, is allowed to depend on $\Antr$ if desired, and is upper-triangular if the network is undirected.
    
    Furthermore, $\hat{\Phi}_n(\Ante, \Antr) \in \mathbb{R}^{K \times K}$ is defined entry-wise as \\ $\hat{\Phi}_{n, k \ell}(\Ante, \Antr) := \dfrac{|\Ic'^{(0)}_{n,k \ell}|}{|\Ic'_{n, k \ell}|} \hat{V}^{(0)}_{n, k \ell}(\Ante, \Antr) + \dfrac{|\Ic'^{(1)}_{n,k \ell}|}{|\Ic'_{n, k \ell}|} \hat{V}^{(1)}_{n, k \ell}(\Ante, \Antr)$, where $\hat{\Phi}_n$ is upper-triangular if the network is undirected, and
    where $\Icnkl'$ is defined in \eqref{eq:relevant_edge_def}, \\
    $\Icnkl'^{(s)} := \left\{ (i,j) \in \Icnkl' : \Antrij = s \right\}$ for $s \in \{0, 1\}$, $\hat{V}^{(s)}_{n, k \ell}(\Ante, \Antr) := \frac{\hat{B}^{(s)}_{n, k \ell}}{\hat{B}^{(s)}_{n, k \ell} + (1 - \hat{B}^{(s)}_{n, k \ell}) e^{c^{(s)}}}$ for $c^{(0)} := \log(\gamma/(1-\gamma))$ and $c^{(1)} := \log((1-\gamma)/\gamma)$, and $\hat{B}^{(s)}_{n, k \ell} := \frac{1}{|\Ic'^{(s)}_{n, k \ell}|} \sum_{(i,j) \in \Ic'^{(s)}_{n, k \ell}} \Anteij$.
    Also define $\hat{\sigma}_n^2 := \sum_{k=1}^K \sum_{\ell=1}^K U_{n, k \ell}^2 \hat{\Delta}_{n, k \ell}$, and $\hat{\Delta}_{n, k \ell} := \sum_{s \in \{0,1\}} \dfrac{|\Ic'^{(s)}_{n,k \ell}|^2}{|\Ic'_{n,k \ell}|^2} \hat{\Delta}^{(s)}_{n, k \ell}$, where \\ $\hat{\Delta}^{(s)}_{n, k \ell} := \frac{\hat{B}^{(s)}_{n,k \ell} (1-\hat{B}^{(s)}_{n,k \ell}) e^{2 c^{(s)}}}{|\Ic'^{(s)}_{n,k \ell}|((1-\hat{B}^{(s)}_{n,k \ell}) e^{c^{(s)}} + \hat{B}^{(s)}_{n,k \ell})^4}$. Then, for $\xi_n(\Antr)$ defined in \eqref{eq:squiggle_undirected}, we have
    \begin{align*}
        \liminf_{n \to \infty} P \Bigg( \xi(\Antr) \in \Big[ \hat{\xi}(\Ante, \Antr) \pm \phi_{1-\alpha/2} \cdot \hat{\sigma}_n \Big] ~\Bigg|~ \Antr \Bigg) \ge 1-\alpha,
    \end{align*}
    provided that the sequence of realizations $\left\{\Antr = \antr \right\}_{n=1}^\infty$ and $\left\{\Zntr = \hat{z}_n \right\}_{n=1}^\infty$ are such that $|\Ic'^{(0)}_{n,kl}|^{-1} = O(n^{-2})$ and $|\Ic'^{(1)}_{n,kl}|^{-1} = O(n^{-2})$ for all $k, \ell \in \{1, 2, \dots, K\}$.
\end{proposition}

\section{Additional simulation details and results}
\label{appendix:sim-details}

\subsection{Additional simulations}
\label{appendix:extra-simulations}

In the same setting as the left-hand panels of Figures~\ref{fig:conf_width_rand_gaussian} and \ref{fig:conf_width_rand}, we vary the value of $n \in\{100, 200, 500\}$ and investigate the empirical coverage of the 90\% confidence intervals from the proposed and naive methods, respectively. Table \ref{tab:coverage_naive_thinning_poisson_gaussian} displays the coverage of their respective selected parameters, with $(\rho_1, \rho_2) = (30, 27)$, and where $\epsilon = 0.5$ in the proposed approach and $\tau^2 = 25$ for networks with Gaussian edges. The proposed approach empirically achieves the 90\% nominal coverage rate for $B_{11}(\Atr)$. In contrast, the naive approach (see Supplement~\ref{appendix:naive-cis}) severely under-covers $B_{11}(A)$.

\begin{table}[!htbp]
    \centering
    \begin{tabular}{|c|c|c|c|c||c|c|c|}
    \hline
    \multicolumn{2}{|c|}{} & \multicolumn{3}{c||}{\textbf{Proposed Approach}} & \multicolumn{3}{c|}{\textbf{Naive Approach}} \\
        \cline{3-8}
        \multicolumn{2}{|c|}{} & $K=2$ & $K=5$ & $K=10$ & $K=2$ & $K=5$ & $K=10$ \\
        \hline
        \multirow{3}{*}{\textbf{Gaussian}} & $n=100$ & 89.80 & 90.02 & 89.90 & 41.42 & 53.44 & 71.60 \\
        & $n=200$ & 90.22 & 90.22 & 90.78 & 16.06 & 38.68 & 48.82 \\
        & $n=500$ & 90.18 & 90.68 & 91.02 & 36.28 & 89.40 & 69.76 \\
                \hline
        \multirow{3}{*}{\textbf{Poisson}} & $n=100$ & 89.64 & 90.04 & 89.96 & 42.75 & 53.04 & 72.20 \\
        & $n=200$ & 90.70 & 89.38 & 89.78 & 15.38 & 38.94 & 47.48 \\
        & $n=500$ & 89.84 & 90.08 & 89.16 & 34.00 & 89.06 & 65.92 \\
        \hline
    \end{tabular}
    \caption{Empirical coverage (as a percentage) of 90\% confidence intervals  for Gaussian and Poisson networks arising from the proposed approach for $B_{11}(\Atr)$ with $\epsilon=0.50$, and the naive method for $B_{11}(A)$, in a setting with $\Ktr=5$ and Gaussian or Poisson edges.
    }
    \label{tab:coverage_naive_thinning_poisson_gaussian}
\end{table}

Next, we return to the simulation setting of the left-hand panel of Figure~\ref{fig:conf_width_rand_bernoulli} and vary the value of $n \in \{100, 200, 500\}$, and we show the coverage of the proposed and naive methods in Table \ref{tab:coverage_naive_fission_bernoulli}. For the proposed method, we report the coverage of both $\Phi_{11}(\Atr)$ and $B_{11}(\Atr)$ by the confidence intervals targeting $\Phi_{11}(\Atr)$. Although the proposed approach should over-cover $\Phi_{11}(\Atr)$ (see Proposition~\ref{prop:squiggle_estimation}), the empirical coverage is very close to the nominal coverage. Notably, the 90\% confidence intervals targeting $\Phi_{11}(\Atr)$ also contain $B_{11}(\Atr)$ with probability near 90\%, providing empirical evidence that $\Phi_{11}(\Atr)$ and $B_{11}(\Atr)$ are nearly the same, as suggested by Proposition~\ref{prop:taylor}. By contrast, the naive approach (see Supplement~\ref{appendix:naive-cis}) does not achieve the nominal 90\% coverage rate for $B_{11}(A)$. 

\begin{table}[!htbp]
    \centering
    \begin{tabular}{|c||c|c|c|c||c|c|c|c|}
    \hline
    \multicolumn{2}{|c|}{} & \multicolumn{3}{c||}{\textbf{Proposed Approach}} & & \multicolumn{3}{c|}{\textbf{Naive Approach}} \\
    \cline{3-5} \cline{7-9}
        \multicolumn{2}{|c|}{} & $K=2$ & $K=5$ & $K=10$ & & $K=2$ & $K=5$ & $K=10$ \\
        \hline
        $n = 100$ & \multirow{3}{*}{$B_{11}(\Atr)$} & 91.12 & 90.08 & 89.94 & \multirow{3}{*}{$B_{11}(A)$} & 29.20 & 54.06 & 71.94 \\
        $n = 200$ & & 89.66 & 90.46 & 89.60 & & 15.26 & 35.90 & 48.32 \\
        $n=500$ & & 90.04 & 89.96 & 90.74 & & 18.96 & 85.24 & 62.82 \\
        \hline
        $n=100$ & \multirow{3}{*}{$\Phi_{11}(\Atr)$} & 91.06 & 90.12 & 89.96 & \multirow{3}{*}{-} & - & - & - \\
        $n=200$ & & 89.94 & 90.52 & 89.66 & & - & - & - \\
        $n=500$ & & 90.04 & 90.28 & 90.64 & & - & - & - \\
        \hline
    \end{tabular}
    \caption{Percent empirical coverage of 90\% confidence intervals from the proposed and naive approaches for Bernoulli networks, averaged over 5,000 simulated datasets, with $\gamma = 0.25$, $\Ktr = 5$, and $(\rho_1, \rho_2) = (0.75, 0.50)$. In the left-hand set of columns, we computed  90\% confidence intervals targeting $\Phi_{11}(\Atr)$, and report the proportion of these intervals that contain $\Phi_{11}(\Atr)$ as well as the proportion that contain $B_{11}(\Atr)$. The intervals targeting $\Phi_{11}(\Atr)$ achieve approximately 90\% coverage for each of  $\Phi_{11}(\Atr)$ and $B_{11}(\Atr)$, in keeping with Proposition \ref{prop:taylor}(b)'s assurance that the two quantities are approximately equal. In the right-hand set of columns, we computed naive 90\% confidence intervals for $B_{11}(A)$ that use the same data both to estimate communities and to test them; details for the naive confidence intervals are provided in Supplement~\ref{appendix:naive-cis}.
    }
    \label{tab:coverage_naive_fission_bernoulli}
\end{table}

Finally, for networks with Gaussian and Poisson edges, we simulate 5,000 networks with $n = 200$ nodes, where we first draw $M_{ij} \ind \Unif(0, 20)$. In this simulation setting, there are no true communities. The left-hand panels of Figures~\ref{fig:conf_width_rand_gaussian_maximum_heterogeneity} and \ref{fig:conf_width_rand_poisson_maximum_heterogeneity} show the empirical and nominal coverages of the proposed and naive methods as we vary the value of $K \in \{2, 5, 10\}$ for Gaussian and Poisson edges, respectively. The left-hand panels show that even when no communities exist in the underlying data, the proposed approach still achieves valid coverage for the selected parameters. The right-hand panels of Figures~\ref{fig:conf_width_rand_gaussian_maximum_heterogeneity} and \ref{fig:conf_width_rand_poisson_maximum_heterogeneity} show the average 90\% confidence interval width as we vary the value of $\epsilon$, showing that as $\epsilon$ increases, less information is allocated to $\Ate$, and so confidence intervals widen.

\begin{figure}[!htbp]
    \centering
    \includegraphics[height=4.2cm]{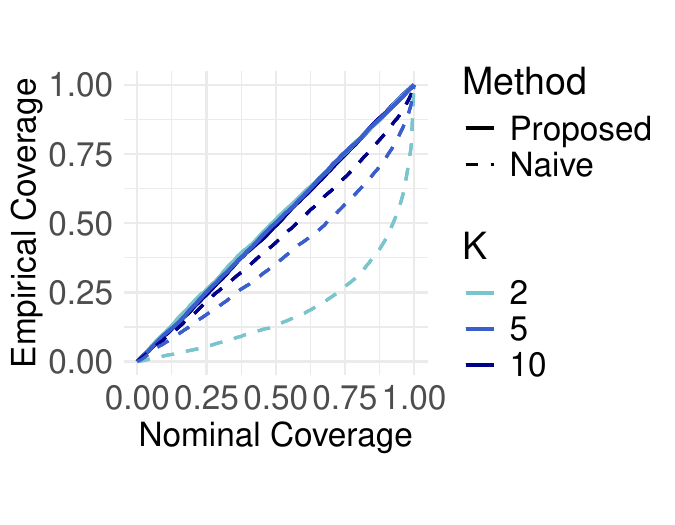}
    \includegraphics[height=4.0cm]{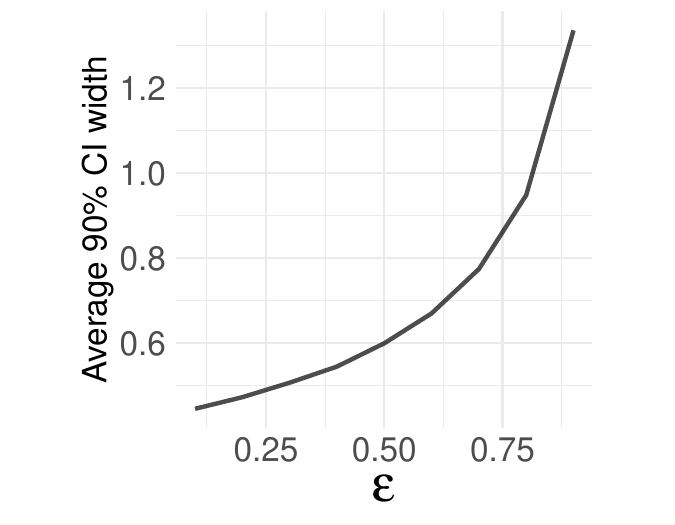}
    \caption{Results for networks with Gaussian edges, averaged over 5,000 simulated networks, where for each simulated network we set $n=200$ and draw $M_{ij} \ind \Unif(0,20)$, and set $\tau^2 = 25$. \textit{Left:} Empirical versus nominal coverage of the confidence intervals for $B_{11}(\Atr)$ (proposed approach), or $B_{11}(A)$ (naive approach as described in Supplement~\ref{appendix:naive-cis}), with $\epsilon = 0.50$, and where $K$ is varied. \textit{Right:} Average 90\% confidence interval width, as a function of $\epsilon$, for the proposed approach on networks with $K=5$.}
    \label{fig:conf_width_rand_gaussian_maximum_heterogeneity}
\end{figure}

\begin{figure}[!htbp]
    \centering
    \includegraphics[height=4.2cm]{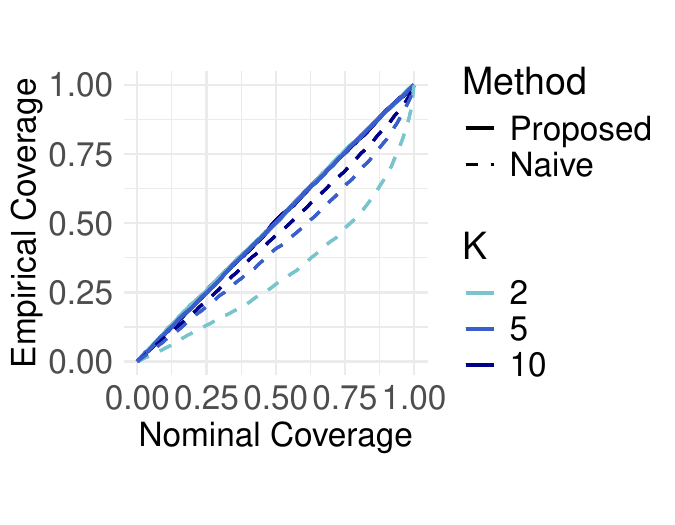}
    \includegraphics[height=4.0cm]{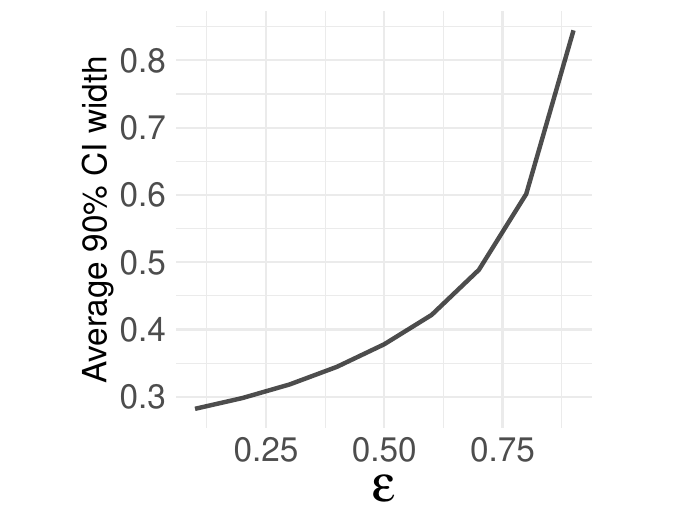}
    \caption{Results for Poisson edges. All other details are the same as Figure~\ref{fig:conf_width_rand_gaussian_maximum_heterogeneity}.}
    \label{fig:conf_width_rand_poisson_maximum_heterogeneity}
\end{figure}

\subsection{Simulation parameters}
\label{appendix:simulation-parameters}

This section details the simulation parameters that were used in the generation of Figure~\ref{fig:daniela}, and the tables and figures in Section~\ref{sec:simulation} and Supplement~\ref{appendix:extra-simulations}.

\paragraph{Figure~\ref{fig:daniela}} In all simulation settings, we fix $n = 100$, and evaluate behavior at $\gamma \in \{0.001, 0.1, 0.2, 0.3, 0.4, 0.5\}$. In the first simulation setting, we set $M_{ij} = 0.5$ for all $(i,j) \in [n]^2$. In the second simulation setting, nodes $1, \dots, 50$ belong to one community, and nodes $51, \dots 100$ belong to a second community. Then, $M_{ij} = 0.6$ if nodes $i$ and $j$ are in the same community, and $M_{ij} = 0.4$ if nodes $i$ and $j$ are in different communities. In the third simulation setting, for each simulation repetition we draw $M_{ij} \ind \Unif(0,1)$. 

\paragraph{Figures \ref{fig:conf_width_rand_gaussian} and \ref{fig:conf_width_rand} - left panel} For Gaussian and Poisson networks, we fix $n=200$, $\rho_1 = 30$, $\rho_2 = 27$, $\Ktr = 5$, $\epsilon = 0.5$, and $\tau^2 = 25$ (Gaussian edges only). Then, for $K \in \{2, 5, 10\}$, we average the results across 5,000 repetitions.

\paragraph{Figures \ref{fig:conf_width_rand_gaussian} and \ref{fig:conf_width_rand} - center and right panels} For Gaussian and Poisson networks, we fix $n = 200$, $\rho_1 = 30$, $\Ktr = 5$, $K = 5$, and $\tau^2 = 25$ (Gaussian edges only). Then, for $\epsilon \in \{0.1, ~0.2, ~0.3, ~0.4, ~0.5, ~0.6, ~0.7, ~0.8, ~0.9\}$ and $\rho_2 \in \{21, ~23, ~25, ~27, ~29\}$, we average the results across 5,000 repetitions.

\paragraph{Figure \ref{fig:conf_width_rand_bernoulli} - left panel} For Bernoulli networks, we fix $n = 200$, $\rho_1 = 0.75$, $\rho_2 = 0.5$, $\Ktr = 5$, and $\gamma = 0.25$. Then, for $K \in \{2, 5, 10\}$ we average the results across 5,000 repetitions.

\paragraph{Figure \ref{fig:conf_width_rand_bernoulli} - center and right panels} For Bernoulli networks, we fix $n = 200$, $\rho_1 = 0.75$, $\Ktr = 5$, and $K = 5$.

Then, for $\gamma \in \{0.001, ~0.005, ~0.01, ~0.015, ~0.02, \\ ~0.03, ~0.04, ~0.05,$$~0.075, ~0.10, ~0.15, ~0.20, ~0.30, ~0.40, ~0.499\}$ and $\rho_2 \in \{0.35, ~0.40, ~0.45, ~0.50, ~0.55\}$, we average the results across 5,000 repetitions.

\paragraph{Figure~\ref{fig:conf_width_rand_bernoulli_maximum_heterogeneity} - left panel} For Bernoulli networks, instead of the simulation setup in Section~\ref{subsec:simulation_methods}, we first draw $M_{ij} \ind \Unif(0,1)$, and fix $n = 200$ and $\gamma = 0.5$. Then, for $K \in \{2, 5, 10\}$, we average the results across 5,000 repetitions.

\paragraph{Figure~\ref{fig:conf_width_rand_bernoulli_maximum_heterogeneity} - right panel} For Bernouli networks, instead of the simulation setup described in Section~\ref{subsec:simulation_methods}, we first draw $M_{ij} \ind \Unif(0,1)$, and fix $n = 200$ and $K = 5$. Then, for $\gamma \in \{0.001, 0.005, 0.01, 0.015, 0.02$, $0.03, 0.04, 0.05, 0.075, 0.1, 0.15, 0.2, 0.3, 0.4, 0.499\}$, we average the results across 5,000 repetitions.

\paragraph{Table \ref{tab:coverage_naive_thinning_poisson_gaussian}} For Gaussian and Poisson networks, we fix $\rho_1 = 30$, $\rho_2 = 27$, $\Ktr = 5$, $\epsilon = 0.5$, and $\tau^2 = 25$ (Gaussian edges only). Then, for $n \in \{100, 200, 500\}$ and $K \in \{2, 5, 10\}$, we average the results across 5,000 repetitions.

\paragraph{Table \ref{tab:coverage_naive_fission_bernoulli}} For Bernoulli networks, we fix $\rho_1 = 0.75$, $\rho_2 = 0.5$, $\Ktr = 5$, and $\gamma = 0.25$. Then, for $n \in \{100, 200, 500\}$ and $K \in \{2, 5, 10\}$, we average the results across 5,000 repetitions.

\paragraph{Figures~\ref{fig:conf_width_rand_gaussian_maximum_heterogeneity} and \ref{fig:conf_width_rand_poisson_maximum_heterogeneity} - left panel} For Gaussian and Poisson networks, instead of the simulation setup in Section~\ref{subsec:simulation_methods}, we first draw $M_{ij} \ind \Unif(0,20)$, and fix $n = 200$, $\epsilon = 0.5$, and $\tau^2 = 25$ (Gaussian edges only). Then, for $K \in \{2, 5, 10\}$, we average the results across 5,000 repetitions.

\paragraph{Figures~\ref{fig:conf_width_rand_gaussian_maximum_heterogeneity} and \ref{fig:conf_width_rand_poisson_maximum_heterogeneity} - right panel} For Gaussian and Poisson networks, instead of the simulation setup in Section~\ref{subsec:simulation_methods}, we first draw $M_{ij} \ind \Unif(0,20)$, and fix $n = 200$, $K = 5$, $\epsilon = 0.5$, and $\tau^2 = 25$ (Gaussian edges only). Then, for $\epsilon \in \{0.10, 0.20, 0.30, 0.40, 0.50, 0.60, 0.70, 0.80, 0.90\}$, we average the results across 5,000 repetitions.

\subsection{Numerical considerations in Proposition~\ref{prop:squiggle_estimation}}
\label{appendix:numerical_issues_squiggle_estimation}

In practice, when $\gamma$ is close to $0$ (and especially if $n$ is relatively small), direct application of Proposition~\ref{prop:squiggle_estimation} may lead to numerical issues that can be readily addressed.

First, the variance of $\hat{V}^{(s)}_{n, k \ell}$ cannot exceed $0.25$ (the maximum variance for any distribution supported in $[0, 1]$), and so we set the variance estimate $\hat{\Delta}^{(s)}_{n, k \ell}$ to $0.25$ when its initial calculation exceeds $0.25$.

Second, when $\gamma$ is close to $0$, $\hat{B}_{n, k \ell}^{(s)}$ may be exactly $0$ or $1$, leading to $\hat{\Delta}_{n, k \ell}^{(s)} = 0$. When this happens, we can either replace $\hat{\Delta}^{(s)}_{n, k \ell}$ with the conservative value of $0.25$, or else re-compute $\hat{\Delta}^{(s)}_{n, k \ell}$ with the occurrences of $\hat{B}_{n, k \ell}^{(s)}$ appearing in the numerator of the expression for $\hat{\Delta}^{(s)}_{n, k \ell}$ replaced by a small constant (e.g., $\eta=10^{-8}$ or $\eta^{(s)}_{n, k \ell} = \frac{1}{2|\Ic'^{(s)}_{n, k \ell}|}$). We use the latter approach in our software implementation, but either of these approaches maintains the valid (conservative) coverage guarantee given in Proposition~\ref{prop:squiggle_estimation}.

\subsection{Construction of naive confidence intervals for $\theta(A)$}
\label{appendix:naive-cis}

In Figures~\ref{fig:conf_width_rand_gaussian}, \ref{fig:conf_width_rand}, \ref{fig:conf_width_rand_bernoulli}, \ref{fig:conf_width_rand_bernoulli_maximum_heterogeneity}, \ref{fig:conf_width_rand_gaussian_maximum_heterogeneity}, \ref{fig:conf_width_rand_poisson_maximum_heterogeneity}, and Tables~\ref{tab:coverage_naive_thinning_poisson_gaussian} and \ref{tab:coverage_naive_fission_bernoulli}, we compare our proposed approach to naive confidence intervals that arise when the same network $A$ is used to both select as well as conduct inference on the selected parameters, without accounting for the double use of data. Here, we provide the details of the naive method.

Recall from Section~\ref{sec:gaussian-poisson} that for networks with Gaussian or Poisson edges, we conduct inference for $\theta(\Atr)$ defined in \eqref{eq:target_of_inference}. We define the naive selected parameter
\begin{equation}
    \theta(A) := u^\top \textnormal{vec} \left( \left( \hat{Z}^\top \hat{Z} \right)^\top \hat{Z}^\top \E[A] \hat{Z} \left( \hat{Z}^\top \hat{Z} \right)^\top \right),
    \label{eq:naive_target_appendix}
\end{equation}
where $\theta(A)$ depends on $A$ through estimated communities $\hat{Z} = \hat{Z}(A)$, and also through $u = u(A)$ where $\Vert u \Vert_2 = 1$.

For Bernoulli networks, recall from Section~\ref{sec:bernoulli_edges} that we conduct inference for $\xi(\Atr)$ in \eqref{eq:xi}. However, there is no analogue of $\xi(\Atr)$ for naive inference, because its construction depends on $\Atr$. So, for Bernoulli edges, our naive selected parameter is $\theta(A)$ as in \eqref{eq:naive_target_appendix}.

Thus, for networks with Gaussian, Poisson, or Bernoulli edges, we construct naive confidence intervals for $\theta(A)$ as $\hat{\theta}(A) \pm \phi_{1-\alpha/2} \cdot \hat{\sigma}$, where $\phi_{1-\alpha/2}$ is the $1-\frac{\alpha}{2}$ quantile of the $\mathcal{N}(0,1)$ distribution, and where $\hat{\theta}(A)$ is defined as $\hat{\theta}(A) := u^\top \textnormal{vec} \left( \left( \hat{Z}^\top \hat{Z} \right)^\top \hat{Z}^\top A \hat{Z} \left( \hat{Z}^\top \hat{Z} \right)^\top \right)$. For networks with Gaussian edges, we construct $\hat{\sigma}^2$ as 
$\hat{\sigma}^2 := \tau^2 u^\top \left( \left( \hat{Z}^\top \hat{Z} \right)^{-1} \otimes \left( \hat{Z}^\top \hat{Z} \right)^{-1} \right) u$, where $\otimes$ is the Kronecker product.

For both Poisson and Bernoulli edges, define $\Ic_{k \ell} := \{(i,j) : \hat{Z}_{i k} = 1, \hat{Z}_{j \ell} = 1\}$, $\hat{B}_{k \ell} := \frac{1}{|\Ic_{k \ell}|} \sum_{(i,j) \in \Ic_{k \ell}} A_{ij}$, and define $\hat{\sigma}^2 := u^\top \textnormal{diag} \left( \textnormal{vec} \left( \hat{\Delta} \right) \right) u$, where $\hat{\Delta} \in \mathbb{R}^{K \times K}$. For Poisson edges, we define $\hat{\Delta}_{k \ell} := \dfrac{\hat{B}_{k \ell}}{|\Ickl|}$, and for Bernoulli edges we define $\hat{\Delta}_{k \ell} := \frac{\hat{B}_{k \ell} (1 - \hat{B}_{k \ell})}{|\Ic_{k \ell}|}$.

\section{Visualizations of dolphin relationship network}
\label{appendix:dolphin_visuals}

Figures~\ref{fig:dolphins_cartoon_adjacency_matrix} and \ref{fig:dolphins_network_cartoon} provide visualizations of the adjacency matrix and undirected network for the dolphin relationship data of \citet{lusseau2003bottlenose}, discussed in Section~\ref{sec:dolphins_application}.

\begin{figure}[!htpb]
    \centering

    \begin{subfigure}[t]{0.48\textwidth}
        \centering
        \includegraphics[height=3.0in]{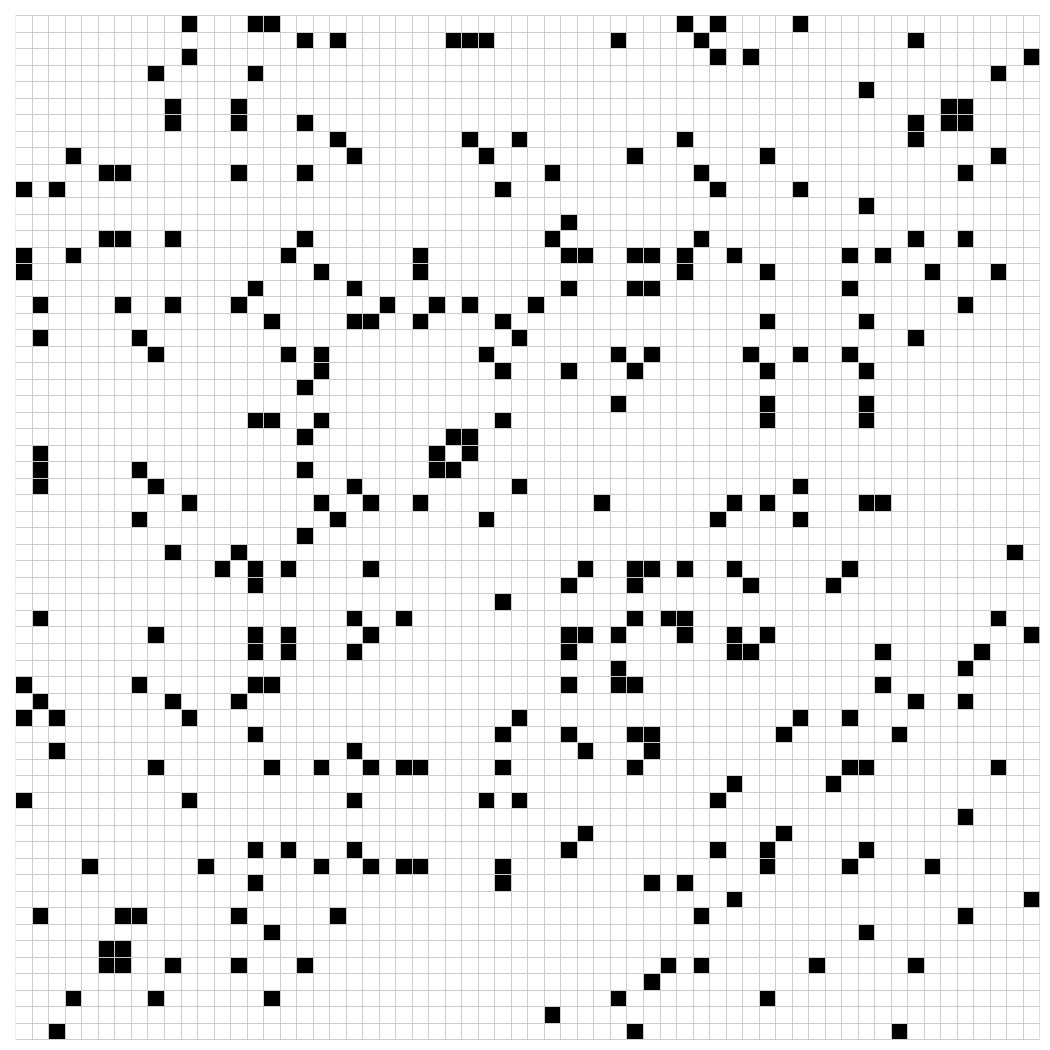}
        \subcaption{The adjacency matrix representing the relationships among dolphins from the study of \citet{lusseau2003bottlenose}, where the dolphins are arbitrarily ordered from 1 to 62. A black cell in the $i$th row and $j$th column indicates that the $i$th and $j$th dolphins interacted with each other consistently over the study period.}
        \label{fig:dolphins_cartoon_adjacency_matrix}
    \end{subfigure}%
    \hfill
    \begin{subfigure}[t]{0.48\textwidth}
        \centering
        \includegraphics[height=3.0in]{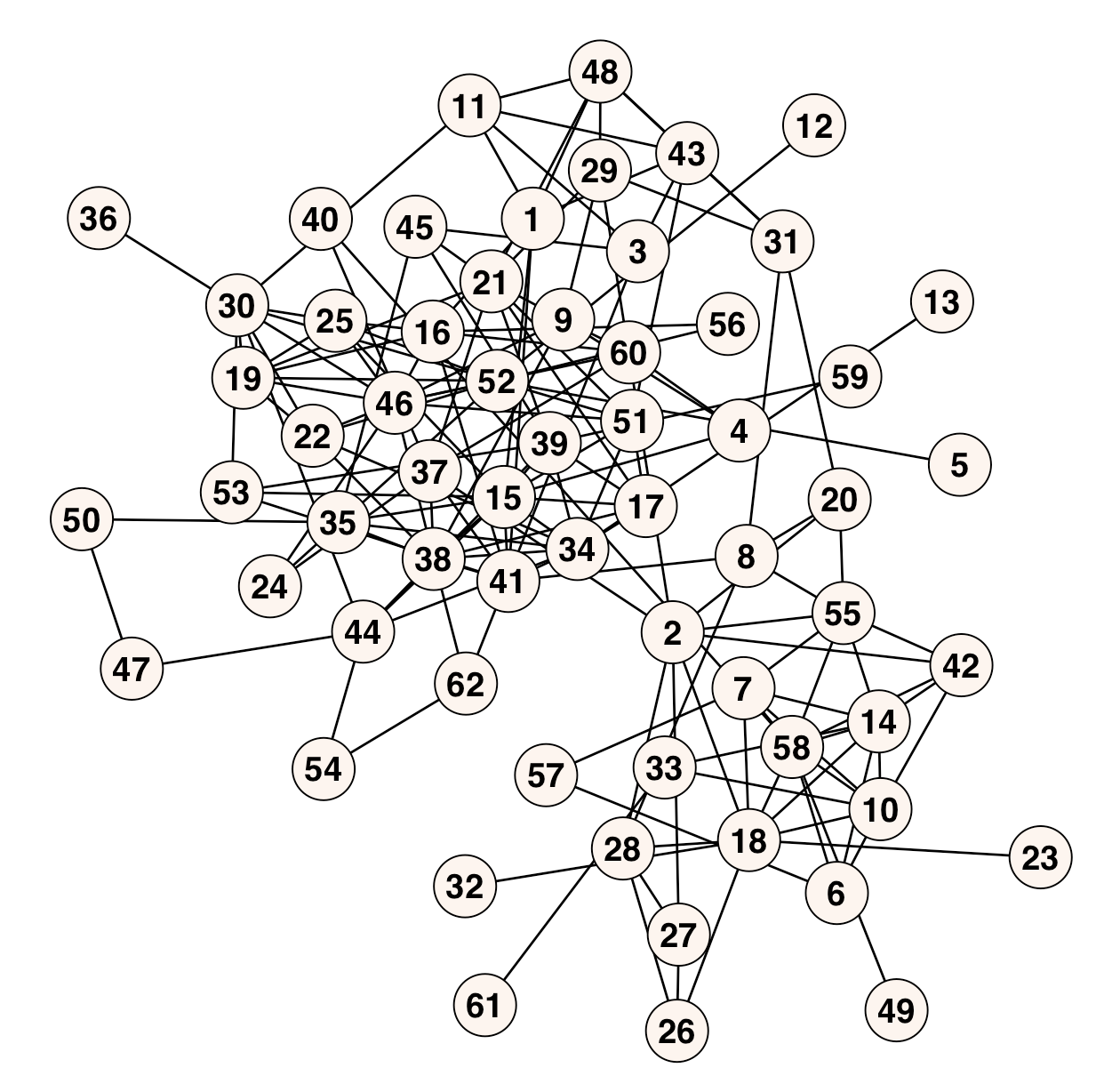}
        \subcaption{A visual representing the relationships between dolphins (represented as nodes in the graph) from the study of \citet{lusseau2003bottlenose}, where the dolphins are arbitrarily ordered from 1 to 62. A solid line between nodes indicates that those dolphins consistently interacted with each other over the study period.}
        \label{fig:dolphins_network_cartoon}
    \end{subfigure}

    \caption{Relationship representations for dolphins in \citet{lusseau2003bottlenose}.}
    \label{fig:dolphins_overall}
\end{figure}



\section{Proofs of theoretical results}
\label{appendix:all-proofs}

\subsection{Machinery for network asymptotics}
\label{subsec:network_asymptotics}

Our proofs of Propositions \ref{prop:normal_estimation},
\ref{prop:poisson_estimation}, and \ref{prop:squiggle_estimation} rely on many
common results; for convenience we establish them in this section. Lemmas
\ref{lemma:mixing_wlln} and \ref{lemma:mixing_clt} extend the weak law of large
numbers and the Lindeberg-Feller central limit theorem, respectively, to the setting
of averages over fixed subsets, where the sizes of the subsets grow at a sufficiently fast
rate. Lemma \ref{lemma:marginal_independence_joint_convergence} is a technical tool used in Propositions \ref{prop:poisson_estimation} and \ref{prop:squiggle_estimation} that allows one to combine multiple convergence results each applying to a single community pair $(k, \ell) \in \{1, 2, \dots, K\}^2$ into a single joint convergence result. Lemma \ref{lemma:alternative_cmt} is a specialization of the continuous mapping theorem used in the proof of Proposition~\ref{prop:squiggle_estimation}. Lemmas \ref{lemma:gaussian_rotation} and
\ref{lemma:asymptotic_variance_un_equal} are technical tools used in the proofs of Propositions~\ref{prop:poisson_estimation} and \ref{prop:squiggle_estimation} that allow us to incorporate the linear combination vector $u$ and establish the validity of our estimate of the estimator's variance.

\begin{lemma}[A weak law of large numbers over subsets]
    \label{lemma:mixing_wlln}
    Consider a pyramidal array $(Y_{n,ij})_{1 \le i,j \le n < \infty}$ of
    random variables such that those within a slice $Y_n = (Y_{n,ij})_{1 \le i,j
      \le n}$ are all mutually independent, but not necessary identically
    distributed, with a uniformly bounded variance $\Var(Y_{n,ij}) \le
    L_0 < \infty$ where $L_0$ does not depend on $n$, $i$, or $j$. Then, consider a $\{0, 1\}$-valued
    non-random pyramidal array $(w_{n,ij})_{1 \le i,j \le n < \infty}$, with a
    corresponding sequence of induced index sets $\Ic_n := \{(i,j) : w_{n,ij} =
    1\}$, and assume that $\lim_{n \to \infty} |\Ic_n|^{-1} = 0$. Then, $\dfrac{1}{|\Ic_n|} \sum_{(i,j) \in \Ic_n} \left( Y_{n,ij} -  \E[Y_{n,ij} ] \right) \overset{\text{p}}{\rightarrow} 0$.
\end{lemma}
\begin{proof}
    Note that $\dfrac{1}{|\Ic_n|} \sum_{(i,j) \in \Ic_n} (Y_{n,ij} - \E[Y_{n,ij}]) = \dfrac{1}{|\Ic_n|} \sum_{i=1}^{n} \sum_{j=1}^n (w_{n,ij} Y_{n,ij} - w_{n,ij} \E[Y_{n,ij}])$. Then, for any $\delta > 0$, by Chebyshev's inequality we have
    \begin{align*}
        P \Bigg( \Bigg| &\dfrac{1}{|\Ic_n|} \sum_{i=1}^{n} \sum_{j=1}^n (w_{n,ij} Y_{n,ij} - w_{n,ij} \E[Y_{n,ij}]) \Bigg| > \delta \Bigg) \le \delta^{-2} \Var \left( \dfrac{1}{|\Ic_n|} \sum_{i=1}^{n} \sum_{j=1}^n w_{n,ij} Y_{n,ij} \right) \\
        &= (|\Ic_n| \delta)^{-2} \sum_{i=1}^{n} \sum_{j=1}^n w_{n,ij} \Var(Y_{n,ij}) \le |\Ic_n|^{-1} (\delta^{-2} L_0).
    \end{align*}
    As $\lim_{n \to \infty} |\Ic_n|^{-1} = 0$, we have $\lim_{n \to \infty} P \Big( \Big| \dfrac{1}{|\Ic_n|} \sum_{(i,j) \in \Ic_n} (Y_{n,ij} - \E[Y_{n,ij}]) \Big| > \delta \Big) = 0$.
\end{proof}

\begin{lemma}[A Central Limit Theorem for taking averages over subsets]
    \label{lemma:mixing_clt}
    Consider a pyramidal array $(Y_{n,ij})_{1 \le i,j \le n < \infty}$ of random variables such that those within a slice $(Y_{n,ij})_{1 \le i,j \le n}$ are all mutually independent, but not necessarily identically distributed,
    where $0 < L_0 \le \Var(Y_{n,ij}) \le L_1$ and $\E[|Y_{n,ij} - \E[Y_{n,ij}]|^3] \le L_2$ for finite constants $L_0$, $L_1$, and $L_2$ not depending on $n$, $i$, or $j$. Then, consider a $\{0, 1\}$-valued non-random pyramidal array
    $(w_{n,ij})_{1 \le i,j \le n < \infty}$, with a corresponding sequence of induced index sets $\Ic_n := \{(i,j) : w_{n,ij} = 1\}$, and  assume that $|\Ic_n|^{-1} = O(n^{-2})$. Then,
    \begin{equation*}
        \dfrac{1}{\sqrt{\sum_{(i,j) \in \Ic_n} \Var(Y_{n,ij})}} \sum_{(i,j) \in \Ic_n} (Y_{n,ij} - \E[Y_{n,ij}]) \overset{\text{d}}{\rightarrow} \mathcal{N}(0,1).
    \end{equation*}
\end{lemma}
\begin{proof}
    We apply the Lindeberg-Feller-Lyapunov Central Limit Theorem. First, define  $Z_{n,ij} := w_{n,ij} Y_{n,ij}$, implying that $\E[Z_{n,ij}] = w_{n,ij} \E[Y_{n,ij}]$ and $\Var(Z_{n,ij}) = w_{n,ij} \Var(Y_{n,ij})$. Defining $\sigma_n^2 := \Var(\sum_{i=1}^n \sum_{j=1}^n Z_{n,ij}) = \sum_{i=1}^n \sum_{j=1}^n w_{n,ij} \Var(Y_{n,ij})$, note that $\sigma_n^2 < \infty$ for all $n$ by the fact that $\Var(Y_{n,ij}) \le L_1$.

    Next, note that $\sigma_n^3 = (\sigma_n^2)^{3/2} \ge (|\Ic_n| L_0)^{3/2} = |\Ic_n|^{3/2} L_0^{3/2}$, and also that $\sum_{i=1}^n \sum_{j=1}^n \E[|Z_{n,ij} - \E[Z_{n,ij}]|^3] \le |\Ic_n| L_2$. Hence, because
    \begin{align*}
        \sigma_n^{-3} \sum_{i=1}^{n} \sum_{j=1}^n &\E [|Z_{n,ij} - \E[Z_{n,ij}]|^3] \le |\Ic_n|^{-3/2} |\Ic_n| L_0^{-3/2} L_2 = |\Ic_n|^{-1/2} (L_0^{-3/2} L_2) = O(n^{-1}),
    \end{align*}
    the Lyapunov condition holds. So, by the Lindeberg-Feller-Lyapunov central limit theorem, $\dfrac{1}{\sigma_n} \sum_{i=1}^{n} \sum_{j=1}^n (Z_{n,ij} - \E[Z_{n,ij}]) \overset{\text{d}}{\rightarrow} \mathcal{N}(0,1)$. Equivalently, \\$\frac{1}{\sqrt{\sum_{(i,j) \in \Ic_n} \Var(Y_{n,ij})}} \sum_{(i,j) \in \Ic_n} (Y_{n,ij} - \E[Y_{n,ij}]) \overset{\text{d}}{\rightarrow} \mathcal{N}(0,1)$.
\end{proof}

\begin{lemma}    \label{lemma:marginal_independence_joint_convergence}
    Suppose that $X^{(r)}_n \overset{\text{d}}{\rightarrow} X^{(r)}$ for $r = 1, 2, \dots, s$, and that $X_n^{(1)}, X_n^{(2)}, \dots, X_n^{(s)}$ are mutually independent for all $n$. Then, there exist mutually independent random variables $\tilde{X}^{(1)}, \tilde{X}^{(2)}, \dots, \tilde{X}^{(s)}$ such that
    $(X^{(1)}_n, X_n^{(2)}, \dots, X^{(s)}_n) \overset{\text{d}}{\rightarrow}(\tilde{X}^{(1)}, \tilde{X}^{(2)}, \dots,  \tilde{X}^{(s)})$.
\end{lemma}
\begin{proof}
    Because $X^{(r)}_n \overset{\text{d}}{\rightarrow} X^{(r)}$, by Theorem 2.13 in \citet{van2000asymptotic}, for all $t_r \in \mathbb{R}$ we have
    \begin{align}
        \lim_{n \to \infty} &\E[\exp(i t_r X^{(r)}_n)] = \E[\exp(i t_r X^{(r)})]. \label{eq:marg_conv_thing_1}
    \end{align}
    Define $Y_n = (X^{(1)}, X^{(2)}, \dots, X^{(s)})$, and let $\tilde{t} = (\tilde{t}_1, \tilde{t}_2, \dots, \tilde{t}_s) \in \mathbb{R}^s$. Then,
    \begin{align*}
        \E[\exp(i \tilde{t}^\top Y_n)] &= \E \left[\exp \left( i \sum_{r=1}^s \tilde{t}_r X_n^{(r)} \right) \right] = \E \left[ \prod_{r=1}^s \exp \left( i \tilde{t}_r X_n^{(r)} \right) \right] = \prod_{r=1}^s \E[\exp(i \tilde{t}_r X^{(r)}_n)],
    \end{align*}
    where the last equality follows by the mutual independence of $X^{(1)}_n, X^{(2)}_n, \dots, X^{(s)}_n$. Taking limits and applying \eqref{eq:marg_conv_thing_1},
    \begin{align}
        \lim_{n \to \infty} \E[\exp(i \tilde{t}^\top Y_n)] = \prod_{r=1}^s \E[\exp(i \tilde{t}_r X^{(r)})], \label{eq:marg_conv_limit_guy}.
    \end{align}
    Next, construct new random variables $\tilde{X}^{(1)}, \tilde{X}^{(2)}, \dots, \tilde{X}^{(s)}$ so that $\tilde{X}^{(r)}$ is equal in distribution to $X^{(r)}$ for all $r = 1, 2, \dots, s$, and $\tilde{X}^{(1)}, \tilde{X}^{(2)}, \dots, \tilde{X}^{(s)}$ are mutually independent.

    By Lemma 2.15 in \citet{van2000asymptotic}, $\E[\exp(i t_r X^{(r)})] = \E[\exp(i t_r \tilde{X}^{(r)})]$ for all $t_r \in \mathbb{R}$. Using this with \eqref{eq:marg_conv_limit_guy}, we have $\lim_{n \to \infty} \E[\exp(i \tilde{t}^\top Y_n)] = \prod_{r=1}^s \E[\exp(i \tilde{t}_r \tilde{X}^{(r)})] = \E[\exp(i \tilde{t}^\top \tilde{Y})]$, where $\tilde{Y} := (\tilde{X}^{(1)}, \tilde{X}^{(2)}, \dots, \tilde{X}^{(s)})$, and where the last equality follows by the mutual independence of $\tilde{X}^{(1)}, \tilde{X}^{(2)}, \dots, \tilde{X}^{(s)}$.
    
    By Theorem 2.13 in \citet{van2000asymptotic}, we have $(X^{(1)}_n, X^{(2)}_n, \dots, X_n^{(s)}) \overset{\text{d}}{\rightarrow} (\tilde{X}^{(1)}, \tilde{X}^{(2)}, \dots, \tilde{X}^{(s)})$, and note that $\tilde{X}^{(1)}, \tilde{X}^{(2)}, \dots, \tilde{X}^{(s)}$ are mutually independent by construction.
\end{proof}

\begin{lemma}
    \label{lemma:alternative_cmt}
    For two (possibly non-convergent) sequences of random variables $(X_n)_{n=1}^{\infty}$ and $(Y_n)_{n=1}^{\infty}$, suppose that $X_n - Y_n = o_p(1)$, and that there exists a compact set $U = [u_0, u_1] \subset \mathbb{R}$ not depending on $n$ such that $0 < u_0$, and $\lim_{n \to \infty} P(X_n \in U, Y_n \in U) = 1$.

    Then, for a function $g$ uniformly continuous on $U$, we have $g(X_n) - g(Y_n) = o_p(1)$.
\end{lemma}
\begin{proof}

    Fix any $\delta > 0$ and $\lambda > 0$. Because $g$ is uniformly continuous on $U$, there exists an $\eta > 0$ such that for any $z_1, z_2 \in U$ with $|z_1 - z_2| < \eta$, we have $|g(z_1) - g(z_2)| < \delta$. As a consequence, $P(|g(X_n) - g(Y_n)| \ge \delta) \le P(|X_n - Y_n| \ge \eta)$. Next, because $X_n - Y_n = o_p(1)$, there exists $K_{1}$ such that for $n > K_{1}$ we have $P(|X_n - Y_n| \ge \eta) \le \frac{\lambda}{2}$. Then, because $\lim_{n \to \infty} P(X_n \in U, Y_n \in U) = 1$, there exists $K_2$ such that for $n > K_2$ we have $P(X_n \in U, Y_n \in U) \ge 1 - \frac{\lambda}{2}$.

    So, for $n > \textnormal{max}(K_1, K_2)$, we have
    \begin{align*}
        P(|g(X_n) &- g(Y_n)| \ge \delta) = P(|g(X_n) - g(Y_n)| \ge \delta, ~ X_n \in U, Y_n \in U) \\
        & \quad + P(|g(X_n) - g(Y_n)| \ge \delta,~ ((X_n \notin U) \cup (Y_n \notin U)))\\
        &\le P(|X_n - Y_n| \ge \eta) + P((X_n \notin U) \cup (Y_n \notin U) ) \\
        &\le \frac{\lambda}{2} + (1 - P(X_n \in U, Y_n \in U)) \le \lambda.
    \end{align*}
    Because the choice of $\delta$ and $\lambda$ was arbitrary, we conclude that $g(X_n) - g(Y_n) = o_p(1)$.
\end{proof}

\begin{lemma}
    \label{lemma:gaussian_rotation}
    Suppose that $\Sigma_n^{-1/2} \left( \hat{\phi}_n - \phi_n \right) \overset{\text{d}}{\rightarrow} \mathcal{N}_r(0, I_r)$, where $\Sigma_n \in \mathbb{R}^{r \times r}$ is a non-random positive definite diagonal matrix, $\hat{\phi}_n$ is a random $r$-dimensional vector, and $\phi_n$ is a non-random $r$-dimensional vector. Then, for any sequence of vectors $(u_n)_{n=1}^{\infty}$ where $u_n \in \mathbb{R}^{r}$ satisfies $\Vert u_n \Vert_2 = 1$, it follows that
    $\left( u_n^\top \Sigma_n u_n \right)^{-1/2} u_n^\top \left( \hat{\phi}_n - \phi_n \right) \overset{\text{d}}{\longrightarrow} \mathcal{N}(0,1)$.
\end{lemma}
\begin{proof}
    Because $u_n$ has unit norm for all $n$, it cannot be the zero vector, and thus there is no possibility of division by zero in the subsequent derivations. It is useful to define $v_n := \Sigma_n^{1/2} u_n$, in which case $u_n = \Sigma_n^{-1/2} v_n$. Note that $(u_n^\top \Sigma_n u_n)^{-1/2} = \Vert v_n \Vert_2^{-1}$, and so by substitution we have $(u_n^\top \Sigma_n u_n)^{-1/2} u_n^\top \left( \hat{\phi}_n - \phi_n \right) = \frac{v_{n}^\top}{\left\lVert v_{n} \right\rVert_2} \Sigma_{n}^{-1/2} \left( \hat{\phi}_n - \phi_n \right)$. Let $w_{n} := \frac{v_{n}}{\left\lVert v_{n} \right\rVert_2}$, and let $Y_{n} := \Sigma_{n}^{-1/2} \left( \hat{\phi}_n - \phi_n \right)$. Note that \(\left\lVert w_{n} \right\rVert_2 = 1\) for all \(n\), and recall that $Y_n \overset{\text{d}}{\rightarrow} \mathcal{N}_r(0, I_r)$ by assumption.

    Now, consider an arbitrary subsequence indexed by $(n_m)_{m=1}^\infty$. Because $\Vert w_{n_m} \Vert_2 = 1$ for all $m$, the sequence is bounded, and by Bolzano-Weierstrass (Theorem~2.42 in \citet{Rudin1976}), there exists a further subsequence $w_{n_{m_s}}$ converging to a fixed $w \in \mathbb{R}^r$. Along this further subsequence, by Slutsky's theorem, it holds that $w_{n_{m_s}}^\top Y_{n_{m_s}} \overset{\text{d}}{\rightarrow} \mathcal{N}(0, \Vert w \Vert_2^2)$. Note that $w_{n_{m_s}}$ lives on the unit sphere, as $\Vert w_{n_{m_s}} \Vert_2 = 1$ for all $s$. The unit sphere is a closed subset of $\mathbb{R}^r$, and because closed sets contain all their limit points, it follows that $w$ also lives on the unit sphere. That is, $\Vert w \Vert_2 = 1$, and so $w_{n_{m_s}}^\top Y_{n_{m_s}} \overset{\text{d}}{\rightarrow} \mathcal{N}(0,1)$.
    
    We have established that each subsequence $w_{n_m}^\top Y_{n_m}$ has a further subsequence $w_{n_{m_s}}^\top Y_{n_{m_s}}$ that converges in distribution to $\mathcal{N}(0,1)$. Letting $P_n$ denote the probability measure of $w_n^\top Y_n$, and letting $P$ denote the probability measure of a $\mathcal{N}(0,1)$ random variable, this means that each subsequence $P_{n_m}$ contains a further subsequence $P_{n_{m_s}}$ that converges weakly to $P$ as $s \to \infty$. Thus, Theorem 2.6 of \citet{billingsley1999ConvergenceProbabilityMeasures} ensures that $P_n$ converges weakly to $P$. Therefore, $w_n^\top Y_n = (u_n^\top \Sigma_n u_n)^{-1/2} u_n^\top \left( \hat{\phi}_n - \phi_n \right) \overset{\text{d}}{\longrightarrow} \mathcal{N}(0,1)$.
\end{proof}

\begin{lemma}
    \label{lemma:asymptotic_variance_un_equal}
    Suppose that $\hat{\Sigma}_n \Sigma_n^{-1} \overset{\text{p}}{\rightarrow} I_r$, where $\Sigma_n \in \mathbb{R}^{r \times r}$ is a non-random positive definite diagonal matrix, and $\hat{\Sigma}_n \in \mathbb{R}^{r \times r}$ is a random positive definite diagonal matrix. Moreover, suppose that $\Sigma_n$ and $\hat{\Sigma}_n$ admit the decompositions $\Sigma_n = N_n^{-1} \tilde{\Sigma}_n$ and $\hat{\Sigma}_n = N_n^{-1} \hat{\tilde{\Sigma}}_n$, where $N_n \in \mathbb{R}^{r \times r}$ is a diagonal matrix such that $0 \le N_{n,ii} \le n^2$ and $(N_{n,ii})^{-1} = O(n^{-2})$ for all $i = 1, 2, \dots, r$, and where $\tilde{\Sigma}_{n,ii}$ is contained in a compact set $[b_0, b_1]$ not depending on $i$ with $0 < b_0$ for all $i = 1, 2, \dots, r$.

    Then, for any sequence of vectors $(u_n)_{n=1}^{\infty}$ with $\Vert u_n \Vert_2 = 1$, we have $\frac{(u_n^\top \hat{\Sigma}_n u_n)^{-1/2}}{(u_n^\top \Sigma_n u_n)^{-1/2}} \overset{\text{p}}{\rightarrow} 1$.
\end{lemma}
\begin{proof}
    Because $(N_{n,ii})^{-1} = O(n^{-2})$ and $0 \le N_{n,ii} \le n^2$, the quantity $d_{i} := \liminf_{n \to \infty} \dfrac{N_{n,ii}}{n^2}$ satisfies $0 < d_i \le 1$. Along an arbitrary subsequence $(n_m)_{m=1}^\infty$, the quantity $d'_{i} := \liminf_{m \to \infty} \dfrac{N_{{n_m},ii}}{n_m^2}$ satisfies $0 < d_i \le d'_i \le 1$. Whenever a limit inferior exists as a real number, there exists a further subsequence converging to that real number; consequently,  we can find a further subsequence $(n_{m_{s}})_{s=1}^{\infty}$ such that $N_{n_{m_s}} / n_{m_s}^2 \rightarrow \tilde{D}$ as $s \to \infty$, where $\tilde{D}$ is a diagonal matrix with the property that $0 < \tilde{D}_{ii} \le 1$ for all $i = 1, 2, \dots, r$. Then, $n_{m_s}^2 N_{n_{m_s}}^{-1} \rightarrow \tilde{D}^{-1}$.
    
    Next, because $\tilde{\Sigma}_{n_{m_s},ii}$ is contained within a compact set $[b_0, b_1]$ bounded away from zero for all $i$, applying Bolzano-Weierstrass (Theorem 2.42 of \citet{Rudin1976}) we can find a further subsequence indexed by $(n_{m_{s_t}})_{t=1}^\infty$ that converges, say $\tilde{\Sigma}_{n_{m_{s_t}}} \rightarrow \tilde{\Sigma}$. Because $[b_0, b_1]$ is a compact (and thus closed) set, it contains all of its limit points, and so $\tilde{\Sigma}$ must be a diagonal matrix where $\tilde{\Sigma}_{ii} \in [b_0, b_1]$ for all $i$.

    Now, we apply Bolzano-Weierstrass (Theorem 2.42 of \citet{Rudin1976}) once again to find a further subsequence $(n_{m_{s_{t_q}}})_{q=1}^\infty$ such that $u_{n_{m_{s_{t_q}}}} \to u$ as $q \to \infty$, where we know that $\Vert u \Vert_2 = 1$ because $u_{n_{m_{s_{t_q}}}}$ lives on the unit sphere, a closed subset, which consequently contains all of its limit points. Because $(n_{m_{s_{t_q}}})_{q=1}^\infty$ is a subsequence of $(n_{m_{s_{t}}})_{t=1}^\infty$ and the limits of sequences are preserved under subsequences, we know that $\tilde{\Sigma}_{n_{m_{s_{t_q}}}} \to \tilde{\Sigma}$ as $q \to \infty$, where $\tilde{\Sigma}_{ii} \in [b_0, b_1]$ for all $i$.

    For simplicity, relabel $(n_{m_{s_{t_q}}})_{q=1}^\infty$ as $(n_{m_p})_{p=1}^\infty$, so that we have $\tilde{\Sigma}_{n_{m_{p}}} \to \tilde{\Sigma}$ and $u_{n_{m_p}} \to u$ as $p \to \infty$, where $\tilde{\Sigma}$ is a diagonal matrix satisfying $\tilde{\Sigma}_{ii} \in [b_0, b_1]$ for all $i$, and $\Vert u \Vert_2 = 1$.
    
    Because we assumed that $\hat{\Sigma}_n \Sigma_n^{-1} \overset{\text{p}}{\rightarrow} I_r$, we also have $\hat{\Sigma}_{n_{m_p}} \Sigma_{n_{m_p}}^{-1} \overset{\text{p}}{\rightarrow} 1$, as limits are preserved under subsequences. Then, $\hat{\Sigma}_{n_{m_p}} \Sigma_{n_{m_p}}^{-1} = N_{n_{m_p}}^{-1} \hat{\tilde{\Sigma}}_{n_{m_p}} \tilde{\Sigma}_{n_{m_p}}^{-1} N_{n_{m_p}} = N_{n_{m_p}}^{-1} N_{n_{m_p}} \hat{\tilde{\Sigma}}_{n_{m_p}} \tilde{\Sigma}_{n_{m_p}}^{-1} = \hat{\tilde{\Sigma}}_{n_{m_p}} \tilde{\Sigma}_{n_{m_p}}^{-1}$, where the commutation of the matrix multiplication follows by the diagonality of the matrices. Hence, $\hat{\tilde{\Sigma}}_{n_{m_p}} \tilde{\Sigma}_{n_{m_p}}^{-1} \overset{\text{p}}{\rightarrow} I_r$, and because $\tilde{\Sigma}_{n_{m_p}} \rightarrow \tilde{\Sigma}$, it follows by the continuous mapping theorem that $\hat{\tilde{\Sigma}}_{n_{m_p}} \overset{\text{p}}{\rightarrow} \tilde{\Sigma}$ as $p \to \infty$.

    Putting this all together, we have $\dfrac{u_{n_{m_p}}^\top \hat{\Sigma}_{n_{m_p}} u_{n_{m_p}}}{u_{n_{m_p}}^\top \Sigma_{n_{m_p}} u_{n_{m_p}}} = \dfrac{u_{n_{m_p}}^\top n_{m_p}^2 N_{n_{m_p}}^{-1} \hat{\tilde{\Sigma}}_{n_{m_p}} u_{n_{m_p}}}{u_{n_{m_p}}^\top n_{m_p}^2 N_{n_{m_p}}^{-1} \tilde{\Sigma}_{n_{m_p}} u_{n_{m_p}}} \overset{\text{p}}{\longrightarrow} \dfrac{u^\top \tilde{D}^{-1} \tilde{\Sigma} u}{u^\top \tilde{D}^{-1} \tilde{\Sigma} u} = 1$.
    We have established that each subsequence $\dfrac{u^\top_{n_m} \hat{\Sigma}_{n_m} u_{n_m}}{u^\top_{n_m} \Sigma_{n_m} u_{n_m}}$ has a further subsequence $\dfrac{u^\top_{n_{m_p}} \hat{\Sigma}_{n_{m_p}} u_{n_{m_p}}}{u^\top_{n_{m_p}} \Sigma_{n_{m_p}} u_{n_{m_p}}}$ that converges in probability to $1$. Note that convergence in probability to a constant is equivalent to convergence in distribution to a constant. Letting $P_n$ denote the probability measure of $\dfrac{u^\top_{n} \hat{\Sigma}_{n} u_{n}}{u^\top_{n} \Sigma_{n} u_{n}}$, and letting $P$ denote the probability measure of the constant $1$, this means that each subsequence $P_{n_m}$ contains a further subsequence $P_{n_{m_p}}$ that converges weakly to $P$ as $p \to \infty$. Thus, Theorem 2.6 of \citet{billingsley1999ConvergenceProbabilityMeasures} ensures that $P_n$ converges weakly to $P$. Therefore, $\frac{u_n^\top \hat{\Sigma}_n u_n}{u_n^\top \Sigma_n u_n} \overset{\text{p}}{\rightarrow} 1$. Finally, applying the continuous mapping theorem with the continuous function $x \mapsto x^{-1/2}$ yields $\frac{(u_n^\top \hat{\Sigma}_n u_n)^{-1/2}}{(u_n^\top \Sigma_n u_n)^{-1/2}} \overset{\text{p}}{\rightarrow} 1$.
\end{proof}

\subsection{Proof of Proposition \ref{prop:normal_estimation}}
\label{sec:proof_of_normal_estimation}

Throughout this proof, we suppose that Proposition \ref{prop:univariate_gaussian_thinning} is applied to $A$ to yield $\Ate$ and $\Atr$, and community estimation is applied to $\Atr$ to yield $\Ztr \in \{0, 1\}^{n \times K}$. We will implicitly condition on $\Atr$ (and thus consider $\Atr$ and $\Ztr$ fixed), and so explicit conditioning in what follows will be suppressed in the notation. Denoting $M := \E[A]$, note that we can write $A^{(\text{te})} \sim \mathcal{MN}_{n \times n}((1 - \epsilon)M, ~ (1-\epsilon)\tau^2 I_n, I_n)$, where $\mathcal{MN}_{n \times n}$ is the matrix-normal distribution of dimension $n \times n$ \citep{glanz2018expectation}. Next, defining $\hat{N}^{-1} := (\ZtrT \Ztr)^{-1}$, by properties of the matrix-normal distribution it follows that 
\begin{align*}
    \hat{N}^{-1} \ZtrT A^{(\text{te})} \Ztr \hat{N}^{-1} \sim \mathcal{MN}_{K \times K}( (1-\epsilon) \hat{N}^{-1} \ZtrT M \Ztr \hat{N}^{-1}, ~ (1-\epsilon) \tau^2 \hat{N}^{-1}, \hat{N}^{-1}).
\end{align*}
Vectorizing the above leads to
\begin{equation}
    \begin{split}
        & \operatorname{vec}(\hat{N}^{-1} \ZtrT A^{(\text{te})} \Ztr \hat{N}^{-1}) \\
    &\quad\sim \mathcal{N}_{K^2}((1-\epsilon) \operatorname{vec}(\hat{N}^{-1} \ZtrT M \Ztr \hat{N}^{-1}), ~ (1-\epsilon) \tau^2 [\hat{N}^{-1} \otimes \hat{N}^{-1}]),
    \end{split}
    \label{eq:normal_proof_1_1}
\end{equation}
where $\otimes$ is the Kronecker product. Then, left-multiplying by $(1-\epsilon)^{-1} u^\top$ in \eqref{eq:normal_proof_1_1} and observing from \eqref{eq:target_of_inference} and the statement of Proposition~\ref{prop:normal_estimation} that
\begin{align*}
    \theta(\Atr) &:= u^\top \operatorname{vec}(\hat{N}^{-1} \ZtrT M \Ztr \hat{N}^{-1}), \\
    \hat{\theta}(\Ate, \Atr) &:= (1-\epsilon)^{-1} u^\top \operatorname{vec}(\hat{N}^{-1} \ZtrT \Ate \Ztr \hat{N}^{-1}), \\
    \sigma^2 &:= (1-\epsilon)^{-1} \tau^2 u^\top [\hat{N}^{-1} \otimes \hat{N}^{-1}] u,
\end{align*}
yields $\hat{\theta}(\Ate, \Atr) \sim \mathcal{N}(\theta(\Atr), ~ \sigma^2)$. Since we have implicitly conditioned on $\{\Atr = \atr\}$, this implies that $$P \Bigg( \theta(\Atr) \in \Big[ \hat{\theta}(\Ate, \Atr) \pm \phi_{1-\alpha/2} \cdot \sigma \Big] ~\Bigg|~ \Atr \Bigg) = 1 - \alpha.$$

\subsection{Proof of Proposition \ref{prop:poisson_estimation}}
\label{app:poisson_proposition_proof}

In this proof, all expressions will implicitly condition on the sequence of realizations $\left\{ \Antr = \antr \right\}_{n=1}^{\infty}$ and $\left\{ \Zntr = \hat{z}_n \right\}_{n=1}^{\infty}$. To begin, note that applying Proposition~\ref{prop:univariate_poisson_thinning} to $A_{n,ij} \ind \text{Poisson}(M_{n,ij})$ yields $\Anteij \ind \text{Poisson}((1-\epsilon) M_{n,ij})$. In what follows, we refer to this as the ``true'' model, $G_n$. That is, the true model $G_n$ is that $\Anteij \ind \operatorname{Poisson}((1-\epsilon)M_{n,ij})$. We now introduce a misspecified (``working'') model $F_{n, k \ell}$ that assumes the presence of the communities characterized by $\Zntr = \hat{z}_n$: that is,
\begin{equation}
  F_{n, k \ell}: \Anteij \ind \operatorname{Poisson}(\psi_{n, k(i) \ell(j)}),
  \label{eq:poisson_true_model_dist}
\end{equation}
where \(k(i)\) returns the value of \(k\) for which \(\hat{z}_{n, ik} = 1\) and \(\ell(j)\) returns the value of \(\ell\) for which \(\hat{z}_{n, j \ell} = 1\). Up to constants, for a given \(n\), under \(F_{n, k \ell}\) a dyad \((i, j)\) in community pair \((k, \ell)\) (i.e., $\ (i, j) \in \Ic_{n,k\ell} := \{(i,j) ~:~ \hat{z}_{n,ik} = 1, \hat{z}_{n,j \ell} = 1\}$)
has log-likelihood $\ell_{n,ij}(\psi_{n, k \ell}) = \Anteij \log(\psi_{n, k\ell}) - \psi_{n, k \ell}$, with derivative $\ell'_{n,ij}(\psi_{n, k \ell}) = \frac{\Anteij}{\psi_{n, k \ell}} - 1$, and second derivative $\ell''_{n,ij}(\psi_{n, k \ell}) = -\frac{\Anteij}{\psi_{n, k \ell}^2}$. The maximum likelihood estimator under $F_{n, k \ell}$ for this community pair is 
\begin{equation}
    \hat{\psi}_{n, k \ell} := \argmax_{\psi \in (0,\infty)} \sum_{(i,j) \in \Icnskl} \ell_{n,ij}(\psi) 
    = \frac{1}{|\Ic_{n,k \ell}|} \sum_{(i,j) \in \Ic_{n,k \ell}} \Anteij.
\label{eq:poisson_mle_misspecified}
\end{equation}
Next, we define $\psi^*_{n,k \ell}$ to be the value that maximizes the expected log-likelihood of the misspecified model $F_{n, k \ell}$ under the true model \(G_{n}\):
\begin{align}
    &\psi^*_{n,k \ell} := \argmax_{\psi \in (0, \infty)} \E_{G_n} \left[ \sum_{(i,j) \in \Ic_{n,k \ell}} \ell_{n,ij}(\psi) \right] \nonumber \\
    &= \argmax_{\psi \in (0, \infty)} \left( \log(\psi) \left( \sum_{(i,j) \in \Ic_{n,k \ell}} \E_{G_n} [\Anteij] \right) - |\Ic_{n,k \ell}| \psi \right) = \dfrac{1-\epsilon}{|\Ic_{n,k \ell}|} \sum_{(i,j) \in \Ic_{n,k \ell}} M_{n,ij}, 
\end{align}
where the final equality follows from \eqref{eq:poisson_true_model_dist}. As a consequence of \eqref{eq:poisson_mle_misspecified}, $0 =  \sum_{(i,j) \in \Icnkl} \ell'_{n,ij}(\hat{\psi}_{\Icnkl})$, and so by the mean-value theorem,
\begin{equation}
    0 = \sum_{(i,j) \in \Icnkl}  \ell'_{n,ij}(\psi^*_{\Icnkl}) + \sum_{(i,j) \in \Icnkl} \ell''_{n,ij}(\tilde{\psi}_{n, k \ell})(\hat{\psi}_{n, k \ell} - \psi^*_{n, k \ell}),
\label{eq:mean_value_poisson_proof}
\end{equation}
where $\tilde{\psi}_{n, k \ell}$ is a random variable contained between $\hat{\psi}_{n, k \ell}$ and $\psi_{n, k \ell}^*$. Now, define
\begin{align}
    V_{n, k \ell}(\psi) &:= \sum_{(i,j) \in \Ic_{n, k \ell}} \Var_{G_n}[\ell_{n,ij}'(\psi)], ~~\text{and}~~ J_{n, k \ell}(\psi) &:= - \sum_{(i,j) \in \Ic_{n, k \ell}} \E_{G_n}[\ell_{n,ij}''(\psi)], \label{eq:v_j_true_lemma3}
\end{align}
and note that
\begin{align}
    J_{n, k \ell}(\psi^*_{n, k \ell}) &= V_{n, k \ell}(\psi^*_{n, k \ell}) = \dfrac{|\Icnkl|}{\psi^*_{n, k \ell}}. 
    \label{eq:v_j_equal_poisson_proof}
\end{align}
So, by algebraic manipulation of \eqref{eq:mean_value_poisson_proof} we have
\begin{align}
    \sqrt{J_{n, k \ell}(\psi^*_{n, k \ell})} (\hat{\psi}_{n, k \ell} - \psi_{n, k \ell}^*)
    = -\dfrac{ \sum_{(i,j) \in \Icnkl}  \ell'_{n,ij}(\psi_{n, k \ell}^*) / \sqrt{J_{n, k \ell}(\psi^*_{n, k \ell})}}{\sum_{(i,j) \in \Icnkl}  \ell''_{n,ij}(\tilde\psi_{n, k \ell}) / J_{n, k \ell}(\psi_{n, k \ell}^*)}.
    \label{eq:almost_sand_poisson_proof}
\end{align}
Considering the numerator of the right hand side of \eqref{eq:almost_sand_poisson_proof}, by \eqref{eq:v_j_equal_poisson_proof} and Lemma \ref{lemma:mixing_clt},
\begin{equation}
    \dfrac{1}{\sqrt{J_{n, k \ell}(\psi_{n, k \ell}^*)}} \sum_{(i,j) \in \Icnkl}  \ell'_{n,ij}(\psi_{n, k \ell}^*) \overset{\text{d}}{\longrightarrow} \mathcal{N}(0,1).
    \label{eq:normal_conv_ll_poisson_proof}
\end{equation}
Next, we consider the denominator of the right hand side of \eqref{eq:almost_sand_poisson_proof}. Note that 
\begin{align}
    \dfrac{\sum_{(i,j) \in \Icnkl} \ell''_{n,ij}(\tilde{\psi}_{n, k \ell})}{\sum_{(i,j) \in \Icnkl} \ell''_{n,ij}(\psi^*_{n, k \ell})} = \left( \dfrac{\tilde{\psi}_{n, k \ell}}{\psi^*_{n, k \ell}} \right)^{-2}.
    \label{eq:poisson_inf_proof_condition_lemma_s3}
\end{align}
Then, because $\hat{\psi}_{n, k \ell} - \psi^*_{n, k \ell} = o_p(1)$ by Lemma~\ref{lemma:mixing_wlln}, we also have $\tilde{\psi}_{n, k \ell} - \psi^*_{n, k \ell} = o_p(1)$. Then, $\dfrac{\tilde{\psi}_{n, k \ell}}{\psi^*_{n, k \ell}} - 1 = \dfrac{\tilde{\psi}_{n, k \ell} - \psi^*_{n, k \ell}}{\psi^*_{n, k \ell}} = O(1) o_p(1) = o_p(1)$, where $(\psi^*_{n, k \ell})^{-1} = O(1)$ follows by the uniform bounds on $M_{n,ij}$. So, $\frac{\tilde{\psi}_{n, k \ell}}{\psi^*_{n, k \ell}} \overset{\text{p}}{\rightarrow} 1$, and by an application of the continuous mapping theorem with the continuous function $x \mapsto x^{-2}$, we have $\left( \frac{\tilde{\psi}_{n, k \ell}}{\psi^*_{n, k \ell}} \right)^{-2} \overset{\text{p}}{\rightarrow} 1$, and so by \eqref{eq:poisson_inf_proof_condition_lemma_s3},
\begin{align}
    \dfrac{\sum_{(i,j) \in \Icnkl} \ell''_{n,ij}(\tilde{\psi}_{n, k \ell})}{\sum_{(i,j) \in \Icnkl} \ell''_{n,ij}(\psi^*_{n, k \ell})} \overset{\text{p}}{\longrightarrow} 1.
    \label{eq:conv_ratio_denominator_poisson_proof}
\end{align}
Returning to \eqref{eq:almost_sand_poisson_proof} and combining \eqref{eq:normal_conv_ll_poisson_proof} and \eqref{eq:conv_ratio_denominator_poisson_proof} with Slutsky's theorem yields 
\begin{equation}
    \sqrt{J_{n,k \ell}(\psi^*_{n,k \ell})} (\hat{\psi}_{n,k \ell} - \psi^*_{n,k \ell}) \overset{\text{d}}{\longrightarrow} \mathcal{N}(0,1).
    \label{eq:convergence_poisson_prop}
\end{equation}
This establishes the convergence result for \emph{a single community pair} indexed by some $(k, \ell) \in \{1, 2, \dots, K\}^2$. Next, we turn to the more general convergence result of Proposition \ref{prop:poisson_estimation}.

Note that the observations $\Anteij$ are mutually independent under both the true and misspecified models, and each $(i,j)$ pair belongs to exactly one set  $\Ic_{n,k \ell}$ where $1 \leq k \leq K, 1 \leq \ell \leq K$. Consequently, all $\{\hat{\psi}_{n, k \ell}\}_{1 \le k, \ell \le K}$ are mutually independent across all community pairs. Defining $\hat{B}_n$ and $B^*_n$ to be the $K \times K$ matrices whose $(k,l)$th entries are $(1-\epsilon)^{-1} \hat{\psi}_{n,k \ell}$ and $(1-\epsilon)^{-1} \psi^*_{n,k \ell}$, respectively, by Lemma~\ref{lemma:marginal_independence_joint_convergence} we have
\begin{align}
   \Sigma_n^{-1/2} \left( \textnormal{vec} \left(\hat{B}_n \right) - \textnormal{vec} \left( B_n^* \right) \right) \overset{\text{d}}{\longrightarrow} \mathcal{N}_{K^2}(0, I_{K^2}),
   \label{eq:poisson_almost_end_convergence}
\end{align}
where $\Sigma_n := \textnormal{diag}(\textnormal{vec}(\Delta_n))$, and $\Delta_n \in \mathbb{R}^{K \times K}$ is defined entry-wise as
\begin{equation}
    \Delta_{n, k \ell} := \dfrac{1}{J_{n, k \ell}(\psi^*_{n, k \ell}) (1-\epsilon)^2} = (1-\epsilon)^{-1} \dfrac{1}{|\Icnkl|^2} \sum_{(i,j) \in \Icnkl} M_{n,ij}.
    \label{eq:delta_nkl}
\end{equation}
In the setting of Proposition \ref{prop:poisson_estimation}, the parameter of interest takes the form \( \theta_n =u_{n}^{\top} \operatorname{vec} \left( B_n^* \right)\) and the estimator takes the form $\hat{\theta}_n := u_n^\top \textnormal{vec} \left( \hat{B}_n \right)$, so by Lemma \ref{lemma:gaussian_rotation}, we have
\begin{equation}
        (u_n^\top \Sigma_n u_n)^{-1/2} u_n^\top \left( \operatorname{vec}(\hat{B}_n) - \operatorname{vec}(B^*_n) \right) = (u_n^\top \Sigma_n u_n)^{-1/2} \left( \hat{\theta}_n - \theta_n \right) \overset{\text{d}}{\longrightarrow} \mathcal{N}(0,1).  \label{eq:poisson_fixed_sigma_convergence}
\end{equation}
In practice, \(\Sigma_{n}\) is unknown, so we will make use of Lemma \ref{lemma:asymptotic_variance_un_equal} to establish that $\frac{(u_n^\top \hat{\Sigma}_n u_n)^{-1/2}}{(u_n^\top \Sigma_n u_n)^{-1/2}} \overset{\text{p}}{\rightarrow} 1$,
where $\hat{\Sigma}_n := \operatorname{diag}(\operatorname{vec}(\hat{\Delta}_n)),$ with $\hat{\Delta}_n$ defined entry-wise as
\begin{equation}
    \hat{\Delta}_{n, k \ell} := \dfrac{1}{\hat{J}_{n, k \ell}(\hat{\psi}_{n, k \ell}) (1-\epsilon)^2},
    \label{eq:delta_hat_nkl}
\end{equation}
where $\hat{J}_{n, k \ell}(\hat{\psi}_{n, k \ell}) =  -\sum_{(i,j) \in \Ic_{n, k \ell}} \ell_{n,ij}''(\hat{\psi}_{n, k \ell}) = \dfrac{|\Icnkl|}{\hat{\psi}_{n, k \ell}}$. To show that the conditions of Lemma \ref{lemma:asymptotic_variance_un_equal} are satisfied, we will first show that $\hat{\Sigma}_n \Sigma_n^{-1} \overset{\text{p}}{\rightarrow} I_{K^2}$. To do this, note that $\frac{\hat{J}_{n, k \ell}(\hat{\psi}_{n, k \ell})}{J_{n, k \ell}(\psi^*_{n, k \ell})} = \left( \frac{\hat{\psi}_{n, k \ell}}{\psi^*_{n, k \ell}} \right)^{-1}$. Then, we have $\frac{\hat{\psi}_{n, k \ell}}{\psi^*_{n, k \ell}} -1 = \dfrac{\hat{\psi}_{n, k \ell} - \psi^*_{n, k \ell}}{\psi^*_{n, k \ell}} = O(1) o_p(1)$, where we use the facts that $(\psi^*_{n, k \ell})^{-1} = O(1)$ by the uniform bounds on $M_{n,ij}$, and $\hat{\psi}_{n, k \ell} - \psi^*_{n, k \ell} = o_p(1)$ by Lemma~\ref{lemma:mixing_wlln}. Hence, $\dfrac{\hat{\psi}_{n, k \ell}}{\psi^*_{n, k \ell}} \overset{\text{p}}{\rightarrow} 1$, and so by the continuous mapping theorem, we have $\dfrac{\hat{J}_{n, k \ell}(\hat{\psi}_{n, k \ell})}{J_{n, k \ell}(\psi^*_{n, k \ell})} \overset{\text{p}}{\rightarrow} 1$, and subsequently that $\hat{\Sigma}_n \Sigma_n^{-1} \overset{\text{p}}{\rightarrow} I_{K^2}$.

Showing the next requirement of Lemma \ref{lemma:asymptotic_variance_un_equal} involves revisiting Equations (\ref{eq:delta_nkl}) and (\ref{eq:delta_hat_nkl}) and rewriting $\Delta_{n, k \ell} \frac{1}{|\Icnkl|} \left( (1-\epsilon)^{-1} \frac{1}{|\Icnkl|} \sum_{(i,j) \in \Icnkl} M_{n,ij} \right) =: \frac{1}{|\Ic_{n, k \ell}|} \tilde{\Delta}_{n, k \ell}$, and $\hat{\Delta}_{n, k \ell} = \frac{1}{|\Icnkl|} \left( (1-\epsilon)^{-2} \frac{1}{|\Icnkl|} \sum_{(i,j) \in \Icnkl} \Anteij \right) =: \frac{1}{|\Ic_{n, k \ell}|} \hat{\tilde{\Delta}}_{n, k \ell}$. With $\tilde{N}_n \in \mathbb{R}^{K \times K}$ defined entry-wise as $\tilde{N}_{n, k \ell} := |\Ic_{n, k \ell}|$, we can define $N_n := \operatorname{diag}(\operatorname{vec}(\tilde{N}_n))$ and write $\hat{\Sigma}_n = N_n^{-1} \hat{\tilde{\Sigma}}_n$ and $\Sigma_n = N_n^{-1} \tilde{\Sigma}_n$, where $\tilde{\Sigma}_n := \operatorname{diag}(\operatorname{vec}(\tilde{\Delta}_n))$ and $\hat{\tilde{\Sigma}}_n := \operatorname{diag}(\operatorname{vec}(\hat{\tilde{\Delta}}_n))$. Then, it is clear that $0 \le N_{n, k k} \le n^2$, and by assumption, $(N_{n, kk})^{-1} = O(n^{-2})$. The diagonal entries of the matrix $\tilde{\Sigma}_n$ are contained in a compact set bounded away from 0, because of the bound $0 < N_0 \le M_{n,ij} \le N_1$. Hence, applying Lemma \ref{lemma:asymptotic_variance_un_equal} yields $\frac{(u_n^\top \hat{\Sigma}_n u_n)^{-1/2}}{(u_n^\top \Sigma_n u_n)^{-1/2}} \overset{\text{p}}{\rightarrow} 1$, which applied to \eqref{eq:poisson_fixed_sigma_convergence} with Slutsky's theorem yields $\left( u_n^\top \hat{\Sigma}_n u_n \right)^{-1/2} \left( \hat{\theta}_n - \theta_n \right) \overset{\text{d}}{\rightarrow} \mathcal{N}(0,1)$.

Convergence in distribution implies pointwise convergence of the cumulative distribution function at all continuity points (which is all points in the case of $\mathcal{N}(0,1)$), so by reminding the reader of the implicit conditioning on $\{\Antr = \antr\}_{n =1}^{\infty}$, we have $\lim_{n \to \infty} P \Big( \theta_n(\Antr) \in [\hat{\theta}_n(\Ante, \Antr) \pm \phi_{1 - \alpha/2} \cdot \hat{\sigma}_n] \mid \Antr = \antr \Big) = 1 - \alpha$, where $\hat{\sigma}^2_n := u_n^\top \hat{\Sigma}_n u_n$.

\subsection{Proof of Proposition \ref{prop:taylor}(\ref{prop:taylor_a})}
\label{app:taylor_a_proof}

For simplicity, we suppress the notation indicating dependence of objects on $\Atr$. Define $d_0 := \frac{\gamma}{1-\gamma}$ and $d_1 := \frac{1-\gamma}{\gamma}$. Then, write $\Phi_{k \ell} = \dfrac{|\Iczkl|}{|\Ickl|} \Phi_{k \ell}^{(0)} + \dfrac{|\Icokl|}{|\Ickl|} \Phi_{k \ell}^{(1)}$, where for $s \in \{0,1\}$,
\begin{equation}
    \Phi_{k \ell}^{(s)} = \Phi_{k \ell}^{(s)}(\mathcal{M}_{k \ell}^{(s)}) := f \Bigg( \frac{1}{|\Icskl|} \sum_{(i,j) \in \Icskl} f(M_{ij}, d_s), ~ d_s^{-1} \Bigg),
    \label{eq:phi_s_original_def_matrix_version}
\end{equation}
where $\mathcal{M}_{k \ell}^{(s)} := \{M_{ij} ~:~ (i,j) \in \Icskl\}$, and where $f: (0,1) \times \mathbb{R}_+ \to (0,1)$ was defined in~\eqref{eq:f_def_bernoulli}. Although $\Phi_{k \ell}^{(s)}$ depends on the constant $d_s$, we will not vary this argument in this proof, and so we suppress the dependence of $\Phi_{k \ell}^{(s)}$ on $d_s$. For simplicity of notation, we vectorize and reindex $\mathcal{M}^{(s)}_{k \ell}$ as $\{M_{i} ~:~ 1 \le i \le n^{(s)}_{k \ell}\}$ where $n^{(s)}_{k \ell} := |\Icskl|$, and $n_{k \ell} := |\Ickl|$.

Then, let us rewrite \eqref{eq:phi_s_original_def_matrix_version} as $\Phi^{(s)}_{k \ell}(M_1, \dots, M_{n_{k \ell}^{(s)}}) = f \left( \dfrac{1}{n_{k \ell}^{(s)}} \sum_{i=1}^{n_{k\ell}^{(s)}} f(M_i, d_s), d_s^{-1} \right)$. Now, for $i \in \{1, 2, \dots, n_{k \ell}^{(s)}\}$ we define $L^{(s)}_i(t) := tM_i + (1-t) B^{(s)}_{k \ell}$, where $B^{(s)}_{k \ell} := \frac{1}{n_{k \ell}^{(s)}} \sum_{i=1}^{n_{k \ell}^{(s)}} M_i$. Then, we consider a Taylor expansion of $\Xi^{(s)}_{k \ell}(t) := \Phi^{(s)}_{k \ell} (L^{(s)}_1(t), \dots, L^{(s)}_{n_{k \ell}^{(s)}}(t))$
around $t=0$. By the mean-value form of Taylor's theorem,
\begin{equation}
    \Xi^{(s)}_{k \ell}(1) = \Xi^{(s)}_{k \ell}(0) + \Xi_{k \ell}'^{(s)}(0) + \frac{1}{2} \Xi''^{(s)}_{k \ell}(t^{(s)})
    \label{eq:Xi_mean_value_6a_proof}
\end{equation}
for some $0 \le t^{(s)} \le 1$. Next, by algebraic simplification, $\Xi^{(s)}_{k \ell}(1) = \Phi_{k \ell}^{(s)}(M_1, \dots, M_{n_{k \ell}^{(s)}})$, and $\Xi^{(s)}_{k \ell}(0) = \Phi_{k \ell}^{(s)}(B_{k \ell}^{(s)}, \dots, B_{k \ell}^{(s)}) = B_{k \ell}^{(s)}$. By the multivariate chain rule,
\begin{align}
    \Xi'^{(s)}_{k \ell}(0) = \dfrac{d}{dt} \Phi^{(s)}_{k \ell}(L^{(s)}_1(t), \dots, L^{(s)}_{n_{k \ell}^{(s)}}(t)) \Bigg|_{t=0} = \sum_{i=1}^{n_{k \ell}^{(s)}} L'^{(s)}_i(t) \cdot D_i \Phi_{k \ell}^{(s)}(L^{(s)}_1(t), \dots, L^{(s)}_{n_{k \ell}^{(s)}}(t)) \Bigg|_{t=0} \label{eq:Xi_s_prime_at_zero_prop_6a}
\end{align}
where $D_i \Phi_{k \ell}^{(s)}(z_1, \dots, z_{n^{(s)}_{k \ell}})$ is the partial derivative of $\Phi^{(s)}_{k \ell}$ with respect to its $i$th component, evaluated at the point $(z_1, \dots, z_{n_{k \ell}^{(s)}})$. Noting that $L'^{(s)}_i(t) = M_i - B^{(s)}_{k \ell}$ and that
\begin{align*}
    D_i &\Phi^{(s)}_{k \ell}(L^{(s)}_1(t), \dots, L^{(s)}_{n^{(s)}_{k \ell}}(t)) \\
    &= \dfrac{1}{\nskl} \left( \left(1 + \dfrac{1}{\nskl} \sum_{j=1}^{\nskl} f(L^{(s)}_j(t), d_s)(d_s^{-1}-1) \right) \left( 1 - L^{(s)}_i(t) + L^{(s)}_i(t) d_s \right) \right)^{-2},
\end{align*} 
this implies that $D_i \Phi^{(s)}_{k \ell}(L^{(s)}_1(0), \dots, L^{(s)}_{n_{k \ell}^{(s)}}(0)) = r$ is a constant not depending on $i$. Thus, returning to \eqref{eq:Xi_s_prime_at_zero_prop_6a} we have $\Xi'^{(s)}_{k \ell}(0) = r \sum_{i=1}^{n_{k \ell}^{(s)}} (M_i - \Bskl) = 0$. Moving on to the second derivative, by the product rule and the multivariate chain rule, we have
\begin{align*}
    &\Xi''^{(s)}_{k \ell}(t) = \dfrac{d}{dt} \sum_{i=1}^{n_{k \ell}^{(s)}} L_i'^{(s)}(t) D_i \Phi^{(s)}(L^{(s)}_1(t), \dots, L^{(s)}_{\nskl}(t)) \\
    &= \sum_{i=1}^{n_{k \ell}^{(s)}} \Bigg( L_i''^{(s)}(t) D_i \Phi^{(s)}_{k \ell}(L^{(s)}_1(t), \dots, L^{(s)}_{\nskl}(t)) + L_i'^{(s)}(t) \cdot \dfrac{d}{dt} D_i \Phi^{(s)}(L^{(s)}_1(t), \dots, L^{(s)}_{n_{k \ell}^{(s)}}(t)) \Bigg) \\
    &= \sum_{i=1}^{\nskl} \sum_{j=1}^{\nskl} (M_i - \Bskl)(M_j - \Bskl) D_i D_j \Phi^{(s)}_{k \ell}(L^{(s)}_1(t), \dots, L^{(s)}_{\nskl}(t)) \\
    &= \sum_{i=1}^{\nskl} \sum_{j=1}^{\nskl} (M_i - \Bskl)(M_j - \Bskl) h_{ij}^{(s)(k \ell)}(t),
\end{align*}
where we define
\begin{align}
    &h_{ij}^{(s)(k \ell)}(t) := D_i D_j \Phi_{k \ell}^{(s)}(L^{(s)}_1(t), \dots, L^{(s)}_{\nskl}(t)) \nonumber \\
    &= \dfrac{d}{d z_i} \Bigg( \dfrac{1}{\nskl} \Bigg[ 1 + \dfrac{1}{\nskl} \sum_{k=1} f( z_k, d_s )(d_s^{-1} - 1) \Bigg]^{-2} \left( 1 - z_j + z_j d_s \right)^{-2} \Bigg) \Bigg|_{z = (L^{(s)}_1(t), \dots, L^{(s)}_{\nskl}(t))}. \label{eq:h_original_def_6a}
\end{align}
Inserting these results into \eqref{eq:Xi_mean_value_6a_proof} and taking the sum $\frac{\nzkl}{\nkl} \Phi^{(0)}_{k \ell} + \frac{\nokl}{\nkl} \Phi^{(1)}_{k \ell}$,
\begin{equation}
    \Phi_{k \ell} = B_{k \ell} + \frac{1}{2} \sum_{s \in \{0,1\}} \frac{\nskl}{\nkl} \sum_{\substack{i=1 \\ j=1}}^{\nskl} (M_i - \Bskl)(M_j - \Bskl)h_{ij}^{(s)(k \ell)}(t^{(s)})
    \label{eq:prop_6a_reindexed_almost_done}
\end{equation}
for some $0 \le t^{(s)} \le 1$ for $s \in \{0,1\}$.

To conclude, we revert to the original indexing and unvectorized version of the expression (i.e., replace $i$ with $(i,j)$, $j$ with $(i',j')$, $\nkl$ with $|\Ickl|$, and $\nskl$ with $|\Icskl|$) and write
\begin{align*}
    \Phi_{k \ell} &= B_{k \ell} + \sum_{s \in \{0,1\}} \dfrac{|\Icskl|}{|\Ickl|} \sum_{\substack{(i,j) \in \Icskl \\ (i',j') \in \Icskl}} (M_{ij} - \Bskl)(M_{i'j'} - \Bskl) h^{(s)(k \ell)}_{iji'j'}(t_s),
\end{align*}
and absorb the $\frac{1}{2}$ factor from \eqref{eq:prop_6a_reindexed_almost_done} to rewrite $h^{(s)(k \ell)}_{ij}$ from \eqref{eq:h_original_def_6a} as $h^{(s)(k \ell)}_{iji'j'}$, where
\begin{align}
    h^{(s)(k \ell)}_{iji'j'}(t) &= \dfrac{d}{d z_{ij}} \Bigg[ \dfrac{1}{2n} \Bigg( 1 + \dfrac{1}{n} \sum_{(\tilde{i},\tilde{j}) \in \Icskl} f (z_{\tilde{i} \tilde{j}}, d_s)(d_s^{-1}-1) \Bigg)^{-2} \nonumber \\
    &\quad \quad \quad \cdot (1-z_{i'j'} + z_{i'j'} d_s)^{-2} \Bigg]_{\{z_{ij} = L^{(s)}_{ij}(t) ~ \forall (i,j) \in \Icskl\}} ,\label{eq:h_function_prop6a_def}
\end{align}
and remind the reader that $L^{(s)}_{ij}(t) := t M_{ij} + (1-t) \Bskl$.

\subsection{Proof of Proposition \ref{prop:taylor}(\ref{prop:taylor_b})}
\label{app:taylor_b_proof}

For simplicity of notation, in this proof we will suppress the explicit dependence of all objects on $\Atr$. Define $d_0 := \frac{\gamma}{1-\gamma}$ and $d_1 := \frac{1-\gamma}{\gamma}$. Then, write
\begin{equation}
    \Phi_{k \ell} = \dfrac{|\Iczkl|}{|\Ickl|} \Phi_{k \ell}^{(0)} + \dfrac{|\Icokl|}{|\Ickl|} \Phi_{k \ell}^{(1)},
    \label{eq:phi_combined_result_6b_taylor}
\end{equation}
so that for $s \in \{0,1\}$, we have $\Phi_{k \ell}^{(s)} := f \left( \dfrac{1}{|\Icskl|} \sum_{(i,j) \in \Icskl} f(M_{ij}, d_s), ~ d_s^{-1} \right)$, where $f: (0,1) \times \mathbb{R}_+ \to (0,1)$ was defined in~\eqref{eq:f_def_bernoulli}. First, we will produce a Taylor expansion of $\Phi^{(0)}_{k \ell}$ (viewed as a function of $d_0$) around $1$. Using the definition of $f$, we have that
\begin{align}
    \Phi^{(0)}_{k \ell} &= \expit \Bigg( \logit \Bigg( \dfrac{1}{|\Iczkl|} \sum_{(i,j) \in \Iczkl} \expit \Big( \logit(M_{ij}) + \log(d_0) \Big) \Bigg) - \log(d_0) \Bigg) =: g^{(0)}_{k \ell}(d_0). \label{eq:g0kl_def_taylor6b}
\end{align}
Our Taylor expansion result will be of the form
\begin{equation}
    g^{(0)}_{k \ell}(d_0) = g^{(0)}_{k \ell}(1) + g'^{(0)}_{k \ell}(1) (d_0-1) + \frac{g''^{(0)}_{k \ell}(\lambda_0)}{2}(d_0 -1)^2,
    \label{eq:taylor_expand_6b_step_g0_intermediate}
\end{equation}
where $d_0 \le \lambda_0 \le 1$. Taking a derivative with respect to $d_0$, we have
\begin{align*}
    g'^{(0)}_{k \ell}(d_0) &= \dfrac{\expit \Big( \logit \Big( \frac{1}{|\Iczkl|} \sum_{(i,j) \in \Iczkl} \expit( \logit(M_{ij}) + \log(d_0)) \Big) - \log(d_0) \Big)}{1 + \exp \Big( \logit \Big( \frac{1}{|\Iczkl|} \sum_{(i,j) \in \Iczkl} \expit(\logit(M_{ij}) + \log(d_0)) \Big) - \log(d_0) \Big)} \\
    \cdot \Bigg[ &-\frac{1}{d_0} + \Big( \frac{1}{|\Iczkl|} \sum_{(i,j) \in \Iczkl} \expit(\logit(M_{ij}) + \log(d_0))  \Big)^{-1} \\
    &\quad \quad \quad \quad \cdot \Big( 1 - \frac{1}{|\Iczkl|} \sum_{(i,j) \in \Iczkl} \expit(\logit(M_{ij}) + \log(d_0)) \Big)^{-1} \\
    &\quad \quad \quad \quad \cdot \dfrac{1}{|\Iczkl|} \sum_{(i,j) \in \Iczkl} \Big( \dfrac{\expit(\logit(M_{ij}) + \log(d_0))}{1 + \exp(\logit(M_{ij}) + \log(d_0))} \Big) \cdot \dfrac{1}{d_0} \Bigg].
\end{align*}
Evaluating $g'^{(0)}_{k \ell}(1)$ and algebraically simplifying, we have $g'^{(0)}_{k \ell}(1) = -\dfrac{1}{|\Iczkl|} \sum_{(i,j) \in \Iczkl} (M_{ij} - \Bzkl)^2$, where $\Bzkl := \frac{1}{|\Iczkl|} \sum_{(i,j) \in \Iczkl} M_{ij} = g^{(0)}_{k \ell}(1)$. Returning to \eqref{eq:taylor_expand_6b_step_g0_intermediate}, we have
\begin{align}
    \Phi^{(0)}_{k \ell} = \Bzkl + (1-d_0) \frac{1}{|\Iczkl|} \sum_{(i,j) \in \Iczkl} (M_{ij} - \Bzkl)^2 + \dfrac{g''^{(0)}_{k \ell}(\lambda_0)}{2}(d_0 - 1)^2.
    \label{eq:taylor_6b_phi0_almost_done}
\end{align}
Proceeding with $\Phi^{(1)}_{k \ell}$ does not add much difficulty: noting that $d_1 = d_0^{-1}$, we have
\begin{align}
    \Phi^{(1)}_{k \ell} &= \expit \Bigg( \logit \Bigg( \dfrac{1}{|\Icokl|} \sum_{(i,j) \in \Icokl} \expit \Big( \logit(M_{ij}) + \log(d_0^{-1}) \Big) \Bigg) - \log(d_0^{-1}) \Bigg) \nonumber \\
    &=: g^{(1)}_{k \ell}(d_0^{-1}) =: \tilde{g}^{(1)}_{k \ell}(d_0), \label{eq:g1kl_tilde_def_taylor6b}
\end{align}
where $g^{(1)}_{k \ell}$ is defined analogously to $g^{(0)}_{k \ell}$ by replacing $\Iczkl$ with $\Icokl$. Then, by Taylor's theorem,
\begin{equation}
    \Phi^{(1)}_{k \ell} = \tilde{g}^{(1)}_{k \ell}(d_0) = \tilde{g}^{(1)}_{k \ell}(1) + \tilde{g}'^{(1)}_{k \ell}(1)(d_0 - 1) + \frac{\tilde{g}_{k \ell}''^{(1)}(\lambda_1)}{2}(d_0 - 1)^2.
    \label{eq:taylor_phi1_prop_6b_intermediate}
\end{equation}
To get the derivative $\tilde{g}'^{(1)}_{k \ell}$, by the chain rule,
\begin{align*}
    \tilde{g}'^{(1)}_{k \ell}(d_0) = \frac{\mathrm{d}}{\mathrm{d} ~d_0} g^{(1)}_{k \ell}(d_0^{-1}) &= -g'^{(1)}_{k \ell}(d_0^{-1}) \frac{1}{d_0^2} = -g'^{(1)}_{k \ell}(d_1) \frac{1}{d_1^2}.
\end{align*}
Calculating $g'^{(1)}_{k \ell}$ is nearly identical to the previous calculation of $g'^{(0)}_{k \ell}$, and so we have $\tilde{g}'^{(1)}_{k \ell}(1) = -g'^{(1)}_{k \ell}(1) = \dfrac{1}{|\Icokl|} \sum_{(i,j) \in \Icokl} (M_{ij} - \Bokl)^2$. Substituting into \eqref{eq:taylor_phi1_prop_6b_intermediate} and noting that $B_{k \ell}^{(1)} = \tilde{g}^{(1)}_{k \ell}(1)$, we obtain
\begin{equation}
    \Phi^{(1)}_{k \ell} = \Bokl -(1 - d_0) \frac{1}{|\Icokl|} \sum_{(i,j) \in \Icokl} (M_{ij} - \Bokl)^2 + \dfrac{\tilde{g}''^{(1)}_{k \ell}(\lambda_1)}{2}(d_0 - 1)^2.
    \label{eq:taylor_6b_phi1_almost_done}
\end{equation}
Combining \eqref{eq:taylor_6b_phi0_almost_done} and \eqref{eq:taylor_6b_phi1_almost_done} via \eqref{eq:phi_combined_result_6b_taylor}, we have
\begin{align*}
    \Phi_{k \ell} &= B_{k \ell} + \frac{1 - d_0}{|\Ickl|} \left[ \sum_{(i,j) \in \Iczkl} (M_{ij} - \Bzkl)^2 - \sum_{(i,j) \in \Icokl} (M_{ij} - \Bokl)^2 \right] + (d_0 - 1)^2 q_{k \ell}(\lambda_0, \lambda_1),
\end{align*}
for some $\lambda_0, \lambda_1 \in \left[ \frac{\gamma}{1-\gamma}, 1 \right]$, where we define
\begin{align}
    q_{k \ell}(\lambda_0, \lambda_1) := q^{(0)}_{k \ell}(\lambda_0) + q^{(1)}_{k \ell}(\lambda_1), \label{eq:q_s_def_prop6b_proof}
\end{align}
and for $s \in \{0,1\}$ we define $q^{(s)}_{k \ell}(z) := \begin{cases}
        \frac{g''^{(0)}_{k \ell}(z)}{2}, & \text{if }s=0, \\
        \frac{\tilde{g}''^{(1)}_{k \ell}(z)}{2}, & \text{if }s=1,
    \end{cases}$, where $g^{(0)}_{k \ell}(z)$ was defined in \eqref{eq:g0kl_def_taylor6b} and where $\tilde{g}^{(1)}_{k \ell}(z)$ was defined in \eqref{eq:g1kl_tilde_def_taylor6b}.

\subsection{Proof of Proposition \ref{prop:gamma_to_zero}}
\label{app:gamma_to_zero_proof}

Define $d_0 := \frac{\gamma}{1-\gamma}$, $Q_{ij} := \frac{M_{ij}}{1-M_{ij}}$, $c^{(0)} := \log(d_0)$, and  $c^{(1)} := -\log(d_0)$. Recalling the definition of $\Phi_{k \ell}(\Atr)$ in \eqref{eq:Phi_kl_def} and the definitions of $f(a,v)$ in \eqref{eq:f_def_bernoulli} and $V^{(s)}_{k \ell}(\Atr)$ in \eqref{eq:V}, let us write $\Phi_{k \ell}(\Atr) = \frac{|\Iczkl|}{|\Ickl|} g_\gamma^{(0)}(\Iczkl) + \frac{|\Icokl|}{|\Ickl|} g_\gamma^{(1)}(\Icokl)$, where
\begin{equation}
    g^{(s)}_\gamma(\Icskl) := \expit \left( \textnormal{logit} \left( \frac{1}{|\Ic_{k \ell}^{(s)}|} \sum_{(i,j) \in \Ic^{(s)}_{k \ell}}  \expit(\textnormal{logit}(M_{ij}) + c^{(s)}) \right) - c^{(s)} \right),
    \label{eq:prop_gamma0_g_def}
\end{equation}
and where $\gamma \in (0, \frac{1}{2})$. In what follows, to avoid division by zero, we assume that all sets (e.g., $\Ickl$ and $\Icskl$) are non-empty.

First, we show that for any \textit{fixed} and non-empty $\Ic \subset [n]^2$, we have $\lim_{\gamma \to 0}g_\gamma^{(s)}(\Ic) = g_0^{(s)}(\Ic)$, where
\begin{equation}
    g_0^{(s)}(\Ic) := \begin{cases}
    \expit \left( \log \left( \frac{1}{|\Ic|} \sum_{(i,j) \in \Ic} \frac{M_{ij}}{1-M_{ij}} \right) \right), & \text{if } s = 0, \\ 
    \expit \left( \log( \left( \frac{1}{|\Ic|} \sum_{(i,j) \in \Ic} \frac{1-M_{ij}}{M_{ij}} \right)^{-1} \right), & \text{if } s = 1.
\end{cases}
\label{eq:g0_def_prop_gamma0}
\end{equation}
We will often apply the identity $\expit(\log(y)) = \frac{y}{1+y}$ for $y > 0$. First, for the $s = 0$ case,
\begin{align*}
    g_\gamma^{(0)}(\Ic) &= \expit \left( \logit \left( \frac{1}{|\Ic|} \sum_{(i,j) \in \Ic} \expit( \log(Q_{ij}) + \log(d_0)) \right) - \log(d_0) \right) \\
    &= \expit \left( \log \left( \frac{\frac{1}{|\Ic|} \sum_{(i,j) \in \Ic} \frac{Q_{ij}}{1 + Q_{ij} d_0}}{1 - \frac{1}{|\Ic|} \sum_{(i,j) \in \Ic} \frac{Q_{ij} d_0}{1 + Q_{ij} d_0}} \right) \right) = \frac{\frac{\frac{1}{|\Ic|} \sum_{(i,j) \in \Ic} \frac{Q_{ij}}{1 + Q_{ij} d_0}}{1 - \frac{1}{|\Ic|} \sum_{(i,j) \in \Ic} \frac{Q_{ij} d_0}{1 + Q_{ij} d_0}}}{1 + \frac{\frac{1}{|\Ic|} \sum_{(i,j) \in \Ic} \frac{Q_{ij}}{1 + Q_{ij} d_0}}{1 - \frac{1}{|\Ic|} \sum_{(i,j) \in \Ic} \frac{Q_{ij} d_0}{1 + Q_{ij} d_0}}}.
\end{align*}
Because $\lim_{\gamma \to 0} d_0 = 0$, we have
\begin{align*}
    \lim_{\gamma \to 0} g^{(0)}_\gamma(\Ic) &= \frac{\dfrac{1}{|\Ic|} \sum_{(i,j) \in \Ic} Q_{ij}}{1 + \frac{1}{|\Ic|} \sum_{(i,j) \in \Ic} Q_{ij}} = \expit \left( \log \left( \frac{1}{|\Ic|} \sum_{(i,j) \in \Ic} \frac{M_{ij}}{1 - M_{ij}} \right) \right) = g_0^{(0)}(\Ic).
\end{align*}
Next, for the $s = 1$ case, 
\begin{align*}
    g_\gamma^{(1)}(\Ic) &= \expit \left( \logit \left( \frac{1}{|\Ic|} \sum_{(i,j) \in \Ic} \expit( \log(Q_{ij}) - \log(d_0)) \right) + \log(d_0) \right) \\
    &= \expit \left( \log \left( \frac{\frac{1}{|\Ic|} \sum_{(i,j) \in \Ic} \frac{ d_0 Q_{ij}}{d_0 + Q_{ij}}}{1 - \frac{1}{|\Ic|} \sum_{(i,j) \in \Ic} \frac{Q_{ij}}{d_0 + Q_{ij}}} \right) \right) = \expit \left( \log \left( \dfrac{f_1(d_0)}{f_2(d_0)} \right) \right).
\end{align*}
L'H\^{o}pital's Rule implies that $\lim_{d_0 \to 0} \frac{f_1(d_0)}{f_2(d_0)} = \frac{1}{\frac{1}{|\Ic|} \sum_{(i,j) \in \Ic} Q_{ij}^{-1}}$, and so by the continuity of $x \mapsto \expit(\log(x))$,
\begin{align*}
    \lim_{\gamma \to 0} g_\gamma^{(1)}(\Ic) = \expit \left( \log \left( \frac{1}{\frac{1}{|\Ic|} \sum_{(i,j) \in \Ic} Q_{ij}^{-1}} \right) \right) = g_0^{(1)}(\Ic).
\end{align*}
Thus, we have shown that $\lim_{\gamma \to 0} g^{(s)}_\gamma(\Ic) = g_0^{(s)}(\Ic)$ for all fixed $\Ic \subset [n]^2$. Next,
\begin{align*}
    P \left( \left| g_\gamma^{(s)}(\Icskl) - g_0^{(s)}(\Icskl) \right| > \epsilon \right) &\le P \left( \max_{\Ic \subset [n]^2} \left| g_\gamma^{(s)}(\Ic) - g_0^{(s)}(\Ic) \right| > \epsilon \right) \\
    &\le \sum_{\Ic \subset [n]^2} P \left( \left| g_\gamma^{(s)}(\Ic) - g_0^{(s)}(\Ic) \right| > \epsilon \right),
\end{align*}
where the last inequality follows from a union bound. Taking limits, we have
\begin{align}
    \lim_{\gamma \to 0} P \left( \left| g_\gamma^{(s)}(\Icskl) - g_0^{(s)}(\Icskl) \right| > \epsilon \right) &\le \lim_{\gamma \to 0} \sum_{\Ic \subset [n]^2} P \left( \left| g_\gamma^{(s)}(\Ic) - g_0^{(s)}(\Ic) \right| > \epsilon \right) \nonumber \\
    &= \sum_{\Ic \in [n]^2} \left( \lim_{\gamma \to 0} P \left( \left| g_\gamma^{(s)}(\Ic) - g_0^{(s)}(\Ic) \right| > \epsilon \right) \right) = 0, \label{eq:prop11proof_g_gamma_tilde}
\end{align}
where the last equality follows from the fact that $[n]^2$ has finite cardinality, and because \\
$\lim_{\gamma \to 0} P \left( \left| g_\gamma^{(s)}(\Ic) - g_0^{(s)}(\Ic) \right| > \epsilon \right) = 0$ for all $\Ic \subset [n]^2$. Next, for all $\epsilon > 0$, we have
\begin{align*}
    P \left( \left| g_0^{(s)}(\Icskl) - g_0^{(s)}(\Icskltil) \right| > \epsilon \right) &\le P \left( g_0^{(s)}(\Icskl) \ne g_0^{(s)}(\Icskltil) \right) \le P \left( \Icskl \ne \Icskltil \right).
\end{align*}
Hence, taking limits we have
\begin{align}
    \lim_{\gamma \to 0} P \left( \left| g_0^{(s)}(\Icskl) - g_0^{(s)}(\Icskltil) \right| > \epsilon \right) \le \lim_{\gamma \to 0} P \left( \Icskl \ne \Icskltil \right) = 0,
    \label{eq:prop11proof_g_zero_tilde}
\end{align}
where the equality on the right-hand side of \eqref{eq:prop11proof_g_zero_tilde} follows by combining the assumption that $\lim_{\gamma \to 0} P (\Ickl \ne \Ickltil) = 0$, together with the observation that under the fission procedure of Proposition~\ref{prop:univariate_bernoulli_fission}, $\lim_{\gamma \to 0} P(\exists (i,j) : \Aijtr \ne A_{ij}) = 0$. Next, we have
\begin{align*}
    P \Big( \Big| g_\gamma^{(s)}(\Icskl) &- g_0^{(s)}(\Icskltil) \Big| > \epsilon \Big) = P \left( \left| g_\gamma^{(s)}(\Icskl) - g_0^{(s)}(\Icskl) + g_0^{(s)}(\Icskl) - g_0^{(s)}(\Icskltil) \right| > \epsilon \right) \\
    &\le P \left( \left| g_\gamma^{(s)}(\Icskl) - g_0^{(s)}(\Icskl) \right| + \left| g_0^{(s)}(\Icskl) - g_0^{(s)}(\Icskltil) \right| > \epsilon \right) \\
    &\le P \left( \left\{ \left| g_\gamma^{(s)}(\Icskl) - g_0^{(s)}(\Icskl) \right| \ge \frac{\epsilon}{2} \right\} \cup \left\{ \left| g_0^{(s)}(\Icskl) - g_0^{(s)}(\Icskltil) \right| \ge \frac{\epsilon}{2} \right\} \right) \\
    &\le P \left( \left| g_\gamma^{(s)}(\Icskl) - g_0^{(s)}(\Icskl) \right| \ge \frac{\epsilon}{2} \right) + P \left( \left| g_0^{(s)}(\Icskl) - g_0^{(s)}(\Icskltil) \right| \ge \frac{\epsilon}{2} \right).
\end{align*}
Taking limits, this implies that
\begin{align}
    &\lim_{\gamma \to 0} P \Big( \Big| g_\gamma^{(s)}(\Icskl) - g_0^{(s)}(\Icskltil) \Big| > \epsilon \Big) \nonumber \\
    &\le \left( \lim_{\gamma \to 0} P \left( \left| g_\gamma^{(s)}(\Icskl) - g_0^{(s)}(\Icskl) \right| \ge \frac{\epsilon}{2} \right) \right) + \left( \lim_{\gamma \to 0} P \left( \left| g_0^{(s)}(\Icskl) - g_0^{(s)}(\Icskltil) \right| \ge \frac{\epsilon}{2} \right) \right) = 0, \label{eq:Phis_to_Phistil}
\end{align}
where the first limit term is $0$ by \eqref{eq:prop11proof_g_gamma_tilde}, and the second limit term is $0$ by \eqref{eq:prop11proof_g_zero_tilde}.

Next, for all $\epsilon > 0$, $P \left( \left| \dfrac{|\Icskl|}{|\Ickl|} - \dfrac{|\Icskltil|}{|\Ickltil|} \right| > \epsilon \right) \le P \left( \dfrac{|\Icskl|}{|\Ickl|} \ne \dfrac{|\Icskltil|}{|\Ickltil|} \right) \le P \left( \Ickl \ne \Ickltil \right)$, and so taking limits and recalling the assumption that $\lim_{\gamma \to 0} P(\Ickl \ne \Ickltil) = 0$, we have
\begin{equation}
    \lim_{\gamma \to 0} P \left( \left| \dfrac{|\Icskl|}{|\Ickl|} - \dfrac{|\Icskltil|}{|\Ickltil|} \right| > \epsilon \right) = 0.
    \label{eq:Icskl_conv_Icskltil}
\end{equation}

To combine the convergence results across $s \in \{0, 1\}$, we apply Slutsky's theorem. Slutsky's theorem is commonly stated as a result involving discretely-indexed sequences of random variables, so we will make use of the general fact that $\lim_{a \to a_0} h(a) = L$ if and only if $\lim_{m \to \infty} h(a_m) = L$ for all sequences $(a_m)_{m=1}^\infty$ such that $\lim_{m \to \infty} a_m = a_0$.

Fix any sequence $(\gamma_m)_{m=1}^\infty$ such that $\lim_{m \to \infty} \gamma_m = 0$. Then, for all $\epsilon > 0$, from \eqref{eq:Phis_to_Phistil},
\begin{equation*}
    \lim_{m \to \infty} P \left( \left| g_{\gamma_m}(\Icskl) - g_0^{(s)}(\Icskltil) \right| > \epsilon \right) = 0
\end{equation*}
Equivalently, $g^{(s)}_{\gamma_m}(\Icskl) \overset{\text{p}}{\rightarrow} g_0(\Icskltil)$ as $m \to \infty$. Then, let us emphasize that the distributions of $\Icskl$ and $\Ickl$ depend on $\gamma$, and consider the sequence of random variables $\frac{|\Icskl(\gamma_m)|}{|\Ickl(\gamma_m)|} - \frac{|\Icskltil|}{|\Ickltil|}$. From \eqref{eq:Icskl_conv_Icskltil}, it follows that
\begin{equation*}
    \lim_{m \to \infty} P \left( \left| \frac{|\Icskl(\gamma_m)|}{|\Ickl(\gamma_m)|} - \frac{|\Icskltil|}{|\Ickltil|} \right| > \epsilon \right) = 0.
\end{equation*}
Equivalently, we have $\frac{|\Icskl(\gamma_m)|}{|\Ickl(\gamma_m)|} \overset{\text{p}}{\rightarrow} \frac{|\Icskltil|}{|\Ickltil|}$ as $m \to \infty$. So, invoking Slutsky's theorem with these results, we have that
\begin{align*}
    \frac{|\Iczkl(\gamma_m)|}{|\Ickl(\gamma_m)|} g^{(0)}_{\gamma_m}(\Iczkl(\gamma_m)) &+ \frac{|\Icokl(\gamma_m)|}{|\Ickl(\gamma_m)|} g^{(1)}_{\gamma_m}(\Icokl(\gamma_m)) \overset{\text{p}}{\longrightarrow} \frac{|\Iczkltil|}{|\Ickltil|} g^{(0)}_0(\Iczkltil) + \frac{|\Icokltil|}{|\Ickltil|} g^{(1)}_0(\Icokltil)
\end{align*}
as $m \to \infty$. Equivalently, we have
\begin{align*}
    \lim_{m \to \infty} P \Bigg( \Bigg| &\Bigg( \frac{|\Iczkl(\gamma_m)|}{|\Ickl(\gamma_m)|} g^{(0)}_{\gamma_m}(\Iczkl(\gamma_m)) + \frac{|\Icokl(\gamma_m)|}{|\Ickl(\gamma_m)|} g^{(1)}_{\gamma_m}(\Icokl(\gamma_m)) \Bigg) \\
    &- \Bigg( \frac{|\Iczkltil|}{|\Ickltil|} g^{(0)}_0(\Iczkltil) + \frac{|\Icokltil|}{|\Ickltil|} g^{(1)}_0(\Icokltil) \Bigg) \Bigg| > \epsilon \Bigg) = 0
\end{align*}
for all $\epsilon > 0$. Recalling that the sequence $(\gamma_m)_{m =1}^\infty$ was arbitrary, this implies that 
\begin{align*}
    \lim_{\gamma \to 0} P \Bigg( \Bigg| &\Bigg( \frac{|\Iczkl|}{|\Ickl|} g^{(0)}_{\gamma}(\Iczkl) + \frac{|\Icokl|}{|\Ickl|} g^{(1)}_{\gamma}(\Icokl) \Bigg) - \Bigg( \frac{|\Iczkltil|}{|\Ickltil|} g^{(0)}_0(\Iczkltil) + \frac{|\Icokltil|}{|\Ickltil|} g^{(1)}_0(\Icokltil) \Bigg) \Bigg| > \epsilon \Bigg) = 0.
\end{align*}
To conclude the proof, recall the definitions of $g_\gamma^{(0)}$ and $g_\gamma^{(1)}$ in \eqref{eq:g0_def_prop_gamma0}, and the definitions of $\Phi_{k \ell}(\Atr)$ in \eqref{eq:Phi_kl_def} and $\tilde{\Phi}(A)$ in the statement of Proposition~\ref{prop:gamma_to_zero}.

\subsection{Proof of Proposition \ref{prop:squiggle_estimation}}
\label{app:squiggle_proposition_proof}

To begin, we note that applying Proposition~\ref{prop:univariate_bernoulli_fission} to $A_{n,ij} \ind \text{Bernoulli}(M_{n,ij})$ yields $\Anteij \mid \Antrij \ind \text{Bernoulli}(T_{n,ij})$, where $T_{n,ij}$ was defined in Proposition~\ref{prop:univariate_bernoulli_fission}(\ref{prop:fission_define_T}). In what follows, we refer to this as the ``true'' model, $G_n$. That is, the true model $G_n$ is
\begin{equation}
    G_n: \Anteij \mid \Antrij \ind \operatorname{Bernoulli}(T_{n,ij}),
    \label{eq:bernoulli_true_model_prop_8_start}
\end{equation}
We now introduce a misspecified (``working'') model $F^{(s)}_{n, k \ell}$ that assumes the presence of the communities characterized by $\Zntr = \hat{z}_n$, splitting into the two cases where $\Aijtr = 0$ or $\Aijtr = 1$. That is, for $s \in \{0, 1\}$, 
\begin{equation*}
  F^{(s)}_{n, k \ell}: \Anteij \mid \{\Antrij = s\} \ind \operatorname{Bernoulli}(\psi^{(s)}_{n, k(i) \ell(j)}),
\end{equation*}
where \(k(i)\) returns the value of \(k\) for which \(\hat{z}_{n, ik} = 1\) and \(\ell(j)\) returns the value of \(\ell\) for which \(\hat{z}_{n, j \ell} = 1\). For a given \(n\), under \(F_{n, k \ell}^{(s)}\) a dyad \((i, j)\) in community pair \((k, \ell)\) and where $\Antrij = s$ (i.e., $\ (i, j) \in \Ic^{(s)}_{n,k\ell} := \{(i,j) ~:~ \hat{z}_{n,ik} = 1, \hat{z}_{n,j \ell} = 1, \Antrij = s\}$)
has a conditional log-likelihood and derivatives (up to constants) 
\begin{align}
    \ell_{n,ij}(\psi^{(s)}_{n, k \ell} \mid \Antr) &= \Anteij \log(\psi^{(s)}_{n, k\ell}) + (1 - \Anteij) \log(1 - \psi^{(s)}_{n, k \ell}), \nonumber \\
    \ell'_{n,ij}(\psi^{(s)}_{n, k \ell} \mid \Antr) &= \dfrac{\Anteij}{\psi^{(s)}_{n, k \ell}} + \dfrac{\Anteij - 1}{1 - \psi^{(s)}_{n, k \ell}}, \nonumber \\
    \ell''_{n,ij}(\psi^{(s)}_{n, k \ell} \mid \Antr) &= -\dfrac{\Anteij}{\psi^{(s)2}_{n, k \ell}} - \dfrac{1 - \Anteij}{(1 - \psi^{(s)2}_{n, k \ell})}.
\end{align}
The maximum likelihood estimator of $\psi^{(s)}_{n, k \ell}$ under \(F_{n, k \ell}^{(s)}\) is 
\begin{equation}
    \hat{B}^{(s)}_{n, k \ell} := \argmax_{\psi \in [0, 1]} \left\{ \sum_{(i,j) \in \Ic^{(s)}_{n,k \ell}} \ell_{n,ij}(\psi \mid \Antr) \right\} = \dfrac{1}{|\Ic^{(s)}_{n,k \ell}|} \sum_{(i,j) \in \Ic^{(s)}_{n,k \ell}} \Anteij.
    \label{eq:bernoulli_MLE_misspecified_prop_8}
\end{equation}
Next, we define $B^{(s)}_{n,k \ell}$ to be the value that maximizes the expected log-likelihood of the misspecified model \(F_{n, k \ell}^{(s)}\) under the true model \(G_{n}\):
\begin{align}
    B^{(s)}_{n,k \ell} &:= \argmax_{\psi \in (0, 1)} \E_{G_n} \left[ \sum_{(i,j) \in \Ic^{(s)}_{n,k \ell}} \ell_{n,ij}(\psi \mid \Antr) ~\Bigg|~ \Antr \right] = \dfrac{1}{|\Ic^{(s)}_{n,k \ell}|} \sum_{(i,j) \in \Ic^{(s)}_{n,k \ell}} T_{n,ij}, \label{eq:misspecified_target_poisson}
\end{align}
where the final equality follows from \eqref{eq:bernoulli_true_model_prop_8_start}. Note that $\E_{G_n}[\hat{B}^{(s)}_{n,k \ell} \mid \Antr] = B^{(s)}_{n, k \ell}$ and
\begin{equation}
    \tau^{(s)}_{n, k \ell} := \Var_{G_n}(\hat{B}^{(s)}_{n, k \ell} \mid \Antr) = \dfrac{1}{|\Ic^{(s)}_{n, k \ell}|^2} \sum_{(i,j) \in \Ic^{(s)}_{n, k \ell}} T_{n,ij}(1 - T_{n,ij}).
    \label{eq:tau_variance_squiggle}
\end{equation}
Because we assume $|\Icnskl|^{-1} = O(n^{-2})$ for the sequence of realizations $\{\Antr = \antr\}_{n=1}^\infty$, by Lemma \ref{lemma:mixing_clt} we have
\begin{equation}
    (\tau_{n, k \ell}^{(s)})^{-1/2} (\hat{B}^{(s)}_{n, k \ell} - B^{(s)}_{n, k \ell}) \mid \{\Antr = \antr \} \overset{\text{d}}{\rightarrow} \mathcal{N}(0,1).
    \label{eq:pre_delta_bernoulli_proof}
\end{equation}
Next, we apply a Taylor expansion to the left hand side of \eqref{eq:pre_delta_bernoulli_proof}. Define the function 
$h_{k \ell}^{(s)}(z) := \expit(\logit(z) - c^{(s)}),$
where $c^{(0)} :=   \log(\gamma/(1-\gamma)  )$ and  $c^{(1)} :=   \log((1-\gamma) / \gamma)$.
The function $h^{(s)}_{k \ell}$ is differentiable on $(0,1)$ with derivative $h'^{(s)}_{k \ell}(z) = \frac{\expit(\logit(z) - c^{(s)})}{1 + \exp(\logit(z) - c^{(s)})} \cdot \dfrac{1}{z(1-z)}$. Then, by the mean value theorem, $h^{(s)}_{k \ell} (\hat{B}^{(s)}_{n, k \ell}) = h^{(s)}_{k \ell}(B^{(s)}_{n, k \ell}) + h'^{(s)}_{k \ell} (\tilde{B}^{(s)}_{n, k \ell})(\hat{B}^{(s)}_{n, k \ell} - B^{(s)}_{n, k \ell})$, where $\tilde{B}^{(s)}_{n, k \ell}$ is a random variable between $\hat{B}^{(s)}_{n, k \ell}$ and $B^{(s)}_{n, k \ell}$. Rearranging, we have
\begin{equation}
     (\tau^{(s)}_{n, k \ell})^{-1/2} (\hat{B}^{(s)}_{n, k \ell} - B^{(s)}_{n, k \ell}) = (\tau^{(s)}_{n, k \ell})^{-1/2} \left( \dfrac{h_{k \ell}^{(s)}(\hat{B}^{(s)}_{n, k \ell}) - h_{k \ell}^{(s)}(B^{(s)}_{n, k \ell})}{h'^{(s)}_{k \ell}(\tilde{B}^{(s)}_{n, k \ell})} \right).
    \label{eq:delta_step_squiggle}
\end{equation}

For simplicity, until the end of this proof we will implicitly condition on $\{\Antr = \antr\}$ rather than writing this out explicitly.

Now, we show that the conditions of Lemma~\ref{lemma:alternative_cmt} hold and invoke it to argue that $\frac{h_{k \ell}'^{(s)}(\tilde{B}^{(s)}_{n, k \ell})}{h_{k \ell}'^{(s)}(B^{(s)}_{n, k \ell})} \overset{\text{p}}{\rightarrow} 1$. First, because $\hat{B}^{(s)}_{n, k \ell} - B^{(s)}_{n, k \ell} = o_p(1)$, it follows that $\tilde{B}^{(s)}_{n, k \ell} - B^{(s)}_{n, k \ell} = o_p(1)$.
Next, by the assumption that $0 < N_0 \le M_{n,ij} \le N_1 < 1$, it follows that $0 < \tilde{N}_0 \le T_{n,ij} \le \tilde{N}_1 < 1$ for constants $\tilde{N}_0$ and $\tilde{N}_1$, and so $0 < \tilde{N}_0 \le B^{(s)}_{n, k \ell} \le \tilde{N}_1 < 1$.
For a small constant $\epsilon > 0$ such that $0 < \tilde{N}_0 - \epsilon < \tilde{N}_0$ and $\tilde{N}_1 < \tilde{N}_1 + \epsilon < 1$, define the compact set $U := [\tilde{N}_0 - \epsilon, \tilde{N}_1 + \epsilon]$. Then, because $h'^{(s)}_{k \ell}$ is continuous on $(0,1)$, it is uniformly continuous on the compact set $U$. Next, because $B^{(s)}_{n, k \ell} \in \textnormal{int}(U)$ for all $n$ where $\textnormal{int}(U)$ is the interior of $U$ and $\tilde{B}^{(s)}_{n, k \ell} - B^{(s)}_{n, k \ell} = o_p(1)$, it follows that $\lim_{n \to \infty} P(\hat{B}^{(s)}_{n, k \ell} \in U) = 1$. So, by Lemma~\ref{lemma:alternative_cmt} we conclude that $h'^{(s)}_{k \ell}(\tilde{B}^{(s)}_{n, k \ell}) - h'^{(s)}_{k \ell}(B^{(s)}_{n, k \ell}) = o_p(1)$. Finally, noting that $0 < \breve{N}_0 \le  h'^{(s)}_{k \ell}(B^{(s)}_{n, k \ell}) \le \breve{N}_1 < 0$ for constants $\breve{N}_0$ and $\breve{N}_1$, we can divide both sides of this convergence by $h'^{(s)}_{k \ell}(B^{(s)}_{n, k \ell})$ to achieve $\frac{h'^{(s)}_{k \ell}(\tilde{B}^{(s)}_{n, k \ell})}{h'^{(s)}_{k \ell}(B^{(s)}_{n, k \ell})} - 1 = o_p(1)$. Equivalently,  $\frac{h_{k \ell}'^{(s)}(\tilde{B}^{(s)}_{n, k \ell})}{h_{k \ell}'^{(s)}(B^{(s)}_{n, k \ell})} \overset{\text{p}}{\rightarrow} 1$. Using this convergence in tandem with Equation~\eqref{eq:delta_step_squiggle}, by Slutsky's theorem we conclude that
\begin{equation}
    \left( \tau^{(s)}_{n, k \ell} h'^{(s)}_{k \ell}(B^{(s)}_{n,k \ell})^2 \right)^{-1/2} \left( h_{k \ell}^{(s)}(\hat{B}^{(s)}_{n, k \ell}) - h_{k \ell}^{(s)}(B^{(s)}_{n, k \ell}) \right) \overset{\text{d}}{\rightarrow} \mathcal{N}(0,1).
    \label{eq:individual_convergence_squiggle_uglier}
\end{equation}
To simplify notation, defining
\begin{align}
    \hat{V}^{(s)}_{n, k \ell} &:= h_{k \ell}^{(s)}(\hat{B}^{(s)}_{n, k \ell}), \label{eq:vhat} \\
    V^{(s)*}_{n, k \ell} &:= h^{(s)}_{k \ell}(B^{(s)}_{n, k \ell}), \label{eq:vstar} \\
    \zeta^{(s)}_{n, k \ell} &:= \tau^{(s)}_{n, k \ell} h'^{(s)}_{k \ell}(B^{(s)}_{n, k \ell})^2, \label{eq:zeta_s_def}
\end{align}
we rewrite the statement in \eqref{eq:individual_convergence_squiggle_uglier} as
\begin{equation}
    \left( \zeta^{(s)}_{n, k \ell} \right)^{-1/2} \left( \hat{V}^{(s)}_{n, k \ell} - V^{(s)*}_{n, k \ell} \right) \overset{\text{d}}{\rightarrow} \mathcal{N}(0,1).
    \label{eq:individual_convergence_squiggle}
\end{equation}
The convergence in \eqref{eq:individual_convergence_squiggle} summarizes the convergence result among the dyads $(i,j)$ such that $\Aijtr = s$ for a \emph{single community pair} indexed by a given $(k, \ell) \in \{1, 2, \dots, K\}^2$.

To treat \textit{all} dyads for a given community pair $(k, \ell) \in \{1, 2, \dots, K\}^2$, define the quantity $\zeta_{n, k \ell} := \sum_{s \in \{0,1\}} \dfrac{|\Icnskl|^2}{|\Icnkl|^2} \zeta^{(s)}_{n, k \ell}$, the quantity $\hat{\Phi}_{n, k \ell} := \dfrac{|\Icnzkl|}{|\Icnkl|} \hat{V}^{(0)}_{n, k \ell} + \dfrac{|\Icnokl|}{|\Icnkl|} \hat{V}^{(1)}_{n, k \ell}$, as well as $\Phi_{n, k \ell} := \dfrac{|\Icnzkl|}{|\Icnkl|} V^{(0)*}_{n, k \ell} + \dfrac{|\Icnokl|}{|\Icnkl|} V^{(1)*}_{n, k \ell}$. Also defining $U^{(s)}_{n, k \ell} := \left( \zeta^{(s)}_{n, k \ell} \right)^{-1/2} \left( \hat{V}^{(s)}_{n, k \ell} - V^{(s)*}_{n, k \ell} \right)$ for $s \in \{0, 1\}$, by the independence of $\hat{V}_{n, k \ell}^{(0)}$ and $\hat{V}_{n, k \ell}^{(1)}$, and by \eqref{eq:individual_convergence_squiggle}, we have $\begin{bmatrix}U_{n, k \ell}^{(0)} \\ U_{n, k \ell}^{(1)}\end{bmatrix} \overset{\text{d}}{\rightarrow} \mathcal{N}_2 \left( 0, I_2\right)$. Also define
\begin{align*}
    W_{n,k \ell} &:= \dfrac{\frac{|\Icnzkl|}{|\Icnkl|} (\hat{V}^{(0)}_{n, k \ell} - V^{(0)*}_{n, k \ell}) + \frac{|\Icnokl|}{|\Icnkl|} (\hat{V}^{(1)}_{n, k \ell} - V^{(1)*}_{n, k \ell})}{\sqrt{\frac{|\Icnzkl|^2}{|\Icnkl|^2} \zeta^{(0)}_{n, k \ell} + \frac{|\Icnokl|^2}{|\Icnkl|^2} \zeta^{(1)}_{n, k \ell}}} = a^{(0)}_{n, k \ell} U^{(0)}_{n, k \ell} + a^{(1)}_{n, k \ell} U^{(1)}_{n, k \ell},
\end{align*}
where $a^{(s)}_{n, k \ell} := \dfrac{\frac{|\Icnzkl|}{|\Icnkl|} (\zeta_n^{(s)})^{1/2}}{\sqrt{\frac{|\Icnzkl|^2}{|\Icnkl|^2} \zeta^{(0)}_{n, k \ell} + \frac{|\Icnokl|^2}{|\Icnkl|^2} \zeta^{(1)}_{n, k \ell}}}$, which satisfies $(a^{(0)}_{n, k \ell})^2 + (a^{(1)}_{n, k \ell})^2 = 1$.

Now, take any subsequence given by $(n_m)_{m=1}^\infty$. Because $(a_{n_m, k \ell}^{(0)}, a_{n_m, k \ell}^{(1)})$ lies on a compact set (the unit circle in $\mathbb{R}^2$), there exists a further subsequence $n_{m_r}$ such that $(a_{n_{m_r}, k \ell}^{(0)}, a_{n_{m_r}, k \ell}^{(1)}) \to (a^{(0)}, a^{(1)})$, where $(a^{(0)})^2 + (a^{(1)})^2 = 1$ by the fact that the unit circle is a closed subspace of $\mathbb{R}^2$, consequently contains its limit points. Along this further subsequence, by Slutsky's theorem and the continuous mapping theorem, we have $W_{n_{m_r}, k \ell} \overset{\text{d}}{\rightarrow} \mathcal{N}(0,1)$. (\textit{Note that in the proof of Proposition~\ref{prop:poisson_estimation}, this particular step was not necessary, as Proposition~\ref{prop:poisson_estimation} does not involve splitting and combining results across $\Icnzkl$ and $\Icnokl$, and instead is able to directly establish a result for $\Icnkl$.})

We have established that each subsequence $W_{n_m, k \ell}$ has a further subsequence $W_{n_{m_r}, k \ell}$ that converges in distribution to $\mathcal{N}(0,1)$. Letting $P_n$ denote the probability measure of $W_{n, k \ell}$, and letting $P$ denote the probability measure of a $\mathcal{N}(0,1)$ random variable, this means that each subsequence $P_{n_m}$ contains a further subsequence $P_{n_{m_r}}$ that converges weakly to $P$ as $r \to \infty$. Thus, Theorem 2.6 of \citet{billingsley1999ConvergenceProbabilityMeasures} ensures that $P_n$ converges weakly to $P$. Therefore, $W_{n, k \ell} \overset{\text{d}}{\rightarrow} \mathcal{N}(0,1)$. Then, recognizing by algebraic manipulation that $W_{n, k \ell} = (\zeta_{n, k \ell})^{-1/2} (\hat{\Phi}_{n, k \ell} - \Phi_{n, k \ell})$, we conclude that
\begin{equation}
    \left( \zeta_{n, k \ell} \right)^{-1/2} \left( \hat{\Phi}_{n, k \ell} - \Phi_{n, k \ell} \right) \overset{\text{d}}{\rightarrow} \mathcal{N}(0,1).
    \label{eq:convergence_phi}
\end{equation}

Noting that the collection $\left\{ \hat{\Phi}_{n, k \ell} \right\}_{k = 1, \ell=1}^K$ is mutually independent, and defining $\hat{\Phi}_n$ and $\Phi_{n}$ to be the $K \times K$ matrices whose $(k,l)$th entries are $\hat{\Phi}_{n, k \ell}$ and $\Phi_{n, k \ell}$ respectively, by Lemma~\ref{lemma:marginal_independence_joint_convergence} we have
\begin{align}
   \left( \Xi_n \right)^{-1/2} \left( \textnormal{vec} \left( \hat{\Phi}_n \right) - \textnormal{vec} \left( \Phi_n \right) \right) \overset{\text{d}}{\longrightarrow} \mathcal{N}_{K^2}(0, I_{K^2}),
   \label{eq:squiggle_almost_end_convergence_partial}
\end{align}
where $\Xi_n := \textnormal{diag}(\textnormal{vec}(\zeta_{n}))$, and $\zeta_n \in \mathbb{R}^{K \times K}$ is defined entry-wise as $\zeta_{n, k \ell}$.

In the setting of Proposition \ref{prop:squiggle_estimation}, the parameter of interest and estimator take the form
\begin{align}
    \hat{\xi}_n &:= u_n^\top \textnormal{vec} \left( \hat{\Phi}_n \right),
    \label{eq:squiggle_hat_proof_def} \\
    \xi_n &:= u_{n}^{\top} \operatorname{vec} \left( \Phi_n \right) \label{eq:squiggle_target_proof_def}.
\end{align}
Defining $\omega_n^2 := u_n^\top \Xi_n u_n$, by Lemma \ref{lemma:gaussian_rotation}, we have
\begin{equation}
    \begin{split}
        (u_n^\top \Xi_n u_n)^{-1/2} &u_n^\top \left( \operatorname{vec}(\hat{\Phi}_n) - \operatorname{vec}(\Phi_n) \right) = \dfrac{\hat{\xi}_n - \xi_n}{\omega_n} \overset{\text{d}}{\longrightarrow} \mathcal{N}(0,1).
    \end{split}
    \label{eq:bernoulli_conv_theoretical}
\end{equation}

Recall that $\Xi_n := \textnormal{diag}(\textnormal{vec}(\zeta_n))$ where $\zeta_{n, k \ell} := \sum_{s \in \{0, 1\}} \frac{|\Icnskl|^2}{|\Icnkl|^2} \zeta^{(s)}_{n, k \ell}$, and where $\zeta^{(s)}_{n, k \ell}$ was defined in \eqref{eq:zeta_s_def} and is proportional to $\tau^{(s)}_{n, k \ell}$. To provide an upper bound on $\zeta^{(s)}_{n, k \ell}$, first note by Jensen's inequality and the convexity of $x \mapsto x^2$ that
\begin{align}
    \tau^{(s)}_{n, k \ell} &:= \dfrac{1}{|\Icnskl|^2} \sum_{(i,j) \in \Icnskl} T_{n,ij}(1 - T_{n,ij}) \nonumber \\
    &= \dfrac{1}{|\Icnskl|} \left[ \dfrac{1}{|\Icnskl|} \sum_{(i,j) \in \Icnskl} T_{n,ij} - \sum_{(i,j) \Icnskl} \dfrac{1}{|\Icnskl|} T_{n,ij}^2  \right] \nonumber \\
    &\le \dfrac{1}{|\Icnskl|} \left[ \dfrac{1}{|\Icnskl|} \sum_{(i,j) \in \Icnskl} T_{n,ij} - \left( \dfrac{1}{|\Icnskl|} \sum_{(i,j) \Icnskl} T_{n,ij} \right)^2 \right] = \dfrac{B^{(s)}_{n, k \ell}(1 - B^{(s)}_{n, k \ell})}{|\Icnskl|} \nonumber \\
    &=: \tilde{\tau}^{(s)}_{n, k \ell}. \label{eq:tau_tilde_tau_bound_proof}
\end{align}
Then, define $\Delta^{(s)}_{n, k \ell} = \tilde{\tau}^{(s)}_{n, k \ell} \cdot (h'^{(s)}_{k \ell}(B^{(s)}_{n, k \ell}))^2 = \frac{B^{(s)}_{n, k \ell}(1 - B^{(s)}_{n, k \ell}) e^{2c^{(s)}}}{|\Icnskl| ((1-B^{(s)}_{n, k \ell}) e^{c^{(s)}} + B^{(s)}_{n, k \ell})^4}$, and note that $\zeta^{(s)}_{n, k \ell} \le \Delta^{(s)}_{n, k \ell}$. Further define $\Delta_{n, k \ell} := \sum_{s \in \{0,1\}} \frac{|\Icnskl|^2}{|\Icnkl|^2} \Delta^{(s)}_{n, k \ell}$ and $\Sigma_n := \textnormal{diag}(\textnormal{vec}(\Delta_n))$. It follows that
\begin{equation}
    \omega^2_n = u_n^\top \Xi_n u_n \le u_n^\top \Sigma_n u_n =: \sigma^2_n.
    \label{eq:sigma_sq_def_squiggle_proof}
\end{equation}
The variance estimate is $\hat{\sigma}^2 := u_n^\top \hat{\Sigma}_n u_n$, where $\hat{\Sigma}_n := \textnormal{diag}(\textnormal{vec}(\hat{\Delta}_n))$, $\hat{\Delta}_{n, k \ell} := \sum_{s \in \{0,1\}} \frac{|\Icnskl|^2}{|\Icnkl|^2} \hat{\Delta}^{(s)}_{n, k \ell}$, and $\hat{\Delta}^{(s)}_{n, k \ell} := \frac{\hat{B}^{(s)}_{n, k \ell}(1 - \hat{B}^{(s)}_{n, k \ell}) e^{2c^{(s)}}}{|\Icnskl| ((1-\hat{B}^{(s)}_{n, k \ell}) e^{c^{(s)}} + \hat{B}^{(s)}_{n, k \ell})^4}$. Next, we show that the conditions of Lemma~\ref{lemma:asymptotic_variance_un_equal} hold, so that 
$\dfrac{\sigma_n}{
\hat{\sigma}_n} = \dfrac{(u_n \hat{\Sigma}_n u_n)^{-1/2}}{(u_n \Sigma_n u_n)^{-1/2}} \overset{\text{p}}{\rightarrow} 1$. Define the function $\varphi_s: x \mapsto \dfrac{x(1-x)e^{2c^{(s)}}}{((1-x)e^{c^{(s)}} + x)^4}$, so that $|\Icnskl| \hat{\Delta}_{n, k \ell}^{(s)} = \varphi_s(\hat{B}^{(s)}_{n, k \ell})$ and $|\Icnskl| \Delta_{n, k \ell}^{(s)} = \varphi_s(B^{(s)}_{n, k \ell})$. Then, by the assumed bound $0 < N_0 \le M_{n,ij} \le N_1 < 1$, similar to an argument made previously in this proof, we can construct a compact set $U = [u_0, u_1] \subset (0,1)$ with $0 < u_0$ satisfying $B^{(s)}_{n, k \ell} \in U$ for all $n$ and $\lim_{n \to \infty} P(\hat{B}^{(s)}_{n, k \ell} \in U) = 1$. Because $\hat{B}^{(s)}_{n, k \ell} - B^{(s)}_{n, k \ell} = o_p(1)$, applying Lemma~\ref{lemma:alternative_cmt} we obtain
\begin{equation}
    |\Icnskl| (\hat{\Delta}^{(s)}_{n, k \ell} - \Delta^{(s)}_{n, k \ell}) = \varphi_s(\hat{B}^{(s)}_{n, k \ell}) - \varphi_s(B^{(s)}_{n, k \ell}) \overset{\text{p}}{\rightarrow} 0.
    \label{eq:difference_var_phi_to_0}
\end{equation}
Because $|\Icnzkl|^{-1} = O(n^{-2})$ and $|\Icnokl|^{-1} = O(n^{-2})$ and $|\Icnzkl| + |\Icnokl| = |\Icnkl|$, we have that $|\Icnkl|^{-1} = O(n^{-2})$. So,
\begin{equation}
    0 < \liminf_{n \to \infty} \dfrac{|\Icnokl|}{|\Icnkl|} \le \limsup_{n \to \infty} \dfrac{|\Icnokl|}{|\Icnkl|} < 1.
    \label{eq:liminf_ratio_zeros}
\end{equation}
The statement in \eqref{eq:liminf_ratio_zeros} also holds for an arbitrary subsequence $(n_m)_{m=1}^\infty$ as $m \to \infty$, so from \eqref{eq:liminf_ratio_zeros}, so we now use the fact that whenever the limit inferior or superior exists as a real number $\rho$, we can always find a further subsequence such that the \textit{limit} of the further subsequence is $\rho$. Hence, there exists a further subsequence $n_{m_r}$ such that $\lim_{r \to \infty} \frac{|\Ic^{(1)}_{n_{m_r}, k \ell}|}{|\Ic_{n_{m_r}, k \ell}|} = \rho$ and $\lim_{r \to \infty} \frac{|\Ic^{(0)}_{n_{m_r}, k \ell}|}{|\Ic_{n_{m_r}, k \ell}|} = 1-\rho$ for some $0 < \rho < 1$. Because $\varphi_s(B^{(s)}_{n_{m_r}, k \ell})$ is bounded from below by $0$ and bounded from above by a constant, we can find a further subsequence $(n_{m_{r_q}})_{q=1}^{\infty}$ such that $\varphi_s(B^{(s)}_{n_{m_{r_q}}, k \ell}) \to \lambda$ as $q \to \infty$ for some $0 < \lambda < \infty$.

For simplicity, let us relabel the subsequence $(n_{m_{r_q}})_{q=1}^{\infty}$ as $(n_{m_p})_{p=1}^{\infty}$. Using the fact that limits are preserved under subsequences, we have $\lim_{p \to \infty} \frac{|\Ic^{(1)}_{n_{m_p}, k \ell}|}{|\Ic_{n_{m_p}, k \ell}|} = \rho$, $\lim_{p \to \infty} \frac{|\Ic^{(0)}_{n_{m_p}, k \ell}|}{|\Ic_{n_{m_p}, k \ell}|} = 1-\rho$, and $\lim_{p \to \infty} \varphi_s(B^{(s)}_{n_{m_p}, k \ell}) = \lambda$ for some $0 < \lambda < \infty$.

Note that we can express $\frac{\hat{\Delta}_{n_{m_p}, k \ell}}{\Delta_{n_{m_p}, k \ell}}$ as
\begin{align*}
    \dfrac{\hat{\Delta}_{n_{m_p}, k \ell}}{\Delta_{n_{m_p}, k \ell}} &= \dfrac{\sum_{s \in \{0,1\}} \frac{|\Ic^{(s)}_{n_{m_p}, k \ell}|^2}{|\Ic_{n_{m_p}, k \ell}|^2} \hat{\Delta}^{(s)}_{n_{m_p}, k \ell} }{\sum_{s \in \{0,1\}} \frac{|\Ic^{(s)}_{n_{m_p}, k \ell}|^2}{|\Ic_{n_{m_p}, k \ell}|^2} \Delta^{(s)}_{n_{m_p}, k \ell}} = \dfrac{\sum_{s \in \{0, 1\}} \frac{|\Ic^{(s)}_{n_{m_p}, k \ell}|}{|\Ic_{n_{m_p}, k \ell}|}\varphi_s(\hat{B}^{(s)}_{n_{m_p}, k \ell})}{\sum_{s \in \{0, 1\}} \frac{|\Ic^{(s)}_{n_{m_p}, k \ell}|}{|\Ic_{n_{m_p}, k \ell}|}\varphi_s(B^{(s)}_{n_{m_p}, k \ell})} \\
    &= 1 + \dfrac{\sum_{s \in \{0, 1\}} \frac{|\Ic^{(s)}_{n_{m_p}, k \ell}|}{|\Ic_{n_{m_p}, k \ell}|} (\varphi_s(\hat{B}^{(s)}_{n_{m_p}, k \ell}) - \varphi_s(B^{(s)}_{n_{m_p}, k \ell}))}{\sum_{s \in \{0, 1\}} \frac{|\Ic^{(s)}_{n_{m_p}, k \ell}|}{|\Ic_{n_{m_p}, k \ell}|}\varphi_s(B^{(s)}_{n_{m_p}, k \ell})}.
\end{align*}
Along the further subsequence $(n_{m_p})_{p=1}^\infty$, by Slutsky's theorem we have $\frac{\hat{\Delta}_{n_{m_p}, k \ell}}{\Delta_{n_{m_p}, k \ell}} \overset{\text{p}}{\rightarrow} 1$.

We have established that each subsequence $\frac{\hat{\Delta}_{n_m, k \ell}}{\Delta_{n_m, k \ell}}$ has a further subsequence $\frac{\hat{\Delta}_{n_{m_p}, k \ell}}{\Delta_{n_{m_p}, k \ell}}$ that converges in probability to $1$. Note that convergence in probability to a constant is equivalent to convergence in distribution to a constant. Letting $P_n$ denote the probability measure of $\frac{\hat{\Delta}_{n, k \ell}}{\Delta_{n, k \ell}}$, and letting $P$ denote the probability measure of the constant $1$, this means that each subsequence $P_{n_m}$ contains a further subsequence $P_{n_{m_p}}$ that converges weakly to $P$ as $p \to \infty$. Thus, Theorem 2.6 of \citet{billingsley1999ConvergenceProbabilityMeasures} ensures that $P_n$ converges weakly to $P$. Therefore,
\begin{equation}
    \dfrac{\hat{\Delta}_{n, k \ell}}{\Delta_{n, k \ell}} \overset{\text{p}}{\rightarrow} 1
\label{eq:Delta_diagonal_convergence_squiggle}
\end{equation}
for all $k, \ell \in \{1, 2, \dots, K\}$.

By the construction of $\Sigma_n := \textnormal{diag}(\textnormal{vec}(\Delta_n))$ and $\hat{\Sigma}_n := \textnormal{diag}(\textnormal{vec}(\hat{\Delta}_n))$, the result in \eqref{eq:Delta_diagonal_convergence_squiggle} implies that $\hat{\Sigma}_n \Sigma_n^{-1} \overset{\text{p}}{\rightarrow} I_{K^2 \times K^2}$. To show that the final condition of Lemma~\eqref{lemma:asymptotic_variance_un_equal} holds, we can decompose $\Sigma_n = N_n^{-1} \tilde{\Sigma}_n$ and $\hat{\Sigma}_n = N_n^{-1} \hat{\tilde{\Sigma}}_n$, where $N_n := \textnormal{diag}(\textnormal{vec}(\tilde{N}_n))$ where $\tilde{N}_{n, k \ell} := |\Icnkl|$, and where $\tilde{\Sigma}_n := \textnormal{diag}(\textnormal{vec}(\tilde{\Delta}_n))$ and $\hat{\tilde{\Sigma}}_n := \textnormal{diag}(\textnormal{vec}(\hat{\tilde{\Delta}}_n))$, with $\hat{\tilde{\Delta}}_{n, k \ell} := |\Icnkl| \hat{\Delta}_{n, k \ell}$, and 
\begin{align*}
    \tilde{\Delta}_{n, k \ell} &:= |\Icnkl| \Delta_{n, k \ell} = |\Icnkl| \sum_{s \in \{0, 1\}} \dfrac{|\Icnskl|^2}{|\Icnkl|^2} \Delta^{(s)}_{n, k \ell} = \sum_{s \in \{0, 1\}} \dfrac{|\Icnskl|}{|\Icnkl|} \varphi_s(B^{(s)}_{n, k\ell}),
\end{align*}
where the last equality above (and the fact that $0 < b_0 \le \varphi_s(B^{(0)}_{n, k \ell}), \varphi_s(B^{(1)}_{n, k \ell}) \le \tilde{K}_1 < b_1$ for constants $b_0$ and $b_1$) imply that $\tilde{\Delta}_{n, k \ell}$ (and consequently $\tilde{\Sigma}_{n,ii}$ for all $i$) is contained in a compact set $[b_0, b_1]$. So, by Lemma~\ref{lemma:asymptotic_variance_un_equal}, we have $\dfrac{\sigma_n}{\hat{\sigma}_n} = \dfrac{(u_n^\top \hat{\Sigma}_n u_n)^{-1/2}}{(u_n^\top \Sigma_n u_n)^{-1/2}} \overset{\text{p}}{\longrightarrow} 1$. Now, recall from (\ref{eq:bernoulli_conv_theoretical}) that $\frac{\hat{\xi}_n - \xi_n}{\omega_n} = \frac{\sigma_n}{\omega_n} \cdot \frac{\hat{\xi}_n - \xi_n}{\sigma_n} \overset{\text{d}}{\longrightarrow} \mathcal{N}(0,1)$. Thus, by Slutsky's theorem,
\begin{align}
    \dfrac{\sigma_n}{\omega_n} \cdot \dfrac{\hat{\xi}_n - \xi_n}{\hat{\sigma}_n} \overset{\text{d}}{\longrightarrow} \mathcal{N}(0,1).
    \label{eq:idontknowwhattocallthis}
\end{align}
By the definition of convergence in distribution, the cumulative distribution function (CDF) of the left-hand size converges pointwise to the CDF of the $\mathcal{N}(0,1)$ distribution at all continuity points (which is every point in the case of $\mathcal{N}(0,1)$). Hence, denoting $\phi_{1 - \alpha/2}$ to be the $(1 - \alpha/2)$-quantile of the $\mathcal{N}(0,1)$ distribution, \eqref{eq:idontknowwhattocallthis} implies that
\begin{align}
    \lim_{n \to \infty} P \left(-\phi_{1 - \alpha/2} \cdot \dfrac{\sigma_n}{\omega_n} \le \dfrac{\hat{\xi}_n - \xi_n}{\hat{\sigma}_n} \le \phi_{1 - \alpha/2} \cdot \dfrac{\sigma_n}{\omega_n} \right) = 1 - \alpha.
    \label{eq:ialsodontknow}
\end{align}
Because $\frac{\sigma_n}{\omega_n} \le 1$ for all $n$, we have
\begin{align}
    P \left(-\phi_{1 - \alpha/2} \cdot \dfrac{\sigma_n}{\omega_n} \le \dfrac{\hat{\xi}_n - \xi_n}{\hat{\sigma}_n} \le \phi_{1 - \alpha/2} \cdot \dfrac{\sigma_n}{\omega_n} \right) \le P \left(-\phi_{1 - \alpha/2} \le \dfrac{\hat{\xi}_n - \xi_n}{\hat{\sigma}_n} \le \phi_{1 - \alpha/2}\right), \label{ineq:almost_done_squiggle_xi_thing}
\end{align}
which in combination with (\ref{eq:ialsodontknow}) implies that
\begin{align}
    \liminf_{n \to \infty} P \left(\hat{\xi}_n - \phi_{1 - \alpha/2} \cdot \hat{\sigma}_n \le \xi_n \le \hat{\xi}_n + \phi_{1 - \alpha/2} \cdot \hat{\sigma}_n\right) \ge 1 - \alpha. \label{eq:almost_done_bernoulli}
\end{align}
Finally, we return to explicitly writing out the conditioning on $\{\Antr = \antr\}$, and write the arguments of the estimator $\hat{\xi}_n = \hat{\xi}_n(\Ante, \Antr)$ and estimand $\xi_n(\Antr)$ to rewrite (\ref{eq:almost_done_bernoulli}) as
\begin{align}
    \liminf_{n \to \infty} P \Big(\hat{\xi}_n &\left( \Ante, \Antr \right) - \phi_{1 - \alpha/2} \cdot \hat{\sigma}_n \le \xi_n \left( \Antr \right) \\
    &\le \hat{\xi}_n \left( \Ante, \Antr \right) + \phi_{1 - \alpha/2} \cdot \hat{\sigma}_n \mid \Antr = \antr \Big) \ge 1 - \alpha,
    \label{eq:done_bernoulli}
\end{align}
with $\xi(\Antr)$ defined in \eqref{eq:squiggle_target_proof_def}, $\hat{\xi}(\Ante, \Antr)$ in \eqref{eq:squiggle_hat_proof_def}, and $\hat{\sigma}_n$ in \eqref{eq:sigma_sq_def_squiggle_proof}.

\subsection{Proof of Corollary \ref{cor:exact_coverage_bernoulli_target}}
\label{app:proof_of_corollary_to_squiggle}

Under the additional condition given in Corollary~\ref{cor:exact_coverage_bernoulli_target}, for large $n$, we have by Proposition~\ref{prop:taylor}(\ref{prop:taylor_a}) that $\Phi_{k \ell} = B_{k \ell}$ for all $(k, \ell)$ such that the corresponding entry of $u_n \in \mathbb{R}^{K^2}$ is nonzero. So, it follows that $\xi_n(\Antr) = \theta_n(\Antr)$ where $\theta_n(\Antr)$ is defined in \eqref{eq:target_of_inference}.

By this additional condition given in Corollary~\ref{cor:exact_coverage_bernoulli_target}, it also follows that for large $n$, the inequality $\tau^{(s)}_{n, k \ell} \le \tilde{\tau}^{(s)}_{n, k \ell}$ from \eqref{eq:tau_tilde_tau_bound_proof} in Supplement~\ref{app:squiggle_proposition_proof} becomes an equality $\tau^{(s)}_{n, k \ell} = \tilde{\tau}^{(s)}_{n, k \ell}$ for all $(k, \ell)$ such that the corresponding entry of $u_n \in \mathbb{R}^{K^2}$ is nonzero. Consequently, for large $n$ we have $\omega^2_n = \sigma_n^2$ in \eqref{eq:sigma_sq_def_squiggle_proof}, and the inequality in \eqref{ineq:almost_done_squiggle_xi_thing} also becomes an equality.

Putting these facts together, for large $n$ we can replace $\xi_n(\Antr)$ with $\theta_n(\Antr)$, and due to the change from inequalities to equalities in \eqref{eq:tau_tilde_tau_bound_proof}, \eqref{eq:sigma_sq_def_squiggle_proof}, and \eqref{ineq:almost_done_squiggle_xi_thing}, the limit inferior in \eqref{eq:almost_done_bernoulli} and \eqref{eq:done_bernoulli} can be replaced with a limit, and so we have \\ $\lim_{n \to \infty} P \Big( \theta(\Antr) \in \hat{\xi}(\Ante, \Antr) \pm \phi_{1-\alpha/2} \cdot \hat{\sigma}_n \mid \Antr \Big) = 1-\alpha$.



\bibliographystyle{agsm}
\bibliography{refs}

\end{document}